\documentclass{llncs}
\usepackage[utf8]{inputenc}
\usepackage{microtype}

\usepackage{amsthm,amsfonts,longtable,float,etoolbox,tikz,tkz-graph,algorithm2e}
\usepackage[nosumlimits]{amsmath}
\usepackage{tabularx}
\usetikzlibrary{graphs,quotes,decorations.shapes,decorations.pathreplacing}
\usepackage{url}
\usepackage{comment, marvosym}
\usepackage{wasysym}
\usepackage[nooneline]{caption}%

\newcommand{\edgeScale}{0.75}

\newcommand{\myEdge}[2]{ \tikz[baseline=-1pt]{
\draw[#2,line width=0.3pt] (0,0) -- ++(0.6,0) node[anchor=base, yshift=3pt, pos=0.5] {\scalebox{\edgeScale}{$#1$}};
}}
\newcommand{\edge}[1]{\myEdge{#1}{->}}

\newcommand{\escale}[1]{\ensuremath{\textbf{\scalebox{0.7}{#1}}}}
\newcommand{\nscale}[1]{\ensuremath{\textbf{\scalebox{0.6}{#1}}}}

\newcommand{\R}{\ensuremath{\mathcal{R}}}

\newtheorem{objective}{Objective}

\allowdisplaybreaks 
\begin{document}
\title{Hardness Results for the Synthesis of $b$-bounded Petri Nets (Technical Report)}
\author{Ronny Tredup}
\institute{Universit\"at Rostock, Institut f\"ur Informatik, Theoretische Informatik, Albert-Einstein-Stra\ss e 22, 18059, Rostock }
\maketitle


\begin{abstract}
Synthesis for a type $\tau$ of Petri nets is the following search problem: 
For a transition system $A$, find a Petri net $N$ of type $\tau$ whose state graph is isomorphic to $A$, if there is one. 
To determine the computational complexity of synthesis for types of bounded Petri nets we investigate their corresponding decision version, called feasibility.
We show that feasibility is NP-complete for (pure) $b$-bounded P/T-nets if $b\in \mathbb{N}^+$.
We extend (pure) $b$-bounded P/T-nets by the additive group $\mathbb{Z}_{b+1}$ of integers modulo $(b+1)$ and show feasibility to be NP-complete for the resulting type.
To decide if $A$ has the \emph{event state separation property} is shown to be NP-complete for (pure) $b$-bounded and group extended (pure) $b$-bounded P/T-nets.
Deciding if $A$ has the \emph{state separation property} is proven to be NP-complete for (pure) $b$-bounded P/T-nets.
\end{abstract}

\section{Introduction}

\emph{Synthesis} for a Petri net type $\tau$ is the task to find, for a given transition system (TS, for short) $A$, a Petri net $N$ of this type such that its state graph is isomorphic to $A$ if such a net exists.
The decision version of synthesis is called $\tau$-\emph{feasibility}.
It asks whether for a given TS $A$ a Petri net $N$ of type $\tau$ exists whose state graph is isomorphic to $A$.

Synthesis for Petri nets has been investigated and applied for many years and in numerous fields:
It is used to extract concurrency and distributability data from sequential specifications like transition systems or languages \cite{DBLP:journals/fac/BadouelCD02}.
Synthesis has applications in the field of process discovery to reconstruct a model from its execution traces \cite{DBLP:books/daglib/0027363}.
In \cite{DBLP:journals/deds/HollowayKG97}, it is employed in supervisory control for discrete event systems and in \cite{DBLP:journals/tcad/CortadellaKKLY97} it is used for the synthesis of speed-independent circuits. 
This paper deals with the computational complexity of synthesis for types of \emph{T2019b} Petri nets, that is, Petri nets for which there is a positive integer $b$ restricting the number of tokens on every place in any reachable marking.

In \cite{DBLP:conf/tapsoft/BadouelBD95,DBLP:series/txtcs/BadouelBD15}, synthesis has been shown to be solvable in polynomial time for bounded and pure bounded P/T-nets.
The approach provided in \cite{DBLP:conf/tapsoft/BadouelBD95,DBLP:series/txtcs/BadouelBD15} guarantees a (pure) bounded P/T-net to be output if such a net exists.
Unfortunately, it does not work for preselected bounds.
In fact, in \cite{DBLP:journals/tcs/BadouelBD97} it has been shown that feasibility is NP-complete for $1$-bounded P/T-nets, that is, if the bound $b=1$ is chosen \emph{a priori}.
In \cite{DBLP:conf/apn/TredupRW18,DBLP:conf/concur/TredupR18}, it was proven that this remains true even for strongly restricted input TSs.
In contrast, \cite{DBLP:conf/stacs/Schmitt96} shows that it suffices to extend pure $1$-bounded P/T-nets by the additive group $\mathbb{Z}_2$ of integers modulo $2$ to bring the complexity of synthesis down to polynomial time.
The work of \cite{TR2019a} confirms also for other types of $1$-bounded Petri nets that the presence or absence of interactions between places and transitions tip the scales of synthesis complexity.
However, some questions in the area of synthesis for Petri nets are still open.
Recently, in \cite{DBLP:conf/concur/SchlachterW17} the complexity status of synthesis for (pure) $b$-bounded P/T-nets, $2\leq b$, has been reported as unknown. 
Furthermore, it has not yet been analyzed whether extending (pure) $b$-bounded P/T-nets by the group $\mathbb{Z}_{b+1}$ provides also a tractable superclass if $b\geq 2$.

In this paper, we show that feasibility for (pure) $b$-bounded P/T-nets, $b\in \mathbb{N}^+$, is NP-complete.
This makes their synthesis NP-hard.
Moreover, we introduce (pure) $\mathbb{Z}_{b+1}$-extended $b$-bounded P/T-nets, $b\geq 2$.
This type origins from (pure) $b$-bounded P/T-nets by adding interactions between places and transitions simulating addition of integers modulo $b+1$.
This extension is a natural generalization of Schmitt's approach \cite{DBLP:conf/stacs/Schmitt96}, which does this for $b=1$.
In contrast to the result of \cite{DBLP:conf/stacs/Schmitt96}, this paper shows that feasibility for (pure) $\mathbb{Z}_{b+1}$-extended $b$-bounded P/T-nets remains NP-complete if $b\geq 2$.

To prove the NP-completeness of feasibility we use its well known close connection to the so-called \emph{event state separation property} (ESSP) and \emph{state separation property} (SSP).
In fact, a TS $A$ is feasible with respect to a Petri net type if and only if it has the type related ESSP \emph{and} SSP \cite{DBLP:series/txtcs/BadouelBD15}.
The question of whether a TS $A$ has the ESSP or the SSP also defines decision problems.
The possibility to decide efficiently if $A$ has at least one of both properties serves as quick-fail pre-processing mechanisms for feasibility.
Moreover, if $A$ has the ESSP then synthesizing Petri nets up to language equivalence is possible \cite{DBLP:series/txtcs/BadouelBD15}.
This makes the decision problems ESSP and SSP worth to study. 
In \cite{DBLP:journals/tcs/Hiraishi94}, both problems have been shown to be NP-complete for pure $1$-bounded P/T-nets.
This has been confirmed for almost trivial inputs in \cite{DBLP:conf/apn/TredupRW18,DBLP:conf/concur/TredupR18}.

This paper shows feasibility, ESSP and SSP to be NP-complete for $b$-bounded P/T-nets, $b\in \mathbb{N}^+$.
Moreover, feasibility and ESSP are shown to remain NP-complete for (pure) $\mathbb{Z}_{b+1}$-extended $b$-bounded P/T-nets if $b\geq 2$.
Interestingly, \cite{T2019b} shows that SSP is decidable in polynomial time for (pure) $\mathbb{Z}_{b+1}$-extended $b$-bounded P/T-nets, $b\in \mathbb{N}^+$.
So far, this is the first net family where the provable computational complexity of SSP is different to feasibility and ESSP.

All presented NP-completeness proofs base on a reduction from the monotone one-in-three 3-SAT problem which is known to be NP-complete \cite{DBLP:journals/dcg/MooreR01}.
Every reduction starts from a given boolean input expression $\varphi$ and results in a TS $A_\varphi$.
The expression $\varphi$ belongs to monotone one-in-three 3-SAT if and only if $A_\varphi$ has the (target) property ESSP, SSP or feasibility, respectively.

This paper is organized as follows:
Section~\ref{sec:preliminaries} introduces the formal definitions and notions. 
Section~\ref{sec:unions} introduces the concept of unions applied in by our proofs.
Section~\ref{sec:hardness_results} provides the reductions and proves their functionality.
A short conclusion completes the paper.
This paper is an extended abstract of the technical report \cite{T2019b}. 
The proofs that had to be removed due to space limitation are given in \cite{T2019b}.

\section{Preliminaries}\label{sec:preliminaries}%

See Figure~\ref{fig:types} and Figure~\ref{fig:example_prelis} for an example of the notions defined in this section.
A \emph{transition system} (TS for short) $A = (S,E,\delta)$ consists of finite disjoint sets $S$ of states and $E$ of events and a partial \emph{transition function} $\delta: S\times E\rightarrow S$.
Usually, we think of $A$ as an edge-labeled directed graph with node set $S$ where every triple $\delta(s,e)=s'$ is interpreted as an $e$-labeled edge $s\edge{e}s'$, called \emph{transition}.
We say that an event $e$ \emph{occurs} at state $s$ if $\delta(s,e)=s'$ for some state $s'$ and abbreviate this with $s\edge{e}$.
This notation is extended to words $w'=wa$, $w\in E^*, a\in E$ by inductively defining $s\edge{\varepsilon}s$ for all $s\in S$ and $s\edge{w'}s''$ if and only if $s\edge{w}s'$ and $s'\edge{a}s''$. 
If $w\in E^*$ then $s\edge{w}$ denotes that there is a state $s'\in S$ such that $s\edge{w}s'$.
An \emph{initialized} TS $A=(S,E,\delta, s_0)$ is a TS with an initial state $s_0 \in S$ where every state is \emph{reachable}: $\forall s\in S, \exists w\in E^*: s_0\edge{w}s$.
The language of $A$ is the set $L(A)=\{w\in E^* \mid s_{0}\edge{w}\}$.
In the remainder of this paper, if not explicitly stated otherwise, we assume all TSs to be initialized and we refer to the components of an (initialized) TS $A$ consistently by $A=(S_A, E_A, \delta_A, s_{0,A})$.

The following notion of \emph{types of nets} has been developed in~\cite{DBLP:series/txtcs/BadouelBD15}.
It allows us to uniformly capture several Petri net types in one general scheme. 
Every introduced Petri net type can be seen as an instantiation of this general scheme.
A type of nets $\tau$ is a TS $\tau=(S_\tau, E_\tau,\delta_\tau)$ and a Petri net $N = (P, T, f, M_0)$ of type $\tau$, $\tau$-net for short, is given by finite and disjoint sets $P$ of places and $T$ of transitions, an initial marking $M_0: P\longrightarrow S_\tau$, and a flow function $f: P \times T \rightarrow E_\tau$. 
The meaning of a $\tau$-net is to realize a certain behavior by cascades of firing transitions. 
In particular, a transition $t \in T$ can fire in a marking $M: P \longrightarrow S_\tau$ and thereby produces  the marking $M': P \longrightarrow S_\tau$ if for all $p\in P$ the transition $M(p)\edge{f(p,t)}M'(p)$ exists in $\tau$. 
This is denoted by $M \edge{t} M'$. 
Again, this notation extends to sequences $\sigma \in T^*$.
Accordingly, $RS(N)=\{M \mid \exists \sigma\in T^*: M_0\edge{\sigma}M \}$ is the set of all reachable markings of $N$.
Given a $\tau$-net $N=(P, T, f,M_0)$, its behavior is captured by the TS $A_N=(RS(N), T,\delta, M_0)$, called the state graph of $N$, where for every reachable marking $M$ of $N$ and transition $t \in T$ with $M \edge{t} M'$ the transition function $\delta$ of $A_N$ is defined by $\delta(M,t) = M'$.

The following notion of $\tau$-regions allows us to define the type related ESSP and SSP.
If $\tau$ is a type of nets then a $\tau$-region of a TS $A$ is a pair of mappings $(sup, sig)$, where $sup: S_A \longrightarrow S_\tau$ and $sig: E_A\longrightarrow E_\tau$, such that, for each transition $s\edge{e}s'$ of $A$, we have that $sup(s)\edge{sig(e)}sup(s')$ is a transition of $\tau$.
Two distinct states $s,s'\in S_A$ define an \emph{SSP atom} $(s,s')$, which is said to be $\tau$-solvable if there is a $\tau$-region $(sup, sig)$ of $A$ such that $sup(s)\not=sup(s')$. 
An event $e\in E_A$ and a state $s\in S_A$ at which $e$ does not occur, that is $\neg s\edge{e}$, define an \emph{ESSP atom} $(e,s)$.
The atom is said to be $\tau$-solvable if there is a $\tau$-region $(sup, sig)$ of $A$ such that $\neg sup(s)\edge{sig(e)}$. 
A $\tau$-region solving an ESSP or a SSP atom $(x,y)$ is a \emph{witness} for the $\tau$-solvability of $(x,y)$.
A TS $A$ has the $\tau$-ESSP ($\tau$-SSP) if all its ESSP (SSP) atoms are $\tau$-solvable.
Naturally, $A$ is said to be $\tau$-feasible if it has the $\tau$-ESSP and the $\tau$-SSP.
The following fact is well known from~\cite[p.161]{DBLP:series/txtcs/BadouelBD15}: 
A set $\mathcal{R}$ of $\tau$-regions of $A$ contains a witness for all ESSP and SSP atoms if and only if the \emph{synthesized $\tau$-net} $N^{\mathcal{R}}_A=(\mathcal{R} , E_A, f, M_0)$ has a state graph that is isomorphic to $A$.
The flow function of $N^{\mathcal{R}}_A$ is defined by $f((sup, sig), e)= sig(e)$ and its initial marking is $M_0((sup,sig))=sup(s_{0,A})$ for all $(sup, sig) \in \mathcal{R}, e\in E_A$ .
The regions of $\R$ become places and the events of $E_A$ become transitions of $N^{\mathcal{R}}_A$.
Hence, for a $\tau$-feasible TS $A$ where $\mathcal{R}$ is known, we can synthesize a net $N$ with state graph isomorphic to $A$ by constructing $N^{\mathcal{R}}_A$.
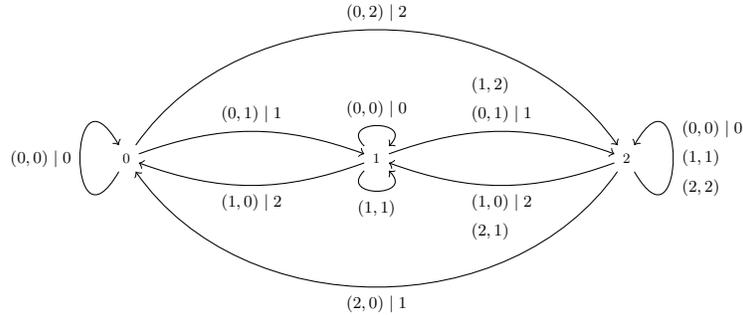
\begin{figure}[b!]
\centering
\begin{tikzpicture}[scale = 0.95]
\begin{scope}%

\node (0) at (0,0) {\nscale{$0$}};
\node (1) at (3.5,0) {\nscale{$1$}};
\node (2) at (7,0) {\nscale{$2$}};

\path (0) edge [<-, out=120,in=-120,looseness=10] node[left] {\escale{$(0,0) \mid 0$}} (0);
\path (1) edge [->, out=225,in=-45,looseness=4] node[below] { \escale{$(1,1)$}  } (1);
\path (1) edge [->, out=-225,in=45,looseness=4] node[above] {  \escale{$(0,0)  \mid  0$}  } (1);
\path (2) edge [->, out=-60,in=60,looseness=10] node[right, align= left] {\escale{$(0,0)  \mid  0$} \\ \escale{$(1,1)$} \\ \escale{$(2,2)$}  } (2);
\path (1) edge [->, bend left=20] node[above, align= left] {   \escale{$(1,2)$}  \\  \escale{$(0,1) \mid 1$}}   (2);
\path (1) edge [<-, bend right=20] node[below, align= left] {   \escale{$(1,0) \mid 2$} \\   \escale{$(2,1)$}   }   (2);
\graph{ 

(0) ->[bend left= 20, "\escale{$(0,1)\mid 1$}"] (1); 
(0) ->[bend left= 55, "\escale{$(0,2)\mid 2$}"] (2); 
(0) <-[bend right= 20,swap,  "\escale{$(1,0)\mid 2$}"] (1); 
(0) <-[bend right= 55, swap, "\escale{$(2,0) \mid 1$}"] (2); 
};
\end{scope}

\end{tikzpicture}
\caption{%
The types $\tau^2_0,\tau^2_1,\tau^2_2$ and $\tau^2_3$.
$\tau^2_0$ is sketched by the $(m,n)$-labeled transitions where edges with different labels represent different transitions.
Discard from $\tau^2_0$ the $(1,1)$, $(1,2)$, $(2,1)$ and $(2,2)$ labeled transitions to get $\tau^2_1$ and 
add for $i\in \{0,1,2\}$ the $i$-labeled transitions and remove $(0,0)$ to have $\tau^2_2$.
Discarding $(1,1),(1,2),(2,1),(2,2)$ leads from $\tau^2_2$ to $\tau^2_3$.
}
\label{fig:types}
\end{figure}

In this paper, we deal with the following $b$-bounded types of Petri nets:
\begin{enumerate}
\item
The type of \emph{$b$-bounded P/T-nets} is defined by $\tau^b_0=(\{0,\dots, b\}, \{0,\dots, b\}^2,\delta_{\tau^b_0})$ where for $s\in S_{\tau^b_0}$ and $(m,n)\in E_{\tau^b_0}$ the transition function is defined by $\delta_{\tau^b_0}(s,(m,n))=s-m+n$ if $s\geq m$ and $ s-m+n \leq b$, and undefined otherwise.

\item
The type of \emph{pure $b$-bounded P/T-nets} is a restriction of $\tau^b_0$-nets that discards all events $(m,n)$ from $E_{\tau^b_0}$ where both, $m$ and $n$, are positive. 
To be exact, $\tau^b_1=(\{0,\dots, b\}, E_{\tau^b_0} \setminus \{(m,n) \mid 1 \leq m,n \leq b\}, \delta_{\tau^b_1})$, and for $s\in S_{\tau^b_1}$ and $e\in E_{\tau^b_1}$ we have $\delta_{\tau^b_1}(s,e)=\delta_{\tau^b_0}(s,e)$.

\item
The type of \emph{$\mathbb{Z}_{b+1}$-extended $b$-bounded P/T-nets} origins from $\tau^b_0$ by extending the event set $E_{\tau^b_0}$ with the elements $0,\dots, b$. 
The transition function additionally simulates the addition modulo (b+1).
More exactly, this type is defined by $\tau^b_2=(\{0,\dots, b\}, (E_{\tau^b_0}\setminus \{(0,0)\}) \cup \{0,\dots, b\}, \delta_{\tau^b_2})$ where for $s\in S_{\tau^b_2}$ and $e\in E_{\tau^b_2}$ we have that $\delta_{\tau^b_2}(s,e)=\delta_{\tau^b_0}(s,e)$ if $e\in E_{\tau^b_0}$ and, otherwise, $\delta_{\tau^b_2}(s,e)=(s+e) \text{ mod } (b+1)$.

\item 
The type of \emph{$\mathbb{Z}_{b+1}$-extended pure $b$-bounded P/T-nets} is a restriction of $\tau^b_2$ being defined by $\tau^b_3=(\{0,\dots, b\}, E_{\tau^b_2}\setminus \{(m,n) \mid 1 \leq m,n \leq b\}, \delta_{\tau^b_3})$ where for $s\in S_{\tau^b_3}$ and $e\in E_{\tau^b_3}$ we have that $\delta_{\tau^b_3}(s,e)=\delta_{\tau^b_2}(s,e)$.
\end{enumerate}
Notice that the type $\tau^1_3$ coincides with Schmitt's type for which the considered decision problems and synthesis become tractable \cite{DBLP:conf/stacs/Schmitt96}.
Moreover, in \cite{TR2019a} it has been shown that $\tau^1_2$, a generalization of $\tau^1_3$, allows polynomial time synthesis, too.
Hence, in the following, if not explicitly stated otherwise, for $\tau\in \{\tau^b_0,\tau^b_1\}$ we let $b\in \mathbb{N}^+$ and for $\tau\in \{\tau^b_2,\tau^b_3\}$ we let $2\leq b\in \mathbb{N}$.
If $\tau\in \{\tau^b_0,\tau^b_1, \tau^b_2,\tau^b_3\}$ and if $(sup, sig)$ is a $\tau$-region of a TS $A$ then for $e\in E_A$ we define $sig^-(e)=m$ and $sig^+(e)=n$ and $\vert sig(e)\vert =0$ if $sig(e)=(m,n)\in E_\tau$, respectively $sig^-(e)=sig^+(e)=0$ and $\vert sig(e)\vert =sig(e)$ if $sig(e)\in \{0,\dots, b\}$.
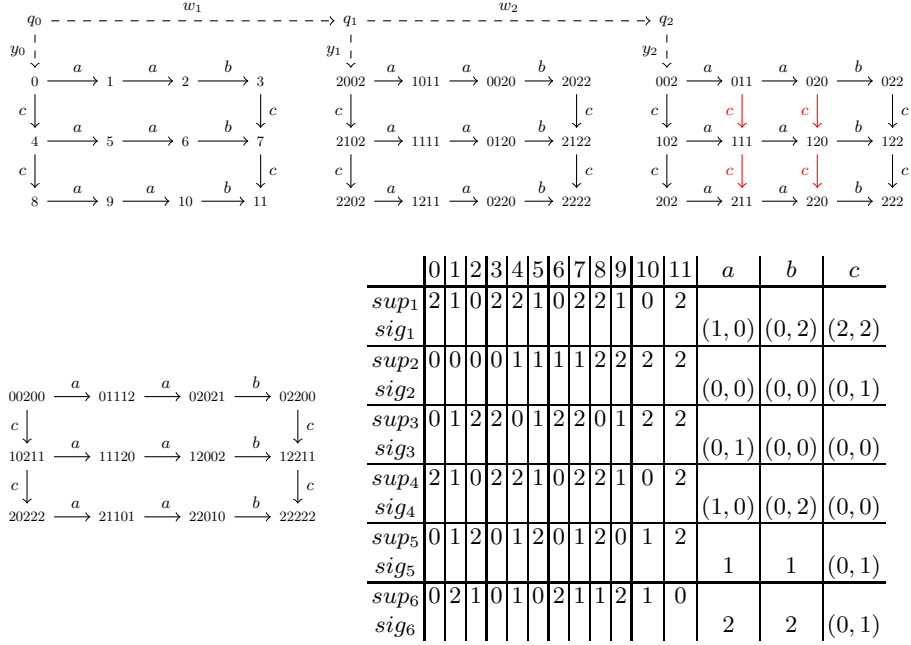
\begin{figure}[t!]
\centering
\begin{minipage}[t!]{1\textwidth}
\centering
\begin{tikzpicture}
\begin{scope}
\node (0) at (0,0) {\nscale{$0$}};
\node (1) at (1,0) {\nscale{$1$}};
\node (2) at (2,0) {\nscale{$2$}};
\node (3) at (3,0) {\nscale{$3$}};
\node (4) at (0,-0.8) {\nscale{$4$}};
\node (5) at (1,-0.8) {\nscale{$5$}};
\node (6) at (2,-0.8) {\nscale{$6$}};
\node (7) at (3,-0.8) {\nscale{$7$}};
\node (8) at (0,-1.6) {\nscale{$8$}};
\node (9) at (1,-1.6) {\nscale{$9$}};
\node (10) at (2,-1.6) {\nscale{$10$}};
\node (11) at (3,-1.6) {\nscale{$11$}};

\path (0) edge [->] node[pos=0.6,above] {\escale{$a$}} (1);
\path (1) edge [->] node[pos=0.6,above] {\escale{$a$}} (2);
\path (2) edge [->] node[pos=0.6,above] {\escale{$b$}} (3);
\path (4) edge [->] node[pos=0.6,above] {\escale{$a$}} (5);
\path (5) edge [->] node[pos=0.6,above] {\escale{$a$}} (6);
\path (6) edge [->] node[pos=0.6,above] {\escale{$b$}} (7);
\path (8) edge [->] node[pos=0.6,above] {\escale{$a$}} (9);
\path (9) edge [->] node[pos=0.6,above] {\escale{$a$}} (10);
\path (10) edge [->] node[pos=0.6,above] {\escale{$b$}} (11);
\path (0) edge [->] node[left] {\escale{$c$}} (4);
\path (4) edge [->] node[left] {\escale{$c$}} (8);
\path (3) edge [->] node[right] {\escale{$c$}} (7);
\path (7) edge [->] node[right] {\escale{$c$}} (11);

\end{scope}

\node (q0) at (0,0.8) {\scalebox{0.7}{$q_0$}};
\path (q0) edge [->, dashed] node[left] {\escale{$y_0$}} (0);

\begin{scope}[xshift=4.2cm]
\node (0) at (0,0) {\nscale{$2002$}};
\node (1) at (1,0) {\nscale{$1011$}};
\node (2) at (2,0) {\nscale{$0020$}};
\node (3) at (3,0) {\nscale{$2022$}};
\node (4) at (0,-0.8) {\nscale{$2102$}};
\node (5) at (1,-0.8) {\nscale{$1111$}};
\node (6) at (2,-0.8) {\nscale{$0120$}};
\node (7) at (3,-0.8) {\nscale{$2122$}};
\node (8) at (0,-1.6) {\nscale{$2202$}};
\node (9) at (1,-1.6) {\nscale{$1211$}};
\node (10) at (2,-1.6) {\nscale{$0220$}};
\node (11) at (3,-1.6) {\nscale{$2222$}};

\path (0) edge [->] node[pos=0.6,above] {\escale{$a$}} (1);
\path (1) edge [->] node[pos=0.6,above] {\escale{$a$}} (2);
\path (2) edge [->] node[pos=0.6,above] {\escale{$b$}} (3);
\path (4) edge [->] node[pos=0.6,above] {\escale{$a$}} (5);
\path (5) edge [->] node[pos=0.6,above] {\escale{$a$}} (6);
\path (6) edge [->] node[pos=0.6,above] {\escale{$b$}} (7);
\path (8) edge [->] node[pos=0.6,above] {\escale{$a$}} (9);
\path (9) edge [->] node[pos=0.6,above] {\escale{$a$}} (10);
\path (10) edge [->] node[pos=0.6,above] {\escale{$b$}} (11);
\path (0) edge [->] node[left] {\escale{$c$}} (4);
\path (4) edge [->] node[left] {\escale{$c$}} (8);
\path (3) edge [->] node[right] {\escale{$c$}} (7);
\path (7) edge [->] node[right] {\escale{$c$}} (11);
\end{scope}

\node (q1) at (4.2,0.8) {\scalebox{0.7}{$q_1$}};
\path (q1) edge [->, dashed] node[left] {\escale{$y_1$}} (0);

\begin{scope}[xshift=8.4cm]
\node (0) at (0,0) {\nscale{$002$}};
\node (1) at (1,0) {\nscale{$011$}};
\node (2) at (2,0) {\nscale{$020$}};
\node (3) at (3,0) {\nscale{$022$}};
\node (4) at (0,-0.8) {\nscale{$102$}};
\node (5) at (1,-0.8) {\nscale{$111$}};
\node (6) at (2,-0.8) {\nscale{$120$}};
\node (7) at (3,-0.8) {\nscale{$122$}};
\node (8) at (0,-1.6) {\nscale{$202$}};
\node (9) at (1,-1.6) {\nscale{$211$}};
\node (10) at (2,-1.6) {\nscale{$220$}};
\node (11) at (3,-1.6) {\nscale{$222$}};

\path (0) edge [->] node[pos=0.6,above] {\escale{$a$}} (1);
\path (1) edge [->] node[pos=0.6,above] {\escale{$a$}} (2);
\path (2) edge [->] node[pos=0.6,above] {\escale{$b$}} (3);
\path (1) edge [->, red!=20] node[left] {\escale{$c$}} (5);
\path (2) edge [->, red!=20] node[left] {\escale{$c$}} (6);
\path (5) edge [->, red!=20] node[left] {\escale{$c$}} (9);
\path (6) edge [->, red!=20] node[left] {\escale{$c$}} (10);
\path (4) edge [->] node[pos=0.6,above] {\escale{$a$}} (5);
\path (5) edge [->] node[pos=0.6,above] {\escale{$a$}} (6);
\path (6) edge [->] node[pos=0.6,above] {\escale{$b$}} (7);
\path (8) edge [->] node[pos=0.6,above] {\escale{$a$}} (9);
\path (9) edge [->] node[pos=0.6,above] {\escale{$a$}} (10);
\path (10) edge [->] node[pos=0.6,above] {\escale{$b$}} (11);
\path (0) edge [->] node[left] {\escale{$c$}} (4);
\path (4) edge [->] node[left] {\escale{$c$}} (8);
\path (3) edge [->] node[right] {\escale{$c$}} (7);
\path (7) edge [->] node[right] {\escale{$c$}} (11);
\end{scope}
\node (q2) at (8.4,0.8) {\scalebox{0.7}{$q_2$}};
\path (q2) edge [->, dashed] node[left] {\escale{$y_2$}} (0);
\path (q0) edge [->, dashed] node[above] {\escale{$w_1$}} (q1);
\path (q1) edge [->, dashed] node[above] {\escale{$w_2$}} (q2);

\end{tikzpicture}
\end{minipage}
\vspace{0.5cm}

\begin{minipage}{0.32\textwidth}
\centering
\begin{tikzpicture}
\node (0) at (0,0) {\nscale{$00200$}};
\node (1) at (1.2,0) {\nscale{$01112$}};
\node (2) at (2.4,0) {\nscale{$02021$}};
\node (3) at (3.6,0) {\nscale{$02200$}};
\node (4) at (0,-0.8) {\nscale{$10211$}};
\node (5) at (1.2,-0.8) {\nscale{$11120$}};
\node (6) at (2.4,-0.8) {\nscale{$12002$}};
\node (7) at (3.6,-0.8) {\nscale{$12211$}};
\node (8) at (0,-1.6) {\nscale{$20222$}};
\node (9) at (1.2,-1.6) {\nscale{$21101$}};
\node (10) at (2.4,-1.6) {\nscale{$22010$}};
\node (11) at (3.6,-1.6) {\nscale{$22222$}};
\path (0) edge [->] node[pos=0.6,above] {\escale{$a$}} (1);
\path (1) edge [->] node[pos=0.6,above] {\escale{$a$}} (2);
\path (2) edge [->] node[pos=0.6,above] {\escale{$b$}} (3);
\path (4) edge [->] node[pos=0.6,above] {\escale{$a$}} (5);
\path (5) edge [->] node[pos=0.6,above] {\escale{$a$}} (6);
\path (6) edge [->] node[pos=0.6,above] {\escale{$b$}} (7);
\path (8) edge [->] node[pos=0.6,above] {\escale{$a$}} (9);
\path (9) edge [->] node[pos=0.6,above] {\escale{$a$}} (10);
\path (10) edge [->] node[pos=0.6,above] {\escale{$b$}} (11);
\path (0) edge [->] node[left] {\escale{$c$}} (4);
\path (4) edge [->] node[left] {\escale{$c$}} (8);
\path (3) edge [->] node[right] {\escale{$c$}} (7);
\path (7) edge [->] node[right] {\escale{$c$}} (11);
\end{tikzpicture}

\end{minipage}\hfill
\hspace{0.5cm}
\begin{minipage}[b]{0.6\textwidth}
\begin{tabular}{c|c|c|c|c|c|c|c|c|c|c|c|c|c|c|c}
& $0$ & $1$& $2$ & $3$ & $4$ & $5$ & $6$ & $7$ & $8$ & $9$&   $10$ & $11$ & $a$ & $b$ & $c$ \\ \hline
$sup_1$ & $2$ & $1$ & $0$ & $2$ & $2$ & $1$ & $0$ & $2$ & $2$ & $1$&   $0$ & $2$ & $ $ & $ $ & $ $ \\
$sig_1$ & $ $ & $ $ & $ $ & $ $ & $ $ & $ $ & $ $ & $ $ & $ $ & $ $&  & $ $   & $(1,0)$ & $(0,2)$ & $(2,2)$ \\ \hline 
$sup_2$ & $0$ & $0$ & $0$ & $0$ & $1$ & $1$ & $1$ & $1$ & $2$ & $2$&   $2$ & $2$ & $ $ & $ $ & $ $ \\
$sig_2$ & $ $ & $ $ & $ $ & $ $ & $ $ & $ $ & $ $ & $ $ & $ $ & $ $&  & $ $   & $(0,0)$ & $(0,0)$ & $(0,1)$ \\ \hline 
$sup_3$ & $0$ & $1$ & $2$ & $2$ & $0$ & $1$ & $2$ & $2$ & $0$ & $1$&   $2$ & $2$ & $ $ & $ $ & $ $ \\
$sig_3$ & $ $ & $ $ & $ $ & $ $ & $ $ & $ $ & $ $ & $ $ & $ $ & $ $&  & $ $   & $(0,1)$ & $(0,0)$ & $(0,0)$ \\ \hline 
$sup_4$ & $2$ & $1$ & $0$ & $2$ & $2$ & $1$ & $0$ & $2$ & $2$ & $1$&   $0$ & $2$ & $ $ & $ $ & $ $ \\
$sig_4$ & $ $ & $ $ & $ $ & $ $ & $ $ & $ $ & $ $ & $ $ & $ $ & $ $&  & $ $   & $(1,0)$ & $(0,2)$ & $(0,0)$ \\ \hline 
$sup_5$ & $0$ & $1$ & $2$ & $0$ & $1$ & $2$ & $0$ & $1$ & $2$ & $0$&   $1$ & $2$ & $ $ & $ $ & $ $ \\
$sig_5$ & $ $ & $ $ & $ $ & $ $ & $ $ & $ $ & $ $ & $ $ & $ $ & $ $&  & $ $   & $1$ & $1$ & $(0,1)$ \\ \hline 
$sup_6$ & $0$ & $2$ & $1$ & $0$ & $1$ & $0$ & $2$ & $1$ & $1$ & $2$&   $1$ & $0$ & $ $ & $ $ & $ $ \\
$sig_6$ & $ $ & $ $ & $ $ & $ $ & $ $ & $ $ & $ $ & $ $ & $ $ & $ $&  & $ $   & $2$ & $2$ & $(0,1)$ 
\end{tabular}
\end{minipage}
\caption{%
Upper left, bold lines: Input TS $A$.
Bottom right: 
$\tau^2_0$ regions $\mathcal{R}_1=\{R_1,R_2,R_3,R_4\}$ of $A$, where $R_i=(sup_i, sig_i)$ for $i\in \{1,\dots,4\}$.
At the same time $\tau^2_1$-regions $\mathcal{R}_2=\{R_2,R_3,R_4\}$ and $\tau^2_3$-regions $\mathcal{R}_3=\{R_2,R_3,R_4, R_5, R_6\}$ where for $\tau^2_3$ we simply replace every signature $(0,0)$ by $0$.
$\mathcal{R}_1$ $\tau^2_0$-solves all ESSP and SSP atoms.
While $\mathcal{R}_2$ does solve all SSP atoms it fails to solve the ESSP atoms $\{(c,1),(c,2),(c,5),(c,6)\}$.
By $\tau^2_1$'s lack of event $(2,2)$, these atoms are not $\tau^2_1$-solvable at all.
Obviously, $\tau^2_3$ compensates this deficiency by using the $\mathbb{Z}_3$-events $1,2$: $\mathcal{R}_3$ solves all atoms of $A$.
Upper middle, bold lines: The state graph $A_{N^{\mathcal{R}_1}_A}$ of the synthesized $\tau^2_0$-net $N^{\mathcal{R}_1}_A=(\mathcal{R}_1, \{a,b,c\}, f_1, 2002)$ with flow function $f_1(R_i,e)=sig_i(e)$ for $i\in \{1,\dots, 4\}$ and initial marking $M^1_0(R_1)=2, M^1_0(R_2)=0,M^1_0(R_3)=0$ and $M^1_0(R_4)=2$.
Every marking of $N^{\mathcal{R}_1}_A$'s places $R_1,R_2,R_3,R_4$ is denoted as a 4-tuple.
As $\mathcal{R}_1$ proves $A$'s $\tau^2_0$-feasibility, $N^{\mathcal{R}_1}_A$'s state graph is isomorphic to $A$.
Upper right, bold lines: The state graph $A_{N^{\mathcal{R}_2}_A}$ of the synthesized $\tau^2_1$-net $N^{\mathcal{R}_2}_A=(\mathcal{R}_2, \{a,b,c\}, f_2, 002)$ with flow function $f_1(R_i,e)=sig_i(e)$ for $i\in \{1,\dots, 4\}$ and initial marking $002$.
$A$ has no $\tau^2_1$-ESSP and, hence, $A_{N^{\mathcal{R}_2}_A}$ is not isomorphic to $A$.
Bottom left: The state graph $A_{N^{\mathcal{R}_3}_A}$ of the synthesized $\tau^2_3$-net $N^{\mathcal{R}_3}_A=(\mathcal{R}_3, \{a,b,c\}, f_3, 00200)$ with flow function $f_3(R_i,e)=sig_i(e)$ for $i\in \{2,\dots, 6\}$ and initial marking $00200$.
Again, $N^{\mathcal{R}_3}_A$'s state graph is isomorphic to $A$.
Top, bold and dashed lines:
The joining TS $A(U)=(S_U\cup \{q_0,q_1,q_2\}, E_U\cup\{y_0,y_1,y_2,w_0,w_1\}, \delta_{A(U)}, q_0)$ of the union $U=(A, A_{N^{\mathcal{R}_1}_A}, A_{N^{\mathcal{R}_2}_A})$. 
}\label{fig:example_prelis}
\end{figure}
The observations of the next Lemma are used to simplify our proofs:
\begin{lemma}\label{lem:observations} Let $\tau \in \{\tau^b_0, \tau^b_1, \tau^b_2, \tau^b_3\}$ and $A$ be a TS.
\begin{enumerate}
\item\label{lem:sig_summation_along_paths}
Two mappings $sup: S_A\longrightarrow S_\tau$ and $sig: E_A\longrightarrow E_\tau$ define a $\tau$-region of $A$ if and only if for every word $w=e_1\dots e_\ell \in E_\tau^*$ and state $s_0\in S_A$ the following statement is true:
If $s_{0}\edge{e_1} \dots \edge{e_\ell} s_\ell$, then $sup(s_{i})=sup(s_{i-1})-sig^-(e_i)+sig^+(e_i)+\vert sig(e)\vert$ for $i\in \{1,\dots, \ell\}$, where for $ \tau\in \{\tau^b_2, \tau^b_3\}$ this equation is considered modulo $(b+1)$.
That is, every region $(sup, sig)$ is implicitly completely defined by the signature $sig$ and the support of the initial state: $sup(s_{0,A})$.
\item\label{lem:absolute_value}
If $s_{0}, s_{1},\dots, s_{b}\in S_A$, $e\in E_A$ and $s_{0}\edge{e} \dots \edge{e} s_b$ then a $\tau$-region $(sup, sig)$ of $A$ satisfies $sig(e)= (m,n)$ with $m\not=n$ if and only if $(m,n) \in \{(1,0),(0,1)\}$.
If $sig(e)=(0,1)$ then $sup(s_{0})=0$ and $sup(s_b)=b$.
If $sig(e)=(1,0)$ then $sup(s_0)=b$ and $sup(s_b)=0$.
\end{enumerate}
\end{lemma}

\section{The Concept of Unions}\label{sec:unions}%

For our reductions, we use the technique of \emph{component design} \cite{DBLP:books/fm/GareyJ79}.
Every implemented constituent is a TS locally ensuring the satisfaction of some constraints. 
Commonly, all constituents are finally joined together in a target instance (TS) such that all required constraints are properly globally translated.
However, the concept of unions saves us the need to actually create the target instance:

If $A_0, \dots, A_n$ are TSs with pairwise disjoint states (but not necessarily disjoint events) then $U(A_0, \dots, A_n)$ is their \emph{union} with set of states $S_U=\bigcup_{i=0}^n S_{A_i}$ and set of events $E_U=\bigcup_{i=0}^n E_{A_i}$.
For a flexible formalism, we allow to build unions recursively:
Firstly, we identify every TS $A$ with the union containing only $A$, that is, $A = U(A)$.   
Next, if $U_1= U(A^1_0,\dots,A^1_{n_1}), \dots, U_m=(A^m_0,\dots,A^n_{n_m})$ are unions then $U(U_1, \dots, U_m)$ is the evolved union $U(A^1_0, \dots, A^1_{n_1},\dots, A^m_0, \dots, A^n_{n_m})$.

The concepts of regions, SSP, and ESSP are transferred to unions $U = U(A_0, \dots, A_n)$ as follows:
A $\tau$-region $(sup, sig)$ of $U$ consists of $sup: S_U \rightarrow S_\tau$ and $sig: E_U \rightarrow E_\tau$ such that, for all $i \in \{0, \dots, n\}$, the projection $sup_i(s) = sup(s), s \in S_{A_i}$ and $sig_i(e) = sig(e), e \in E_{A_i}$ defines a region $(sup_i, sig_i)$ of $A_i$.
Then, $U$ has the $\tau$-SSP if for all distinct states $s, s' \in S_U$ of the \emph{same} TS $A_i$ there is a $\tau$-region $(sup,sig)$ of $U$ with $sup(s) \not= sup(s')$.
Moreover, $U$ has the $\tau$-ESSP if for all events $e \in E_U$ and all states $s \in S_U$ with $\neg s\edge{e}$ there is a $\tau$-region $(sup,sig)$ of $U$ where $sup(s) \edge{sig(e)}$ does not hold.
We say $U$ is $\tau$-feasible if it has the $\tau$-SSP and the $\tau$-ESSP.
In the same way, $\tau$-SSP and $\tau$-ESSP are translated to the state and event sets $S_U$ and $E_U$.

To merge a union $U = U(A_0, \dots, A_n)$ into a single TS, we define the joining $A(U)$ as the TS $A(U) = (S_U \cup Q, E_U \cup W \cup Y , \delta, q_0 )$ with additional connector states $Q=\{ q_0, \dots, q_n\}$ and fresh events $W=\{w_1, \dots, w_n\}$, $Y=\{ y_0, \dots, y_n\}$ connecting the individual TSs of $U$ by
\[\delta(s,e) = 
\begin{cases}
s_{0,A_i}, & \text{if } s = q_i \text{ and } e=y_i \text{ and } 0 \leq i \leq n,\\
q_{i+1}, & \text{if } s = q_i \text{ and } e=w_{i+1} \text{ and } 0 \leq i \leq n-1,\\
\delta_i(s,e), & \text{if } s \in S_{A_i} \text{ and } e \in E_{A_i} \text{ and } 0 \leq i \leq n
\end{cases}
\]
Hence, $A(U)$ puts the connector states into a chain with the events from $W$ and links the initial states of TSs from $U$ to this chain using events from $Y$.
For example, the upper part of Figure~\ref{fig:example_prelis} shows $A(U)$ where $U=(A, A_{N^{\mathcal{R}_1}_A},A_{N^{\mathcal{R}_2}_A})$.

In \cite{DBLP:conf/apn/TredupRW18,DBLP:conf/concur/TredupR18}, we have shown that a union $U$ is a useful vehicle to investigate if $A(U)$ has the $\tau$-feasibility, $\tau$-ESSP and $\tau$-SSP if $\tau=\tau^1_1$.
The following lemma generalizes this observation for $\tau\in \{ \tau^b_0, \tau^b_1, \tau^b_2, \tau^b_3\}$: 
\begin{lemma}\label{lem:union_validity}
Let $\tau \in \{\tau^b_0, \tau^b_1, \tau^b_2, \tau^b_3\}$.
If $U = U(A_0, \dots, A_n)$ of TSs $A_0, \dots, A_n$ is a union such that for every event $e\in E_U$ there is a state $s\in S_U$ with $\neg s\edge{e}$ then $U$ has the $\tau$-ESSP, respectively the $\tau$-SSP, if and only if $A(U)$ has the $\tau$-ESSP, respectively the $\tau$-SSP.
\end{lemma}

\section{Main Result}\label{sec:hardness_results}%

\begin{theorem}\label{the:hardness_results}
\begin{enumerate}
\item\label{the:hardness_results_essp}
If $\tau \in \{\tau^b_0, \tau^b_1,\tau^b_2, \tau^b_3\}$ then to decide if a TS $A$ is $\tau$-feasible or has the $\tau$-ESSP is NP-complete.
\item\label{the:hardness_results_ssp}
If $\tau \in \{\tau^b_0, \tau^b_1\}$ then deciding whether a TS $A$ has the $\tau$-SSP is NP complete. 
\end{enumerate}
\end{theorem}

The proof of Theorem~\ref{the:hardness_results} bases on polynomial time reductions of the cubic monotone one-in-three $3$-SAT problem to $\tau$-ESSP, $\tau$-feasibility and $\tau$-SSP, respectively. 
The input for this decision problem is a boolean expression $\varphi=\{C_0, \dots, C_{m-1}\}$ with $3$-clauses $C_i = \{X_{i,0}, X_{i,1}, X_{i,2}\}$ containing unnegated boolean variables $X_{i,0}, X_{i,1}, X_{i,2}$. 
$V(\varphi)$ denotes the set of all variables of $\varphi$.
Every element $X\in V(\varphi)$ occurs in exactly three clauses implying that $V(\varphi)=\{X_0,\dots, X_{m-1}\}$.
Given $\varphi$, cubic monotone one-in-three $3$-SAT asks if there is a one-in-three model $M$ of $\varphi$.
$M$ is a subset of $V(\varphi)$ such that $\vert M \cap C_i \vert =1$ for all $i\in \{0,\dots,m-1\}$.

For Theorem~\ref{the:hardness_results}.\ref{the:hardness_results_essp}, we let $\tau\in \{\tau^b_0,\tau^b_1,\tau^b_2,\tau^b_3\}$ and reduce $\varphi$ to a union $U_\tau=(K_\tau, T_\tau)$ which consists of the \emph{key} $K_{\tau}$ and the \emph{translator} $T_\tau$, both unions of TSs.
The index $\tau$ emphasizes that the components actual peculiarity depends on $\tau$. 

For Theorem~\ref{the:hardness_results}.\ref{the:hardness_results_ssp} the reduction starts from $\varphi$ and results in a union  $W=(K ,T)$ consisting of \emph{key} $K$ and \emph{translator} $T$, both unions.
$W$ needs no index as it has the same shape for $\tau^b_0$ and $\tau^b_1$.

The key $K_\tau$ provides a key ESSP atom $\alpha_\tau=(k,s_\tau)$ with event $k$ and state $s_\tau$.
The key $K$ supplies a key SSP atom $\alpha=(s, s')$ with states $s,s'$.
The translators $T_\tau$ and $T$ represent $\varphi$ by using the variables of $\varphi$ as events.
The unions $K_\tau$ and $T_\tau$ as well as $W$ and $T$ share events which define their \emph{interface} $I_\tau=E_{K_\tau} \cap E_{T_\tau}$ and $I=E_{K} \cap E_{T}$.
The construction ensures via the interface that $K_\tau$ and $T_\tau$ just as $K$ and $T$ interact in way that satisfies the following objectives of \emph{completeness}, \emph{existence} and \emph{sufficiency}:

\begin{objective}[Completeness]\label{int:completeness}
Let $(sup, sig)$ be a region of $K_\tau$ ($K$) solving the key atom.
If $(sup', sig')$ is a region of $T_\tau$ ($T$) satisfying $sup'(e)=sup(e)$ for $e\in I_\tau$ ($e\in I$) then the signature of the variable events reveal a one-in-three model of $\varphi$.
\end{objective}

\begin{objective}[Existence]\label{int:existence}
There is a region $(sup_K, sig_K)$ of $K_\tau$ ($K$) which solves the key atom.
If $\varphi$ is one-in-three satisfiable then there is a region $(sup_T, sig_T)$ of  $T_\tau$ ($T$) such that $sig_T(e)=sig_K(e)$ for $e\in I_\tau$ ($e\in I$) 
\end{objective}

\begin{objective}[Suffiency]\label{int:suffiency}
If the key atom is $\tau$-solvable in $U_\tau$, respectively $W$, then $U_\tau$ has the $\tau$-ESSP and the $\tau$-SSP and $W$ has the $\tau$-SSP.
 \end{objective}

Objective~\ref{int:completeness} ensures that the $\tau$-ESSP just as the $\tau$-feasibility of $U_\tau$ implies the one-in-three satisfiability of $\varphi$, respectively.
More exactly, if $U_\tau$ has the $\tau$-ESSP or the $\tau$-feasibility then there is a $\tau$-region $(sup, sig)$ of $U_\tau$ that solves $\alpha_\tau$.
By definition, this yields corresponding regions $(sup_K , sig_K)$ of $K_\tau$ and $(sup_T, sig_T)$ of $T_\tau$:
$sup_{K_\tau}(s)=sig(s)$ and $sig_{K_\tau}(e)=sig(e)$ if $s\in S_{K_\tau}, e\in E_{K_\tau}$ and $sup_{T_\tau}(s)=sig(s)$ and $sig_{T_\tau}(e)=sig(e)$ if $s\in S_{T_\tau}, e\in E_{T_\tau}$.
Similarly, the $\tau$-SSP of $W$ implies proper regions of $K$ and $T$ by a region $(sup, sig)$ of $W$ which solves $\alpha$.
As $(sup, sig)$ solves $\alpha_\tau$ in $U_\tau$ ($\alpha$ in $W$) the region $(sup_K, sig_K)$ solves $\alpha_\tau$ in $K_\tau$ ($\alpha$ in $K$).
Hence, by Objective~\ref{int:completeness}, the region $(sup_T, sig_T)$ of $T_\tau$ ($T$) reveals a one-in-three model of $\varphi$.

Reversely, Objective~\ref{int:existence} ensures that a one-in-three model of $\varphi$ defines a region $(sup, sig)$ of $U_\tau=(K_\tau, T_\tau)$ solving the key atom $\alpha_\tau$:  
$sup(s)=sup_K(s)$ if $s\in S_{K_\tau}$ and $sup(s)=sup_T(s)$ if $s\in S_{T_\tau}$ as well as $sig(e)=sig_K(e)$ if  $e\in E_{K_\tau}$ and $sig(e)=sig_T(e)$ if $e\in E_{T_\tau}\setminus E_{K_\tau}$.
Similarly, we get a region of $W$ that solves $\alpha$.

Objective~\ref{int:suffiency} guarantees that the solvability of the key atom $\alpha_\tau$ in $U_\tau$ ($\alpha$ in $K$) implies the solvability of all ESSP atoms and SSP atoms of $U_\tau$ (SSP atoms of $W$). 
Hence, by objective~\ref{int:existence}, if $\varphi$ has a one-in-three model then $U_\tau$ has the $\tau$-ESSP and is $\tau$-feasible just as $W$ has the $\tau$-SSP.

The unions $U_\tau$ and $W$ satisfy the conditions of Lemma~\ref{lem:union_validity}.
Therefore, the joining TS $A(U_\tau)$ has the $\tau$-ESSP and is $\tau$-feasible if and only if $\varphi$ is one-in-three satisfiable.
Likewise, the TS $A(W)$ has the $\tau$-SSP if and only if there is a one-in-three model for $\varphi$.
By definition, every TS $A$ has at most $\vert S_A\vert ^2$ SSP, respectively $\vert S_A\vert \cdot \vert E_A\vert$ ESSP atoms.
Consequently, a non-deterministic Turing machine can verify a guessed proof of $\tau$-SSP, $\tau$-ESSP and $\tau$-feasibility in polynomial time in the size of $A$.
Hence, all decision problems are in NP.
All reductions are doable in polynomial time and deciding the one-in-three satisfiability of $\varphi$ is NP-complete.
Thus, our approach proves Theorem~\ref{the:hardness_results}.

In order to prove the functionality of the constituents and to convey the corresponding intuition without becoming too technical, we proceed as follows.
On the one hand, we precisely define the constituents of the unions for arbitrary bound $b$ and input instance $\varphi=\{C_0,\dots, C_{m-1}\}$, $C_i=\{X_{i,0}, X_{i,1}, X_{i,2}\}$, $i\in \{0,\dots, m-1\}$, $V(\varphi)=\{X_0,\dots, X_{m-1}\}$, and prove their functionality.
On the other hand, we provide for comprehensibility full examples for the types $\tau\in \{\tau^b_0,\tau^b_1\}$ and the unions  $U_\tau$ and $W$.
The illustrations also provide a $\tau$-region solving the corresponding key atom.
For a running example, the input instance is $\varphi_0=\{C_0,\dots, C_{5}\}$ with clauses $C_0=\{X_0,X_1,X_2\},\ C_1= \{X_2,X_0,X_3\},\ C_2= \{X_1,X_3,X_0\},\ C_3= \{X_2,X_4,X_5\},\ C_4=\{X_1,X_5,X_4\},\ C_5= \{X_4,X_3,X_5\}$ that allows the one-in-three model $\{X_0,X_4\}$.
A full example for $\tau\in \{\tau^b_2,\tau^b_3\}$ is given in \cite{T2019b}.
For further simplification, we reuse gadgets for several unions as far as possible.
This is not always possible as small differences between two types of nets imply huge differences in the possibilities to build corresponding regions:
The more complex (the transition function of) the considered types, the more difficult the task to connect the solvability of the key atom with the signature of the interface events, respectively to connect the signature of the interface events with an implied model.
Moreover, the more difficult these tasks, the more complex the corresponding gadgets.
Hence, less complex gadgets are useless for more complex types.
Reversely, the more complex the gadgets the more possibilities to solve all ESSP atoms and all SSP atoms are needed.
Hence, more complex gadgets are not useful for less complex types.
At the end, some constituents may differ only slightly at first glance but their differences have a crucial and necessary impact. 

Note, that some techniques of the proof of Theorem~\ref{the:hardness_results} are very general advancements of our previous work \cite{DBLP:conf/apn/TredupRW18,DBLP:conf/concur/TredupR18}.
For example, like in \cite{DBLP:conf/apn/TredupRW18,DBLP:conf/concur/TredupR18} the proof of Theorem~\ref{the:hardness_results} bases on reducing cubic monotone one-in-three $3$-SAT.
Moreover, we apply unions as part of \emph{component design} \cite{DBLP:books/fm/GareyJ79}.
However, the reductions in \cite{DBLP:conf/apn/TredupRW18,DBLP:conf/concur/TredupR18} fit only for the basic type $\tau^1_1$ and they are already useless for $\tau^1_0$.
They fit even less for $\tau^b_0$ and $\tau^b_1$ if $b\geq 2$ and certainly not for their group extensions. 

We proceed as follows. 
Section~\ref{sec:keys_1} and Section~\ref{sec:translators_1} introduce the keys $K_{\tau^b_0}, K_{\tau^b_1}, K$ and translators $T_{\tau^b_0}, T_{\tau^b_1} , T$ and prove their functionality.
Section~\ref{sec:group_extensions_keys} and Section~\ref{sec:group_extensions_translators} present $K_{\tau^b_2}, K_{\tau^b_3}$ and $T_{\tau^b_2}, T_{\tau^b_3}$ and carry out how they work.
Section~\ref{sec:liaison} proves that the keys and translators collaborate properly.

\subsection{The Unions $K_{\tau^b_0}$ and $K_{\tau^b_1}$ and $K$.}\label{sec:keys_1}%

Let $\tau\in \{\tau^b_0, \tau^b_1\}$. 
The aim of $K_\tau$ and $K$ is summarized by the next lemma:
\begin{lemma}\label{lem:key_unions_1}
The keys $K_{\tau}$ and $K$ implement the interface events $k_{0},\dots, k_{6m-1}$ and provide a key atom $a_{\tau}$ and $\alpha$, respectively, such that the following is true:

\begin{enumerate}
\item\label{lem:key_unions_1_completeness}\emph{(Completeness)}
If $(sup_K, sig_K)$ is $\tau$-region of $K_\tau$, respectively of $K$, that solves $\alpha_\tau$, respectively $\alpha$, then $sig_K(k_0)=\dots=sig_K(k_{6m-1})=(0,b)$ or $sig_K(k_0)=\dots=sig_K(k_{6m-1})=(b,0)$.
\item\label{lem:key_unions_1_existence}\emph{(Existence)}
There is a $\tau$-region $(sup_K, sig_K)$ of $K_\tau$, respectively of $K$, that solves $a_{\tau}$, respectively $\alpha$, such that $sig_K(k_0)=\dots=sig_K(k_{6m-1})=(0,b)$.
\end{enumerate}
\end{lemma}

Firstly, we introduce the keys $K_{\tau^b_0}, K_{\tau^b_1}$ and $K$ and show that they satisfy Lemma~\ref{lem:key_unions_1}.\ref{lem:key_unions_1_completeness}.
Secondly, we present corresponding $\tau$-regions which prove Lemma~\ref{lem:key_unions_1}.\ref{lem:key_unions_1_existence}.

\textbf{The union $\mathbf{K_{\tau^b_0}}$} contains the following TS $H_0$ which provides the ESSP atom $(k, h_{0, 4b+1})$:

\begin{tikzpicture}
\node (init) at (-0.75,0) {$H_{0}=$};
\node (h0) at (0,0) {\nscale{$h_{0,0}$}};
\node (h1) at (1,0) {};
\node (dots1) at (1.25,0) {\nscale{$\dots$}};
\node (h_b_1) at (1.5,0) {};
\node (h_b) at (2.5,0) {\nscale{$h_{0,b}$}};
\node (h_b+1) at (3.5,0) {};
\node (dots_2) at (3.75,0) {\nscale{$\dots$}};
\node (h_2b_1) at (4,0) {};
\node (h_2b) at (5,0) {\nscale{$h_{0,2b}$}};
\node (h_2b+1) at (6.5,0) {\nscale{$h_{0,2b+1}$}};
\node (h_2b+2) at (7.5,0) {};
\node (dots_3) at (7.75,0) {\nscale{$\dots$}};
\node (h_3b) at (8,0) {};
\node (h_3b+1) at (9,0) { \nscale{$h_{0,3b+1}$} };
\node (h_3b+2) at (9,-1) {};
\node (dots_4) at (8.75,-1) {\nscale{$\dots$}};
\node (h_4b) at (8.5,-1) { };
\node (h_4b+1) at (7.5,-1) { \nscale{$h_{0,4b+1}$} };
\node (h_4b+2) at (6.5,-1) {};
\node (dots_5) at (6.25,-1) {\nscale{$\dots$}};
\node (h_5b) at (6,-1) { };
\node (h_5b+1) at (5,-1) { \nscale{$h_{0,5b+1}$} };
\node (h_5b+2) at (4,-1) { };
\node (dots_5) at (3.75,-1) {\nscale{$\dots$}};
\node (h_6b) at (3.5,-1) { };
\node (h_6b+1) at (2.5,-1) { \nscale{$h_{0,6b+1}$} };
\graph { 
(h0) ->["\escale{$k$}"] (h1); 
(h_b_1)->["\escale{$k$}"] (h_b) ->["\escale{$z$}"] (h_b+1);
(h_2b_1)->["\escale{$z$}"] (h_2b)->["\escale{$o_0$}"] (h_2b+1)->["\escale{$k$}"] (h_2b+2);
(h_3b)->["\escale{$k$}"] (h_3b+1)->["\escale{$z$}"] (h_3b+2);
(h_4b)->[swap, "\escale{$z$}"] (h_4b+1)->[swap, "\escale{$o_1$}"] (h_4b+2);
(h_5b)->[swap, "\escale{$o_1$}"] (h_5b+1)->[swap, "\escale{$k$}"] (h_5b+2);
(h_6b)->[swap, "\escale{$k$}"] (h_6b+1);
};
\end{tikzpicture}

\noindent
$K_{\tau^b_0}$ also installs for $j\in \{0,\dots, 6m-1\}$ the TS $D_{j,0}$ providing interface event $k_{j}$:

\noindent
\begin{tikzpicture}[yshift=-5cm]
\node (init) at (-1,0) {$D_{0,j}=$};
\foreach \i in {0,...,3} {\coordinate (\i) at (\i*1.2,0);}
\foreach \i in {0,...,3} {\node (p\i) at (\i) {\nscale{$d_{j,0,\i}$}};}
\node (hdots_3) at (4.2,0) {\nscale{$\dots$}};
\node (db+2) at (4.4,0) {};
\node (db+3) at (5.5,0) { \nscale{$d_{j,0,b+3}$} };
\graph { (p0) ->["\escale{$o_{0}$}"] (p1) ->["\escale{$k_j$}"] (p2) ->["\escale{$o_{1}$}"] (p3);
(db+2) ->["\escale{$o_{1}$}"] (db+3);
};
\end{tikzpicture}

\noindent
Overall, $K_{\tau^b_0}=(H_0,D_{0,0}, \dots, D_{6m-1,0})$. 
\begin{proof}[Proof of Lemma~\ref{lem:key_unions_1}.\ref{lem:key_unions_1_completeness} for $\tau^b_0$]
For $j\in \{0,\dots, 6m-1\}$ the TSs $H_0$ and $D_{j,0}$ interact as follows:
If $(sup_K, sig_K)$ is a region of $K_{\tau^b_0}$ solving $(k, h_{0, 4b+1})$ then either $sig_K(o_0)=(0,b)$ \emph{and} $sig_K(o_1)=(0,1)$ or $sig_K(o_0)=(b,0)$ \emph{and} $sig_K(o_1)=(1,0)$.
By $\edge{o_0}d_{j,0,1}$, $d_{j,0,2}\edge{o_1}$ and Lemma~\ref{lem:observations}, 
if $sig_K(o_0)=(0,b), sig_K(o_1)=(0,1)$ then $sup_K(d_{j,0,1})=b$ and $sup_K(d_{j,0,2})=0$.
This implies $sig_K(k_j)=(b,0)$.
Similarly, $sig_K(o_0)=(b,0), sig_K(o_1)=(1,0)$ implies $sup_K(d_{j,0,1})=0$ and $sup_K(d_{j,0,2})=b$ yielding $sig_K(k_j)=(0,b)$.
Hence, it is $sig_K(k_0)=\dots=sig_K(k_{6m-1})=(b,0)$ or $sig_K(k_0)=\dots=sig_K(k_{6m-1})=(0,b)$.

To prove Lemma~\ref{lem:key_unions_1}.\ref{lem:key_unions_1_completeness} for $K_{\tau^b_0}$ it remains to argue that a $\tau^b_0$-region $(sig, sup)$ of $K_{\tau^b_0}$ solving $(k, h_{0, 4b+1})$ satisfies $sig_K(o_0)=(0,b), sig_K(o_1)=(0,1)$ or $sig_K(o_0)=(b,0), sig_K(o_1)=(1,0)$. 
Let $ E_{0}=\{ (m,m) \mid 0\le m \leq b\}$.

By definition, if $sig(k)=(m,m) \in E_0$ then $sup(h_{0,3b+1}),sup(h_{0,5b+1})\geq m$.
Event $(m,m)$ occurs at every state $s$ of $\tau^b_0$ satisfying $s\geq m$.
Hence, by $\neg h_{0,4b+1}\edge{(m,m)}$, we get $sup(h_{0,4b+1}) < m$. 
Observe, that $z$ occurs always $b$ times in a row.
Therefore, by $sup(h_{0,3b+1}) \geq m$, $sup(h_{0,4b+1}) < m$ and Lemma~\ref{lem:observations}, we have $sup(z)=(1,0)$, $sig(o_1)=(0,1)$ and immediately obtain $sup(h_{0,2b})=0$ and $sup(h_{0,3b+1})=b$.
Moreover, by $sig(k)=(m,m) $ and $sup(h_{0,3b+1})=b$ we get $sup(h_{0,2b+1})=b$ implying with $sup(h_{0,2b})=0$ that $sig(o_0)=(0,b)$.
Thus, we have $sig(o_0)=(0,b)$ and $sig(o_1)=(0,1)$.  

Otherwise, if $sig(k)\not\in E_0$, then Lemma~\ref{lem:observations} ensures $sig(k)\in \{(1,0), (0,1)\}$.
If $sig(k)=(0,1)$ then, by $s\edge{(0,1)}$ for every state $s\in \{0,\dots, b-1\}$ of $\tau^b_0$, we have $sup(h_{0, 4b+1})=b$.
Moreover, again by $sig(k)=(0,1)$ we have $sup(h_{0,b})=sup(h_{0,3b+1})=b$ and $sup(h_{0,2b+1})=sup(h_{0, 5b+1})=0$.
By $sup(h_{0, 3b+1})=sup(h_{0,4b+1})=b$ we have $sig(z)\in E_0$ which together with $sup(h_{0,b})=b$ implies $sup(h_{0, 2b})=b$.
Thus, by $sup(h_{0,2b})=b$ and $sup(h_{0, 2b+1})=0$, it is $sig(o_0)=(b,0)$.
Moreover, by $sup(h_{0, 4b+1})=b$ and $sup(h_{0, 5b+1 })=0$, we conclude $sig(o_1)=(1,0)$.
Hence, we have $sig(o_0)=(b,0)$ and $sig(o_1)=(1,0)$. 
Similar arguments show that $sig_K(k)=(1,0)$ implies $sig(o_0)=(0,b)$ and $sig(o_1)=(0,1)$. 
Overall, this proves the announced signatures of $o_0$ and $o_1$.
Hence, $K_{\tau^b_0}$ satisfies Lemma~\ref{lem:key_unions_1}.\ref{lem:key_unions_1_completeness}.
\end{proof}

\textbf{The union $\mathbf{K_{\tau^b_1}}$} uses the next TS $H_1$ to provide the key atom $(k, h_{1,2b+4})$:

\noindent
\begin{tikzpicture}
\node at (-0.75,0) {$H_1=$};
\node (h0) at (0,0) {\nscale{$h_{1,0}$}};
\node (h1) at (1,0) {\nscale{}};
\node (h_2_dots) at (1.25,0) {\nscale{$\dots$}};
\node (h_k_1) at (1.5,0) {};
\node (h_k) at (2.5,0) {\nscale{$h_{1,b}$}};
\node (h_k+1) at (3.7,0) {\nscale{$h_{1,b+1}$}};
\node (h_k+2) at (4.9,0) {\nscale{$h_{1,b+2}$}};
\node (h_k+3) at (6.1,0) {\nscale{}};
\node (h_k+4_dots) at (6.35,0) {\nscale{$\dots$}};
\node (h_2k+1) at (6.6,0) {\nscale{}};
\node (h_2k+2) at (7.8,0) {\nscale{$h_{1,2b+2}$}};
\node (h_2k+3) at (9.2,0) {\nscale{$h_{1,2b+3}$}};
\node (h_2k+4) at (10.6,0) {\nscale{$h_{1,2b+4}$}};
\node (h_2k+5) at (10.6,-1) {\nscale{$h_{1,2b+5}$}};
\node (h_2k+6) at (9.4,-1) {};
\node (h_k+5_dots) at (9.15,-1) {\nscale{$\dots$}};
\node (h_3k+4) at (8.9,-1) {};
\node (h_3k+5) at (7.7,-1) {\nscale{$h_{1,3b+5}$}};
\graph { (h0) ->["\escale{$k$}"] (h1) 
(h_k_1)->["\escale{$k$}"] (h_k) ->["\escale{$z_0$}"] (h_k+1)->["\escale{$o_0$}"] (h_k+2)->["\escale{$k$}"] (h_k+3);
(h_2k+1)->["\escale{$k$}"] (h_2k+2)->["\escale{$z_1$}"] (h_2k+3)->["\escale{$z_0$}"] (h_2k+4)->["\escale{$o_2$}"] (h_2k+5)->[swap, "\escale{$k$}"] (h_2k+6);
(h_3k+4)->[swap, "\escale{$k$}"] (h_3k+5);
;};
\end{tikzpicture}

\noindent
Furthermore, $K_{\tau^b_1}$ contains for $j\in \{0,\dots, 6m-1\}$ the TS $D_{j,1}$ which provides the interface event $k_j$:
\noindent
\begin{tikzpicture}[baseline=-2pt]
\node at (-1,0) {$D_{j,1}=$};
\foreach \i in {0,...,3} {\coordinate (\i) at (\i*1.4,0);}
\foreach \i in {0,...,3} {\node (p\i) at (\i) {\nscale{$d_{j,1,\i}$}};}
\graph { (p0) ->["\escale{$o_{0}$}"] (p1) ->["\escale{$k_j$}"] (p2) ->["\escale{$o_{2}$}"] (p3);};
\end{tikzpicture}

\noindent
Altogether, $K_{\tau^b_1}=U(H_1,D_{0,1},\dots, D_{6m-1,1})$.
\begin{proof}[Proof of Lemma~\ref{lem:key_unions_1}.\ref{lem:key_unions_1_completeness} for $\tau^b_1$]
For $j\in \{0,\dots, 6m-1\}$ the TSs $H_1$ and $D_{j,1}$ interact as follows:
If $(sup_K, sig_K)$ is a $\tau^b_1$-region of $K_{\tau^b_1}$ solving $(k, h_{1,2b+4})$ then either $sig_K(o_0)=sig_K(o_2)=(b,0)$ or $sig_K(o_0)=sig_K(o_2)=(0,b)$.
Clearly, $sig_K(o_0)=sig_K(o_2)=(b,0)$, respectively $sig_K(o_0)=sig_K(o_2)=(0,b)$, implies $sig_K(k_0)=\dots=sig_K(k_{6m-1})=(0,b)$, respectively $sig_K(k_0)=\dots=sig_K(k_{6m-1})=(b,0)$.

We argue that the $\tau^b_1$-solvability of $(k, h_{1,2b+4})$ implies the announced signatures of $o_0,o_2$.
If $(sup_K, sig_K)$ is a $\tau^b_1$-region that solves $(k, h_{1,2b+4})$ then, by definition of $\tau^b_1$ and Lemma~\ref{lem:observations}, we get $sig_K(k)\in \{(1,0), (0,1)\}$.
Let $sig_K(k)=(0,1)$.
The event $(0,1)$ occurs at every $s\in \{0,\dots, b-1\}$ of $\tau^b_1$.
Hence, $\neg sup_K(h_{1,2b+4})\edge{(0,1)}$ implies $sup_K(h_{1,2b+4})=b$.
Moreover, $k$ occurs $b$ times in a row.
Thus, by $sig_K(k)=(0,1)$ and Lemma~\ref{lem:observations}, we obtain $sup_K(h_{1,b})=b$ and $sup_K(h_{1,b+2})=sup_K(h_{1,2b+5})=0$.
This implies, by $h_{1,2b+4}\edge{o_2}h_{1,2b+5}$, $sup(h_{1,2b+4})=b$ and $sup(h_{1,2b+5})=0$, that $sig(o_2)=(b,0)$.
Hence, by $sup_K(h_{1,b})=sup_K(h_{1,2b+4})=b$, $h_{1,b}\edge{z_0}$ and $\edge{z_0}h_{1,2b+4}$, we get $sig(z_0)=(0,0)$.
Finally, by $sup(h_{1,b})=b$, $h_{1,b}\edge{z_0}$ and $sig(z_0)=(0,0)$ we deduce $sup(h_{1,b+1})=b$.
Hence, by $h_{1,b+1}\edge{o_0}h_{1,b+2}$, $sup(h_{1,b+1})=b$ and $sup(h_{1,b+1})=0$ we have $sig(o_0)=(b,0)$.
Altogether, we have that $sig(o_0)=sig(o_2)=(b,0)$.
Similarly, one verifies that $sig_K(k)=(1,0)$ results in $sig(o_0)=sig(o_2)=(0,b)$.
This proves Lemma~\ref{lem:key_unions_1}.\ref{lem:key_unions_1_completeness} for $K_{\tau^b_1}$.
\end{proof}

\textbf{The union $\mathbf{K}$} uses the following TS $H_2$ to provide the key atom $(h_{2,0}, h_{2,b})$:

\noindent
\begin{tikzpicture}
\node at (-0.75,0) {$H_2=$};
\node (h0) at (0,0) {\nscale{$h_{2,0}$}};
\node (h1) at (1,0) {\nscale{}};
\node (hdots_1) at (1.25,0) {\nscale{$\dots$}};
\node (hb_1) at (1.5,0) {};
\node (hb) at (2.5,0) {\nscale{$h_{2,b}$}};
\node (hb+1) at (3.75,0) {\nscale{$h_{2,b+1}$}};
\node (hb+2) at (5,0) {};
\node (hdots_2) at (5.25,0) {\nscale{$\dots$}};
\node (h2b) at (5.5,0) {};
\node (h2b+1) at (6.5,0) {\nscale{$h_{2,2b+1}$}};
\node (h2b+2) at (7.75,0) {\nscale{$h_{2,2b+2}$}};
\node (h2b+3) at (9,0) {};
\node (hdots_3) at (9.25,0) {\nscale{$\dots$}};
\node (h3b+1) at (9.5,0) {};
\node (h3b+2) at (10.5,0) {\nscale{$h_{2,3b+2}$}};
\graph { 
(h0) ->["\escale{$k$}"] (h1); 
(hb_1)->["\escale{$k$}"] (hb) ->["\escale{$o_0$}"] (hb+1)->["\escale{$k$}"] (hb+2);
(h2b) ->["\escale{$k$}"] (h2b+1)->["\escale{$o_2$}"] (h2b+2)->["\escale{$k$}"] (h2b+3);
(h3b+1) ->["\escale{$k$}"] (h3b+2);
};
\end{tikzpicture}

\noindent
$K$ also contains the TSs $D_{0,1},\dots, D_{6m-1,1}$, thus $K=(H_2, D_{0,1},\dots, D_{6m-1,1})$.

\begin{proof}[Proof of Lemma~\ref{lem:key_unions_1}.\ref{lem:key_unions_1_completeness} for $\tau^b_2$]
$K$ works as follows:
The event $k$ occurs $b$ times in a row at $h_{2,0}$.
Therefore, by Lemma~\ref{lem:observations}, a region $(sup_K, sig_K)$ solving $(h_{2,0}, h_{2,b})$ satisfies $sig_K(k)\in \{(1,0), (0,1)\}$.
If $sig_K(k)=(1,0)$ then $sup_K(h_{2,b})=sup_K(h_{2,2b+1})=b$ and $sup_K(h_{2,b+1})=sup_K(h_{2,2b+2})=0$ implying $sig_k(o_0)=sig_k(o_2)=(b,0)$.
Otherwise, if $sig_K(k)=(0,1)$ then $sup_K(h_{2,b})=sup_K(h_{2,2b+1})=0$ and $sup_K(h_{2,b+1})=sup_K(h_{2,2b+2})=b$ which implies $sig_k(o_0)=sig_k(o_2)=(0,b)$.
As already discussed for $K_{\tau^b_1}$, we have that $sig_k(o_0)=sig_k(o_2)=(b,0)$ ($sig_k(o_0)=sig_k(o_2)=(0,b)$) implies $sig_K(k_j)=(0,b)$ ($sig_K(k_j)=(b,0)$) for $j\in \{0,\dots, 6m-1\}$.
Hence, Lemma~\ref{lem:key_unions_1}.\ref{lem:key_unions_1_completeness} is true for $K$.
\end{proof}
It remains to show that $K_{\tau^b_0}, K_{\tau^b_1}$ and $K$ satisfy the objective of \emph{existence}:
\begin{proof}[Proof of Lemma~\ref{lem:key_unions_1}.\ref{lem:key_unions_1_existence}]
We present corresponding regions.
Let $S$ and $E$ be the set of all states and of all events of $K,K_{\tau^b_0}$ and $K_{\tau^b_1}$, respectively.
We define mappings $sig:E\longrightarrow E_{\tau^b_1}$ and $sup: S \longrightarrow S_{\tau^b_1}$ by:

\[sig(e)=
\begin{cases}
(0,b), & \text{if } e\in \{k_{0},\dots, k_{6m-1}\} \\
(0,1), & \text{if } e = k\\
(0,0), &  \text{if } e \in \{z, z_0, z_1 \} \\
(1,0), & \text{if } e = o_1\\
(b,0), & \text{if } e\in \{o_0, o_2 \} \\
\end{cases}
sup(s)=
\begin{cases}
0 , & \text{if } s\in \{h_{0,0},h_{1,0}, h_{2,0}\} \\
b , & \text{if } s\in \{d_{j,0,0}, d_{j,1,0}\}\\
     &  \text{and } 0 \leq j \leq 6m-1
\end{cases}
\]
By $sig_{K}$, $sig_{K_{\tau^b_0}}$ and $sig_{K_{\tau^b_1}}$ ($sup_{K}$, $sup_{K_{\tau^b_0}}$ and $sup_{K_{\tau^b_1}}$) we denote the restriction of $sig$ ($sup$) to the events (states) of $K$, $K_{\tau^b_0}$ and $K_{\tau^b_1}$, respectively.
As $sup$ defines the support of every corresponding initial state, by Lemma~\ref{lem:observations}, we obtain fitting regions $(sup_{K}, sig_{K})$, $(sup_{K_{\tau^b_0}},sig_{K_{\tau^b_0}})$ and $(sup_{K_{\tau^b_1}},sig_{K_{\tau^b_1}})$ that solve the corresponding key atom.
Figure~\ref{fig:example_1} sketches this region for $K^0_{\tau^2_1}$ and $K$.

\end{proof}

\begin{figure}[t!]

\begin{tikzpicture}
\begin{scope}
\foreach \i in {0,...,7} {\coordinate (d\i) at (\i*1cm,0);}
\foreach \i in {0,...,7} {\node (t\i) at (d\i) {\nscale{$t_{0,0,\i}$}};}
\foreach \i in {0,...,7} {\coordinate (s\i) at (\i*1cm,-0.25cm);}
\foreach \i in {0,3,4,5,6} {\node[opacity=0.7](s\i) at (s\i) {\nscale{$[0]$}};}
\foreach \i in {2} {\node[opacity=0.7](s\i) at (s\i) {\nscale{$[1]$}};}
\foreach \i in {1,7} {\node[opacity=0.7](s\i) at (s\i) {\nscale{$[2]$}};}
\foreach \i in {0,...,6} {\coordinate (g\i) at (\i*1cm+0.5cm ,-0.25cm);}
\foreach \i in {0,6} {\node[opacity=0.7](g\i) at (g\i) {\nscale{$(0,2)$}};}
\foreach \i in {1,2} {\node[opacity=0.7](g\i) at (g\i) {\nscale{$(1,0)$}};}
\foreach \i in {3,4,5} {\node[opacity=0.7](g\i) at (g\i) {\nscale{$(0,0)$}};}
\graph { (t0) ->["\escale{$k_0$}"] (t1) ->["\escale{$X_0$}"] (t2) ->["\escale{$X_0$}"] (t3) ->["\escale{$x_0$}"] (t4) ->["\escale{$X_2$}"] (t5)->["\escale{$X_2$}"](t6)->["\escale{$k_1$}"](t7);

};
\end{scope}
\begin{scope}[yshift=-0.9cm]
\foreach \i in {0,...,5} {\coordinate (\i) at (\i*1cm,0);}
\foreach \i in {0,...,5} {\node (t\i) at (\i) {\nscale{$t_{0,1,\i}$}};}
\foreach \i in {0,...,5} {\coordinate (s\i) at (\i*1cm,-.25cm);}
\foreach \i in {0,4} {\node[opacity=0.7] (s\i) at (s\i) {\nscale{$[0]$}};}
\foreach \i in {1,2,3,5} {\node[opacity=0.7] (s\i) at (s\i) {\nscale{$[2]$}};}
\foreach \i in {0,...,4} {\coordinate (g\i) at (\i*1cm+0.5cm ,-0.25cm);}
\foreach \i in {0,4} {\node[opacity=0.7](g\i) at (g\i) {\nscale{$(0,2)$}};}
\foreach \i in {1,2} {\node[opacity=0.7](g\i) at (g\i) {\nscale{$(0,0)$}};}
\foreach \i in {3} {\node[opacity=0.7](g\i) at (g\i) {\nscale{$(2,0)$}};}
\graph { (t0) ->["\escale{$k_2$}"] (t1) ->["\escale{$X_1$}"] (t2) ->["\escale{$X_1$}"] (t3) ->["\escale{$p_0$}"] (t4) ->["\escale{$k_3$}"] (t5);
};
\end{scope}
\begin{scope}[xshift=6cm, yshift=-0.9cm]
\foreach \i in {0,...,4} {\coordinate (\i) at (\i*1cm,0);}
\foreach \i in {0,...,4} {\node (t\i) at (\i) {\nscale{$t_{0,2,\i}$}};}
\foreach \i in {0,...,4} {\coordinate (s\i) at (\i*1cm,-0.25cm);}
\foreach \i in {0,3} {\node[opacity=0.7] (s\i) at (s\i) {\nscale{$[0]$}};}
\foreach \i in {1,2,4} {\node[opacity=0.7] (s\i) at (s\i) {\nscale{$[2]$}};}
\foreach \i in {0,...,3} {\coordinate (g\i) at (\i*1cm+0.5cm ,-0.25cm);}
\foreach \i in {0,3} {\node[opacity=0.7](g\i) at (g\i) {\nscale{$(0,2)$}};}
\foreach \i in {1} {\node[opacity=0.7](g\i) at (g\i) {\nscale{$(0,0)$}};}
\foreach \i in {2} {\node[opacity=0.7](g\i) at (g\i) {\nscale{$(2,0)$}};}
\graph { (t0) ->["\escale{$k_4$}"] (t1) ->["\escale{$x_0$}"] (t2) ->["\escale{$p_0$}"] (t3) ->["\escale{$k_5$}"] (t4);
};
\end{scope}
\begin{scope}[yshift=-1.8cm]

\foreach \i in {0,...,7} {\coordinate (\i) at (\i*1cm,0);}
\foreach \i in {0,...,7} {\node (t\i) at (\i) {\nscale{$t_{1,0,\i}$}};}
\foreach \i in {0,...,7} {\coordinate (s\i) at (\i*1cm,-0.25cm);}
\foreach \i in {0,4,5,6} {\node[opacity=0.7](s\i) at (s\i) {\nscale{$[0]$}};}
\foreach \i in {1,2,3,7} {\node[opacity=0.7](s\i) at (s\i) {\nscale{$[2]$}};}
\foreach \i in {0,...,6} {\coordinate (g\i) at (\i*1cm+0.5cm ,-0.25cm);}
\foreach \i in {0,6} {\node[opacity=0.7](g\i) at (g\i) {\nscale{$(0,2)$}};}
\foreach \i in {3} {\node[opacity=0.7](g\i) at (g\i) {\nscale{$(2,0)$}};}
\foreach \i in {1,2,4,5} {\node[opacity=0.7](g\i) at (g\i) {\nscale{$(0,0)$}};}
\graph { (t0) ->["\escale{$k_6$}"] (t1) ->["\escale{$X_2$}"] (t2) ->["\escale{$X_2$}"] (t3) ->["\escale{$x_1$}"] (t4) ->["\escale{$X_3$}"] (t5)->["\escale{$X_3$}"](t6)->["\escale{$k_7$}"](t7);
};
\end{scope}
\begin{scope}[yshift=-2.7cm]

\foreach \i in {0,...,5} {\coordinate (\i) at (\i*1cm,0);}
\foreach \i in {0,...,5} {\node (t\i) at (\i) {\nscale{$t_{1,1,\i}$}};}
\foreach \i in {0,...,5} {\coordinate (s\i) at (\i*1cm,-0.25cm);}
\foreach \i in {0,3,4} {\node[opacity=0.7](s\i) at (s\i) {\nscale{$[0]$}};}
\foreach \i in {2} {\node[opacity=0.7](s\i) at (s\i) {\nscale{$[1]$}};}
\foreach \i in {1,5} {\node[opacity=0.7](s\i) at (s\i) {\nscale{$[2]$}};}
\foreach \i in {0,...,4} {\coordinate (g\i) at (\i*1cm+0.5cm ,-0.25cm);}
\foreach \i in {0,4} {\node[opacity=0.7](g\i) at (g\i) {\nscale{$(0,2)$}};}
\foreach \i in {1,2} {\node[opacity=0.7](g\i) at (g\i) {\nscale{$(1,0)$}};}
\foreach \i in {3} {\node[opacity=0.7](g\i) at (g\i) {\nscale{$(0,0)$}};}
\graph { (t0) ->["\escale{$k_8$}"] (t1) ->["\escale{$X_0$}"] (t2) ->["\escale{$X_0$}"] (t3) ->["\escale{$p_1$}"] (t4) ->["\escale{$k_9$}"] (t5);
};
\end{scope}
\begin{scope}[xshift=6cm, yshift=-2.7cm]

\foreach \i in {0,...,4} {\coordinate (\i) at (\i*1cm,0);}
\foreach \i in {0,...,4} {\node (t\i) at (\i) {\nscale{$t_{1,2,\i}$}};}
\foreach \i in {0,...,4} {\coordinate (s\i) at (\i*1cm,-0.25cm);}
\foreach \i in {0,2,3} {\node[opacity=0.7](s\i) at (s\i) {\nscale{$[0]$}};}
\foreach \i in {1,4} {\node[opacity=0.7](s\i) at (s\i) {\nscale{$[2]$}};}
\foreach \i in {0,...,3} {\coordinate (g\i) at (\i*1cm+0.5cm ,-0.25cm);}
\foreach \i in {0,3} {\node[opacity=0.7](g\i) at (g\i) {\nscale{$(0,2)$}};}
\foreach \i in {1} {\node[opacity=0.7](g\i) at (g\i) {\nscale{$(2,0)$}};}
\foreach \i in {2} {\node[opacity=0.7](g\i) at (g\i) {\nscale{$(0,0)$}};}
\graph { (t0) ->["\escale{$k_{10}$}"] (t1) ->["\escale{$x_1$}"] (t2) ->["\escale{$p_1$}"] (t3) ->["\escale{$k_{11}$}"] (t4);
};
\end{scope}
\begin{scope}[yshift=-3.6cm]

\foreach \i in {0,...,7} {\coordinate (\i) at (\i*1cm,0);}
\foreach \i in {0,...,7} {\node (t\i) at (\i) {\nscale{$t_{2,0,\i}$}};}
\foreach \i in {0,...,7} {\coordinate (s\i) at (\i*1cm,-0.25cm);}
\foreach \i in {0,6} {\node[opacity=0.7](s\i) at (s\i) {\nscale{$[0]$}};}
\foreach \i in {5} {\node[opacity=0.7](s\i) at (s\i) {\nscale{$[1]$}};}
\foreach \i in {1,2,3,4,7} {\node[opacity=0.7](s\i) at (s\i) {\nscale{$[2]$}};}
\foreach \i in {0,...,6} {\coordinate (g\i) at (\i*1cm+0.5cm ,-0.25cm);}
\foreach \i in {0,6} {\node[opacity=0.7](g\i) at (g\i) {\nscale{$(0,2)$}};}
\foreach \i in {4,5} {\node[opacity=0.7](g\i) at (g\i) {\nscale{$(1,0)$}};}
\foreach \i in {1,2,3} {\node[opacity=0.7](g\i) at (g\i) {\nscale{$(0,0)$}};}
\graph { (t0) ->["\escale{$k_{12}$}"] (t1) ->["\escale{$X_1$}"] (t2) ->["\escale{$X_1$}"] (t3) ->["\escale{$x_2$}"] (t4) ->["\escale{$X_0$}"] (t5)->["\escale{$X_0$}"](t6)->["\escale{$k_{13}$}"](t7);
};
\end{scope}
\begin{scope}[yshift=-4.5cm]

\foreach \i in {0,...,5} {\coordinate (\i) at (\i*1cm,0);}
\foreach \i in {0,...,5} {\node (t\i) at (\i) {\nscale{$t_{2,1,\i}$}};}
\foreach \i in {0,...,5} {\coordinate (s\i) at (\i*1cm,-0.25cm);}
\foreach \i in {0,4} {\node[opacity=0.7](s\i) at (s\i) {\nscale{$[0]$}};}
\foreach \i in {1,2,3,5} {\node[opacity=0.7](s\i) at (s\i) {\nscale{$[2]$}};}
\foreach \i in {0,...,4} {\coordinate (g\i) at (\i*1cm+0.5cm ,-0.25cm);}
\foreach \i in {0,4} {\node[opacity=0.7](g\i) at (g\i) {\nscale{$(0,2)$}};}
\foreach \i in {3} {\node[opacity=0.7](g\i) at (g\i) {\nscale{$(2,0)$}};}
\foreach \i in {1,2} {\node[opacity=0.7](g\i) at (g\i) {\nscale{$(0,0)$}};}
\graph { (t0) ->["\escale{$k_{14}$}"] (t1) ->["\escale{$X_3$}"] (t2) ->["\escale{$X_3$}"] (t3) ->["\escale{$p_2$}"] (t4) ->["\escale{$k_{15}$}"] (t5);
};
\end{scope}
\begin{scope}[xshift=6cm, yshift=-4.5cm]

\foreach \i in {0,...,4} {\coordinate (\i) at (\i*1cm,0);}
\foreach \i in {0,...,4} {\node (t\i) at (\i) {\nscale{$t_{2,2,\i}$}};}
\foreach \i in {0,...,4} {\coordinate (s\i) at (\i*1cm,-0.25cm);}
\foreach \i in {0,3} {\node[opacity=0.7](s\i) at (s\i) {\nscale{$[0]$}};}
\foreach \i in {1,2,4} {\node[opacity=0.7](s\i) at (s\i) {\nscale{$[2]$}};}
\foreach \i in {0,...,3} {\coordinate (g\i) at (\i*1cm+0.5cm ,-0.25cm);}
\foreach \i in {0,3} {\node[opacity=0.7](g\i) at (g\i) {\nscale{$(0,2)$}};}
\foreach \i in {2} {\node[opacity=0.7](g\i) at (g\i) {\nscale{$(2,0)$}};}
\foreach \i in {1} {\node[opacity=0.7](g\i) at (g\i) {\nscale{$(0,0)$}};}
\graph { (t0) ->["\escale{$k_{16}$}"] (t1) ->["\escale{$x_2$}"] (t2) ->["\escale{$p_2$}"] (t3) ->["\escale{$k_{17}$}"] (t4);
};
\end{scope}
\begin{scope}[yshift=-5.4cm]

\foreach \i in {0,...,7} {\coordinate (\i) at (\i*1cm,0);}
\foreach \i in {0,...,7} {\node (t\i) at (\i) {\nscale{$t_{3,0,\i}$}};}
\foreach \i in {0,...,7} {\coordinate (s\i) at (\i*1cm,-0.25cm);}
\foreach \i in {0,4,5,6} {\node[opacity=0.7](s\i) at (s\i) {\nscale{$[0]$}};}
\foreach \i in {1,2,3,7} {\node[opacity=0.7](s\i) at (s\i) {\nscale{$[2]$}};}
\foreach \i in {0,...,6} {\coordinate (g\i) at (\i*1cm+0.5cm ,-0.25cm);}
\foreach \i in {0,6} {\node[opacity=0.7](g\i) at (g\i) {\nscale{$(0,2)$}};}
\foreach \i in {3} {\node[opacity=0.7](g\i) at (g\i) {\nscale{$(2,0)$}};}
\foreach \i in {1,2,4,5} {\node[opacity=0.7](g\i) at (g\i) {\nscale{$(0,0)$}};}
\graph { (t0) ->["\escale{$k_{18}$}"] (t1) ->["\escale{$X_2$}"] (t2) ->["\escale{$X_2$}"] (t3) ->["\escale{$x_3$}"] (t4) ->["\escale{$X_5$}"] (t5)->["\escale{$X_5$}"](t6)->["\escale{$k_{19}$}"](t7);
};
\end{scope}
\begin{scope}[yshift=-6.3cm]

\foreach \i in {0,...,5} {\coordinate (\i) at (\i*1cm,0);}
\foreach \i in {0,...,5} {\node (t\i) at (\i) {\nscale{$t_{3,1,\i}$}};}
\foreach \i in {0,...,5} {\coordinate (s\i) at (\i*1cm,-0.25cm);}
\foreach \i in {0,3,4} {\node[opacity=0.7](s\i) at (s\i) {\nscale{$[0]$}};}
\foreach \i in {2} {\node[opacity=0.7](s\i) at (s\i) {\nscale{$[1]$}};}
\foreach \i in {1,5} {\node[opacity=0.7](s\i) at (s\i) {\nscale{$[2]$}};}
\foreach \i in {0,...,4} {\coordinate (g\i) at (\i*1cm+0.5cm ,-0.25cm);}
\foreach \i in {0,4} {\node[opacity=0.7](g\i) at (g\i) {\nscale{$(0,2)$}};}
\foreach \i in {1,2} {\node[opacity=0.7](g\i) at (g\i) {\nscale{$(1,0)$}};}
\foreach \i in {3} {\node[opacity=0.7](g\i) at (g\i) {\nscale{$(0,0)$}};}
\graph { (t0) ->["\escale{$k_{20}$}"] (t1) ->["\escale{$X_4$}"] (t2) ->["\escale{$X_4$}"] (t3) ->["\escale{$p_3$}"] (t4) ->["\escale{$k_{21}$}"] (t5);
};
\end{scope}
\begin{scope}[xshift=6cm, yshift=-6.3cm]

\foreach \i in {0,...,4} {\coordinate (\i) at (\i*1cm,0);}
\foreach \i in {0,...,4} {\node (t\i) at (\i) {\nscale{$t_{3,2,\i}$}};}
\foreach \i in {0,...,4} {\coordinate (s\i) at (\i*1cm,-0.25cm);}
\foreach \i in {0,2,3} {\node[opacity=0.7](s\i) at (s\i) {\nscale{$[0]$}};}
\foreach \i in {1,4} {\node[opacity=0.7](s\i) at (s\i) {\nscale{$[2]$}};}
\foreach \i in {0,...,3} {\coordinate (g\i) at (\i*1cm+0.5cm ,-0.25cm);}
\foreach \i in {0,3} {\node[opacity=0.7](g\i) at (g\i) {\nscale{$(0,2)$}};}
\foreach \i in {1} {\node[opacity=0.7](g\i) at (g\i) {\nscale{$(2,0)$}};}
\foreach \i in {2} {\node[opacity=0.7](g\i) at (g\i) {\nscale{$(0,0)$}};}
\graph { (t0) ->["\escale{$k_{22}$}"] (t1) ->["\escale{$x_3$}"] (t2) ->["\escale{$p_3$}"] (t3) ->["\escale{$k_{23}$}"] (t4);
};
\end{scope}
\begin{scope}[yshift=-7.2cm]

\foreach \i in {0,...,7} {\coordinate (\i) at (\i*1cm,0);}
\foreach \i in {0,...,7} {\node (t\i) at (\i) {\nscale{$t_{4,0,\i}$}};}
\foreach \i in {0,...,7} {\coordinate (s\i) at (\i*1cm,-0.25cm);}
\foreach \i in {0,6} {\node[opacity=0.7](s\i) at (s\i) {\nscale{$[0]$}};}
\foreach \i in {5} {\node[opacity=0.7](s\i) at (s\i) {\nscale{$[1]$}};}
\foreach \i in {1,2,3,4,7} {\node[opacity=0.7](s\i) at (s\i) {\nscale{$[2]$}};}
\foreach \i in {0,...,6} {\coordinate (g\i) at (\i*1cm+0.5cm ,-0.25cm);}
\foreach \i in {0,6} {\node[opacity=0.7](g\i) at (g\i) {\nscale{$(0,2)$}};}
\foreach \i in {4,5} {\node[opacity=0.7](g\i) at (g\i) {\nscale{$(1,0)$}};}
\foreach \i in {1,2,3} {\node[opacity=0.7](g\i) at (g\i) {\nscale{$(0,0)$}};}
\graph { (t0) ->["\escale{$k_{24}$}"] (t1) ->["\escale{$X_1$}"] (t2) ->["\escale{$X_1$}"] (t3) ->["\escale{$x_4$}"] (t4) ->["\escale{$X_4$}"] (t5)->["\escale{$X_4$}"](t6)->["\escale{$k_{25}$}"](t7);
};
\end{scope}
\begin{scope}[yshift=-8.1cm]

\foreach \i in {0,...,5} {\coordinate (\i) at (\i*1cm,0);}
\foreach \i in {0,...,5} {\node (t\i) at (\i) {\nscale{$t_{4,1,\i}$}};}
\foreach \i in {0,...,5} {\coordinate (s\i) at (\i*1cm,-0.25cm);}
\foreach \i in {0,4,6} {\node[opacity=0.7](s\i) at (s\i) {\nscale{$[0]$}};}
\foreach \i in {1,2,3,5} {\node[opacity=0.7](s\i) at (s\i) {\nscale{$[2]$}};}
\foreach \i in {0,...,4} {\coordinate (g\i) at (\i*1cm+0.5cm ,-0.25cm);}
\foreach \i in {0,4} {\node[opacity=0.7](g\i) at (g\i) {\nscale{$(0,2)$}};}
\foreach \i in {3} {\node[opacity=0.7](g\i) at (g\i) {\nscale{$(2,0)$}};}
\foreach \i in {1,2} {\node[opacity=0.7](g\i) at (g\i) {\nscale{$(0,0)$}};}
\graph { (t0) ->["\escale{$k_{26}$}"] (t1) ->["\escale{$X_5$}"] (t2) ->["\escale{$X_5$}"] (t3) ->["\escale{$p_4$}"] (t4) ->["\escale{$k_{27}$}"] (t5);
};
\end{scope}
\begin{scope}[xshift=6cm, yshift=-8.1cm]

\foreach \i in {0,...,4} {\coordinate (\i) at (\i*1cm,0);}
\foreach \i in {0,...,4} {\node (t\i) at (\i) {\nscale{$t_{4,2,\i}$}};}
\foreach \i in {0,...,4} {\coordinate (s\i) at (\i*1cm,-0.25cm);}
\foreach \i in {0,3} {\node[opacity=0.7](s\i) at (s\i) {\nscale{$[0]$}};}
\foreach \i in {1,2,4} {\node[opacity=0.7](s\i) at (s\i) {\nscale{$[2]$}};}
\foreach \i in {0,...,3} {\coordinate (g\i) at (\i*1cm+0.5cm ,-0.25cm);}
\foreach \i in {0,3} {\node[opacity=0.7](g\i) at (g\i) {\nscale{$(0,2)$}};}
\foreach \i in {2} {\node[opacity=0.7](g\i) at (g\i) {\nscale{$(2,0)$}};}
\foreach \i in {1} {\node[opacity=0.7](g\i) at (g\i) {\nscale{$(0,0)$}};}
\graph { (t0) ->["\escale{$k_{28}$}"] (t1) ->["\escale{$x_4$}"] (t2) ->["\escale{$p_4$}"] (t3) ->["\escale{$k_{29}$}"] (t4);
};
\end{scope}
\begin{scope}[yshift=-9cm]

\foreach \i in {0,...,7} {\coordinate (\i) at (\i*1cm,0);}
\foreach \i in {0,...,7} {\node (t\i) at (\i) {\nscale{$t_{5,0,\i}$}};}
\foreach \i in {0,...,7} {\coordinate (s\i) at (\i*1cm,-0.25cm);}
\foreach \i in {0,3,4,5,6} {\node[opacity=0.7](s\i) at (s\i) {\nscale{$[0]$}};}
\foreach \i in {2} {\node[opacity=0.7](s\i) at (s\i) {\nscale{$[1]$}};}
\foreach \i in {1,7} {\node[opacity=0.7](s\i) at (s\i) {\nscale{$[2]$}};}
\foreach \i in {0,...,6} {\coordinate (g\i) at (\i*1cm+0.5cm ,-0.25cm);}
\foreach \i in {0,6} {\node[opacity=0.7](g\i) at (g\i) {\nscale{$(0,2)$}};}
\foreach \i in {1,2} {\node[opacity=0.7](g\i) at (g\i) {\nscale{$(1,0)$}};}
\foreach \i in {3, 4,5} {\node[opacity=0.7](g\i) at (g\i) {\nscale{$(0,0)$}};}
\graph { (t0) ->["\escale{$k_{30}$}"] (t1) ->["\escale{$X_4$}"] (t2) ->["\escale{$X_4$}"] (t3) ->["\escale{$x_5$}"] (t4) ->["\escale{$X_5$}"] (t5)->["\escale{$X_5$}"](t6)->["\escale{$k_{31}$}"](t7);
};
\end{scope}
\begin{scope}[yshift=-9.9cm]

\foreach \i in {0,...,5} {\coordinate (\i) at (\i*1cm,0);}
\foreach \i in {0,...,5} {\node (t\i) at (\i) {\nscale{$t_{5,1,\i}$}};}
\foreach \i in {0,...,5} {\coordinate (s\i) at (\i*1cm,-0.25cm);}
\foreach \i in {0,4} {\node[opacity=0.7](s\i) at (s\i) {\nscale{$[0]$}};}
\foreach \i in {1,2,3,5} {\node[opacity=0.7](s\i) at (s\i) {\nscale{$[2]$}};}
\foreach \i in {0,...,4} {\coordinate (g\i) at (\i*1cm+0.5cm ,-0.25cm);}
\foreach \i in {0,4} {\node[opacity=0.7](g\i) at (g\i) {\nscale{$(0,2)$}};}
\foreach \i in {3} {\node[opacity=0.7](g\i) at (g\i) {\nscale{$(2,0)$}};}
\foreach \i in {1,2} {\node[opacity=0.7](g\i) at (g\i) {\nscale{$(0,0)$}};}
\graph { (t0) ->["\escale{$k_{32}$}"] (t1) ->["\escale{$X_3$}"] (t2) ->["\escale{$X_3$}"] (t3) ->["\escale{$p_5$}"] (t4) ->["\escale{$k_{33}$}"] (t5);
};
\end{scope}
\begin{scope}[xshift=6cm, yshift=-9.9cm]

\foreach \i in {0,...,4} {\coordinate (\i) at (\i*1cm,0);}
\foreach \i in {0,...,4} {\node (t\i) at (\i) {\nscale{$t_{5,2,\i}$}};}
\foreach \i in {0,...,4} {\coordinate (s\i) at (\i*1cm,-0.25cm);}
\foreach \i in {0,3} {\node[opacity=0.7](s\i) at (s\i) {\nscale{$[0]$}};}
\foreach \i in {1,2,4} {\node[opacity=0.7](s\i) at (s\i) {\nscale{$[2]$}};}
\foreach \i in {0,...,3} {\coordinate (g\i) at (\i*1cm+0.5cm ,-0.25cm);}
\foreach \i in {0,3} {\node[opacity=0.7](g\i) at (g\i) {\nscale{$(0,2)$}};}
\foreach \i in {2} {\node[opacity=0.7](g\i) at (g\i) {\nscale{$(2,0)$}};}
\foreach \i in {1} {\node[opacity=0.7](g\i) at (g\i) {\nscale{$(0,0)$}};}
\graph { (t0) ->["\escale{$k_{34}$}"] (t1) ->["\escale{$x_5$}"] (t2) ->["\escale{$p_5$}"] (t3) ->["\escale{$k_{35}$}"] (t4);
};
\end{scope}
\draw [decorate, decoration={brace, amplitude=3pt}]
(-0.5,-9.9)-- (-0.5,0) node [midway, left,  xshift=-0.1cm] {$T$};
\begin{scope}[yshift=-10.8cm]
\foreach \i in {0,...,11} {\coordinate (\i) at (\i*0.95cm,0);}
\foreach \i in {11} {\coordinate (c\i) at (\i*0.95cm,-0.9);}
\foreach \i in {10} {\coordinate (c\i) at (\i*0.95cm,-0.9);}
\foreach \i in {0,...,11} {\node (h\i) at (\i) {\nscale{$h_{2,\i}$}};}
\foreach \i in {11} {\node (c\i) at (c\i) {\nscale{$h_{2,12}$}};}
\foreach \i in {10} {\node (c\i) at (c\i) {\nscale{$h_{2,13}$}};}
\foreach \i in {0,...,11} {\coordinate (s\i) at (\i*0.95cm,-0.25);}
\foreach \i in {11} {\coordinate (s\i) at (\i*0.95cm,0.25);}
\foreach \i in {11} {\coordinate (sc\i) at (\i*0.95cm,-1.25);}
\foreach \i in {10} {\coordinate (sc\i) at (\i*0.95cm,-1.25);}
\foreach \i in {0,5,11} {\node[opacity=0.7] (s\i) at (s\i) {\nscale{$[0]$}};}
\foreach \i in {1,6,10, c11} {\node[opacity=0.7] (s\i) at (s\i) {\nscale{$[1]$}};}
\foreach \i in {2,3,4,7,8,9, c10} {\node[opacity=0.7] (s\i) at (s\i) {\nscale{$[2]$}};}
\foreach \i in {0,...,10} {\coordinate (g\i) at (\i*0.95cm+0.475cm ,-0.25cm);}
\foreach \i in {10} {\node[opacity=0.7] (sc\i) at (\i*0.95cm+0.475cm ,-1.25) {\nscale{$(0,1)$}} ;}
\foreach \i in {2,3,7,8} {\node[opacity=0.7](g\i) at (g\i) {\nscale{$(0,0)$}}    ;}
\foreach \i in {9,10} {\node[opacity=0.7](g\i) at (g\i) {\nscale{$(1,0)$}};}
\foreach \i in {0,1,5,6} {\node[opacity=0.7](g\i) at (g\i) {\nscale{$(0,1)$}};}
\foreach \i in {4} {\node[opacity=0.7](g\i) at (g\i) {\nscale{$(2,0)$}};}
\graph{ 
(h0) ->["\escale{$k$}"] (h1) ->["\escale{$k$}"] (h2) ->["\escale{$z$}"] (h3) ->["\escale{$z$}"] (h4) ->["\escale{$o_0$}"] (h5)->["\escale{$k$}"](h6)->["\escale{$k$}"](h7)->["\escale{$z$}"](h8)->["\escale{$z$}"](h9)->["\escale{$o_1$}"](h10)->["\escale{$o_1$}"](h11)->["\escale{$k$}"](c11)->[swap, "\escale{$k$}"](c10);
};
\end{scope}
\begin{scope}[yshift=-11.7cm]
\foreach \i in {0,...,4} {\coordinate (\i) at (\i*0.95cm,0);}
\foreach \i in {0,...,4} {\node (d\i) at (\i) {\nscale{$d_{0,0,\i}$}};}
\foreach \i in {0,...,4} {\coordinate (s\i) at (\i*0.95cm,-0.25);}
\foreach \i in {0,2} {\node[opacity=0.7] (s\i) at (s\i) {\nscale{$[2]$}};}
\foreach \i in {3} {\node[opacity=0.7] (s\i) at (s\i) {\nscale{$[1]$}};}
\foreach \i in {1,4} {\node[opacity=0.7] (s\i) at (s\i) {\nscale{$[0]$}};}
\foreach \i in {0,...,3} {\coordinate (g\i) at (\i*0.95cm+0.475cm ,-0.25cm);}
\foreach \i in {2,3} {\node[opacity=0.7](g\i) at (g\i) {\nscale{$(1,0)$}};}
\foreach \i in {0} {\node[opacity=0.7](g\i) at (g\i) {\nscale{$(2,0)$}};}
\foreach \i in {1} {\node[opacity=0.7](g\i) at (g\i) {\nscale{$(0,2)$}};}
\graph { (d0) ->["\escale{$o_0$}"] (d1) ->["\escale{$k_0$}"] (d2) ->["\escale{$o_2$}"] (d3)->["\escale{$o_2$}"] (d4);
};
\end{scope}

\begin{scope}[xshift=4.35cm, yshift=-11.7cm]
\node (dots) at (0,0) {\dots};
\end{scope}
\begin{scope}[xshift=4.9cm, yshift=-11.7cm]
\foreach \i in {0,...,4} {\coordinate (\i) at (\i*1cm,0);}
\foreach \i in {0,...,4} {\node (d\i) at (\i) {\nscale{$d_{35,0,\i}$}};}
\foreach \i in {0,...,4} {\coordinate (s\i) at (\i*1cm,-0.25);}
\foreach \i in {0,2} {\node[opacity=0.7] (s\i) at (s\i) {\nscale{$[2]$}};}
\foreach \i in {3} {\node[opacity=0.7] (s\i) at (s\i) {\nscale{$[1]$}};}
\foreach \i in {1} {\node[opacity=0.7] (s\i) at (s\i) {\nscale{$[0]$}};}
\foreach \i in {0,...,3} {\coordinate (g\i) at (\i*1cm+0.5cm ,-0.25cm);}
\foreach \i in {0} {\node[opacity=0.7](g\i) at (g\i) {\nscale{$(2,0)$}};}
\foreach \i in {2,3} {\node[opacity=0.7](g\i) at (g\i) {\nscale{$(1,0)$}};}
\foreach \i in {1} {\node[opacity=0.7](g\i) at (g\i) {\nscale{$(0,2)$}};}
\graph { (d0) ->["\escale{$o_0$}"] (d1) ->["\escale{$k_{35}$}"] (d2) ->["\escale{$o_2$}"] (d3)->["\escale{$o_2$}"] (d4);
};
\end{scope}
\draw [decorate, decoration={brace, amplitude=3pt}]
(-0.5,-11.7)-- (-0.5,-10.8) node [midway, left,  xshift=-0.1cm] {$K_{\tau^2_0}$};

\begin{scope}[yshift=-12.6cm]
\foreach \i in {0,...,11} {\coordinate (\i) at (\i*0.95cm,0);}
\foreach \i in {0,...,11} {\node (h\i) at (\i) {\nscale{$h_{1,\i}$}};}
\foreach \i in {0,...,11} {\coordinate (s\i) at (\i*0.95cm,-0.25);}
\foreach \i in {0,4,9} {\node[opacity=0.7] (s\i) at (s\i) {\nscale{$[0]$}};}
\foreach \i in {1,5,10} {\node[opacity=0.7] (s\i) at (s\i) {\nscale{$[1]$}};}
\foreach \i in {2,3,6,7,8,11} {\node[opacity=0.7] (s\i) at (s\i) {\nscale{$[2]$}};}
\foreach \i in {0,...,10} {\coordinate (g\i) at (\i*0.95cm+0.475cm ,-0.25cm);}
\foreach \i in {0,1,4,5,9,10} {\node[opacity=0.7](g\i) at (g\i) {\nscale{$(0,1)$}};}
\foreach \i in {2,6,7} {\node[opacity=0.7](g\i) at (g\i) {\nscale{$(0,0)$}};}
\foreach \i in {3,8} {\node[opacity=0.7](g\i) at (g\i) {\nscale{$(2,0)$}};}
\graph { (h0) ->["\escale{$k$}"] (h1) ->["\escale{$k$}"] (h2) ->["\escale{$z_0$}"] (h3) ->["\escale{$o_0$}"] (h4) ->["\escale{$k$}"] (h5)->["\escale{$k$}"](h6)->["\escale{$z_1$}"](h7)->["\escale{$z_0$}"](h8)->["\escale{$o_1$}"](h9)->["\escale{$k$}"](h10)->["\escale{$k$}"](h11);
};
\end{scope}

\begin{scope}[yshift=-13.5cm]
\foreach \i in {0,...,3} {\coordinate (\i) at (\i*1.1cm,0);}
\foreach \i in {0,...,3} {\node (d\i) at (\i) {\nscale{$d_{0,1,\i}$}};}
\foreach \i in {0,...,3} {\coordinate (s\i) at (\i*1.1cm,-0.25);}
\foreach \i in {0,2} {\node[opacity=0.7] (s\i) at (s\i) {\nscale{$[2]$}};}
\foreach \i in {1,3} {\node[opacity=0.7] (s\i) at (s\i) {\nscale{$[0]$}};}
\foreach \i in {0,...,2} {\coordinate (g\i) at (\i*1.1cm+0.55cm ,-0.25cm);}
\foreach \i in {0,2} {\node[opacity=0.7](g\i) at (g\i) {\nscale{$(2,0)$}};}
\foreach \i in {1} {\node[opacity=0.7](g\i) at (g\i) {\nscale{$(0,2)$}};}
\graph { (d0) ->["\escale{$o_0$}"] (d1) ->["\escale{$k_0$}"] (d2) ->["\escale{$o_2$}"] (d3);
};
\end{scope}

\begin{scope}[xshift=4cm, yshift=-13.5cm]
\node (dots) at (0,0) {\dots};
\end{scope}
\begin{scope}[xshift=4.8cm, yshift=-13.5cm]
\foreach \i in {0,...,3} {\coordinate (\i) at (\i*1.2cm,0);}
\foreach \i in {0,...,3} {\node (d\i) at (\i) {\nscale{$d_{35,1,\i}$}};}
\foreach \i in {0,...,3} {\coordinate (s\i) at (\i*1.2cm,-0.25);}
\foreach \i in {0,2} {\node[opacity=0.7] (s\i) at (s\i) {\nscale{$[2]$}};}
\foreach \i in {1,3} {\node[opacity=0.7] (s\i) at (s\i) {\nscale{$[0]$}};}
\foreach \i in {0,...,2} {\coordinate (g\i) at (\i*1.2cm+0.6cm ,-0.25cm);}
\foreach \i in {0,2} {\node[opacity=0.7](g\i) at (g\i) {\nscale{$(2,0)$}};}
\foreach \i in {1} {\node[opacity=0.7](g\i) at (g\i) {\nscale{$(0,2)$}};}
\graph { (d0) ->["\escale{$o_0$}"] (d1) ->["\escale{$k_{35}$}"] (d2) ->["\escale{$o_2$}"] (d3);
};
\end{scope}

\begin{scope}[yshift=-14.4cm]
\foreach \i in {0,...,8} {\coordinate (\i) at (\i*0.95cm,0);}
\foreach \i in {0,...,8} {\node (h\i) at (\i) {\nscale{$h_{2,\i}$}};}
\foreach \i in {0,...,11} {\coordinate (s\i) at (\i*0.95cm,-0.25);}
\foreach \i in {0,3,6} {\node[opacity=0.7] (s\i) at (s\i) {\nscale{$[0]$}};}
\foreach \i in {1,4,7} {\node[opacity=0.7] (s\i) at (s\i) {\nscale{$[1]$}};}
\foreach \i in {2,5,8} {\node[opacity=0.7] (s\i) at (s\i) {\nscale{$[2]$}};}
\foreach \i in {0,...,7} {\coordinate (g\i) at (\i*0.95cm+0.475cm ,-0.25cm);}
\foreach \i in {0,1,3,4,6,7} {\node[opacity=0.7](g\i) at (g\i) {\nscale{$(0,1)$}};}
\foreach \i in {2,5} {\node[opacity=0.7](g\i) at (g\i) {\nscale{$(2,0)$}};}
\graph { (h0) ->["\escale{$k$}"] (h1) ->["\escale{$k$}"] (h2) ->["\escale{$o_0$}"] (h3) ->["\escale{$k$}"] (h4) ->["\escale{$k$}"] (h5)->["\escale{$o_2$}"](h6)->["\escale{$k$}"](h7)->["\escale{$k$}"](h8);
};
\end{scope}
\draw [decorate, decoration={brace, amplitude=3pt}]
(-0.5,-13.4)-- (-0.5,-12.6) node [midway, left,  xshift=-0.1cm] {$K_{\tau^2_1}$};
\draw [decorate, decoration={brace, amplitude=3pt}]
(-0.5,-14.4)-- (-0.5,-13.6) node [midway, left,  xshift=-0.1cm] {$K$};
\end{tikzpicture}
\caption{%
Constituents $K_{\tau}, K, T$ for $\tau\in \{\tau^2_0,\tau^2_1\}$ and $\varphi_0$.
TSs are defined by bold drawn states, edges and events.
Labels with reduced opacity correspond to region $(sup, sig)$ defined in Sections~\ref{sec:keys_1}, \ref{sec:translators_1}: 
$sup(s)$ is presented in square brackets below state s and $sig(e)$ is depicted below every $e$-labeled transition.
The model of $\varphi_0$ is $\{X_0, X_4\}$.
}
\label{fig:example_1}
\end{figure}
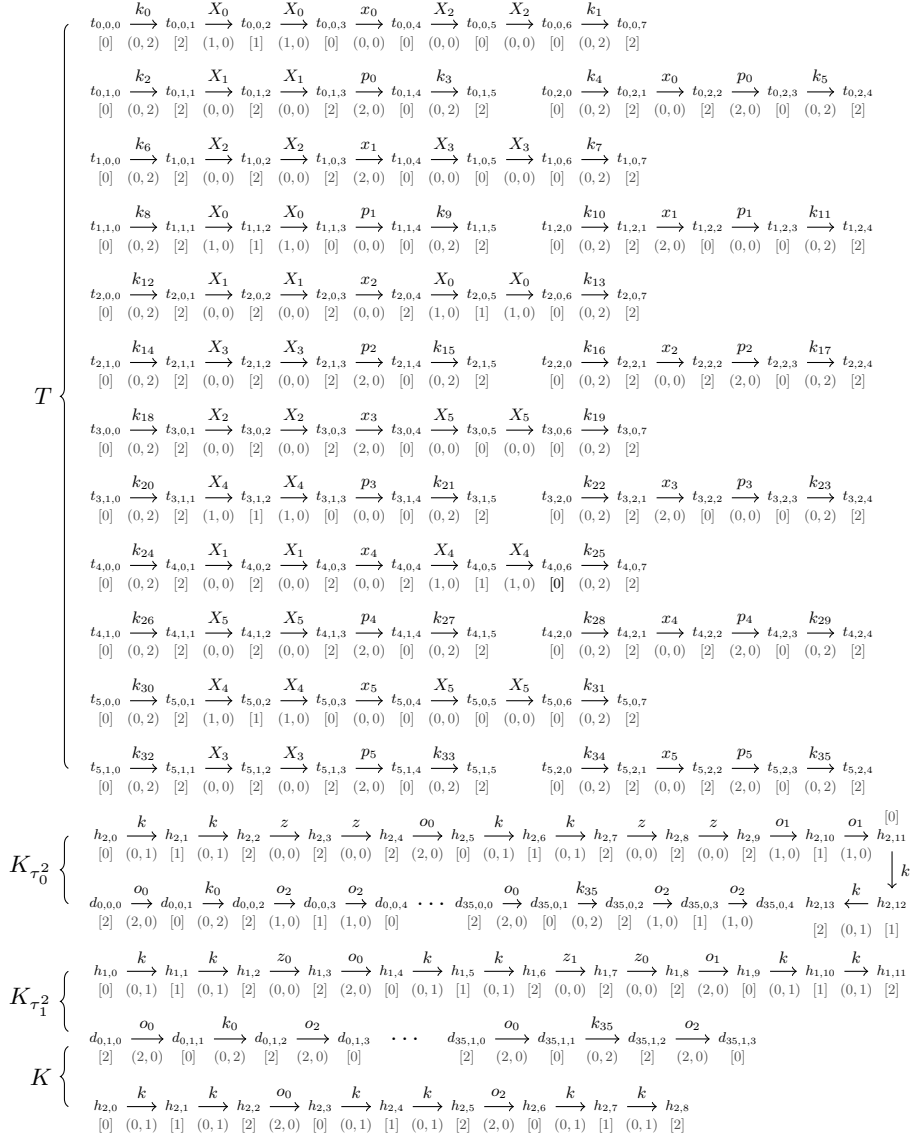

\subsection{The Translators $T_{\tau^b_0}$ and $T_{\tau^b_1}$ and $T$}\label{sec:translators_1}%

In this subsection, we present translator $T$, which we also use as $T_{\tau^b_0}$ and $T_{\tau^b_1}$, that is, $T_{\tau^b_0}=T_{\tau^b_1}=T$. 

For every $i\in \{0,\dots, m-1\}$ the clause $C_i=\{X_{i,0},X_{i,1}, X_{i,2}\}$ is translated into the following three TSs which use the variables of $C_i$ as events:

\noindent
\begin{tikzpicture}[scale=0.9]
\begin{scope}
\node at (-0.75,0) {\scalebox{0.8}{$T_{i,0}=$}};
\node (t0) at (0,0) {\nscale{$t_{i,0,0}$}};
\node (t1) at (1,0) {\nscale{$t_{i,0,1}$}};
\node (t2) at (2,0) {};
\node (h_2_dots) at (2.25,0) {\nscale{$\dots$}};
\node (tb) at (2.4,0) {};
\node (tb+1) at (3.5,0) {\nscale{$t_{i,0,b+1}$}};
\node (tb+2) at (4.9,0) {\nscale{$t_{i,0,b+2}$}};
\node (tb+3) at (6,0) {};
\node (h_k+4_dots) at (6.25,0) {\nscale{$\dots$}};
\node (t2b+1) at (6.4,0) {};
\node (t2b+2) at (7.5,0) {\nscale{$t_{i,0,2b+2}$}};
\node (t2b+3) at (9,0) {\nscale{$t_{i,0,2b+3}$}};
\graph { 
(t0) ->["\escale{$k_{6i}$}"] (t1) ->["\escale{$X_{i,0}$}"] (t2) ;
(tb) ->["\escale{$X_{i,0}$}"] (tb+1) ->["\escale{$x_{i}$}"] (tb+2)->["\escale{$X_{i,2}$}"] (tb+3);
(t2b+1)->["\escale{$X_{i,2}$}"] (t2b+2)->["\escale{$k_{6i+1}$}"] (t2b+3);
;};
\end{scope}
\begin{scope}[yshift = -0.8cm]
\node at (-0.75,0) {\scalebox{0.8}{$T_{i,1}=$}};
\node (t0) at (0,0) {\nscale{$t_{i,1,0}$}};
\node (t1) at (1.2,0) {\nscale{$t_{i,1,1}$}};
\node (t2) at (2.2,0) {};
\node (tdots) at (2.45,0) {\nscale{$\dots$}};
\node (tb) at (2.65,0) {};
\node (tb+1) at (3.7,0) {\nscale{$t_{i,1,b+1}$}};
\node (tb+2) at (5,0) {\nscale{$t_{i,1,b+2}$}};
\node (tb+3) at (6.5,0) {\nscale{$t_{i,1,b+3}$  }};
\graph { (t0) ->["\escale{$k_{6i+2}$}"] (t1) ->["\escale{$X_{i,1}$}"] (t2);
(tb) ->["\escale{$X_{i,1}$}"] (tb+1) ->["\escale{$p_i$}"] (tb+2)->["\escale{$k_{6i+3}$}"] (tb+3);};
\end{scope}
\begin{scope}[yshift = -1.6cm]
\node at (-0.75,0) {\scalebox{0.8}{$T_{i,2}=$}};
\foreach \i in {0,...,4} {\coordinate (\i) at (\i*1.2,0);}
\foreach \i in {0,...,4} {\node (p\i) at (\i) {\nscale{$t_{i,2,\i}$}};}
\graph { (p0) ->["\escale{$k_{6i+4}$}"] (p1) ->["\escale{$x_{i}$}"] (p2) ->["\escale{$p_i$}"] (p3)->["\escale{$k_{6i+5}$}"] (p4);};
\end{scope}

\end{tikzpicture}

\noindent
Altogether, $T=U(T_{0,0}, T_{0,1}, T_{0,2},\dots, T_{m-1,0}, T_{m-1,1}, T_{m-1,2})$.
Figure~\ref{fig:example_1} provides an example for $T$ where $b=2$ and $\varphi=\varphi_0$.
In accordance to our general approach and Lemma~\ref{lem:key_unions_1} the following lemma states the aim of $T$:

\begin{lemma}\label{lem:translator_1}
Let $\tau\in \{\tau^b_0, \tau^b_1\}$.
\begin{enumerate}
\item\label{lem:translator_1_completeness}\emph{(Completeness)}
If $(sup_T, sig_T)$ is a $\tau$-region of $T$ such that $sig_T(k_0)=\dots = sig_T(k_{6m-1})=(0,b)$ or $sig_T(k_0)=\dots = sig_T(k_{6m-1})=(b,0)$ then $\varphi$ has a one-in-three model.

\item\label{lem:translator_1_existence}\emph{(Existence)}
If $\varphi$ has a one-in-three model then there is a $\tau$-region $(sup_T, sig_T)$ of $T$ such that $sig_T(k_0)=\dots = sig_T(k_{6m-1})=(0,b)$.
\end{enumerate}
\end{lemma}

\begin{proof}
To fulfill its destiny, $T$ works as follows.
By definition, if $(sup_T, sig_T)$ is a region of $T$ then $\pi_{i,0},\pi_{i,1},\pi_{i,2}$, defined by

\noindent
\begin{tikzpicture}[scale=0.89]
\begin{scope}
\node at (-1,0) {\scalebox{0.8}{$\pi_{i,0}=$}};

\node (t1) at (0,0) {\nscale{$sup_T(t_{i,0,1})$}};
\node (t2) at (2,0) {};
\node (h_2_dots) at (2.25,0) {\nscale{$\dots$}};
\node (tb) at (2.5,0) {};
\node (tb+1) at (4.5,0) {\nscale{$sup_T(t_{i,0,b+1})$}};
\node (tb+2) at (7,0) {\nscale{$sup_T(t_{i,0,b+2})$}};
\node (tb+3) at (9,0) {};
\node (h_k+4_dots) at (9.25,0) {\nscale{$\dots$}};
\node (t2b+1) at (9.5,0) {};
\node (t2b+2) at (11.5,0) {\nscale{$sup_T(t_{i,0,2b+2})$}};

\graph { 
(t1) ->["\escale{$sig_T(X_{i,0}$})"] (t2) ;
(tb) ->["\escale{$sig_T(X_{i,0})$}"] (tb+1) ->["\escale{$sig_T(x_{i})$}"] (tb+2)->["\escale{$sig_T(X_{i,2})$}"] (tb+3);
(t2b+1)->["\escale{$sig_T(X_{i,2})$}"] (t2b+2);
;};
\end{scope}
\begin{scope}[yshift = -0.8cm]
\node at (-1,0) {\scalebox{0.8}{$\pi_{i,1}=$}};
\node (t1) at (0,0) {\nscale{$sup_T(t_{i,1,1})$}};
\node (t2) at (2,0) {};
\node (tdots) at (2.25,0) {\nscale{$\dots$}};
\node (tb) at (2.5,0) {};
\node (tb+1) at (4.5,0) {\nscale{$sup_T(t_{i,1,b+1})$}};
\node (tb+2) at (6.5,0) {\nscale{$sup_T(t_{i,1,b+2})$}};

\graph { (t1) ->["\escale{$sig_T(X_{i,1}$})"] (t2);
(tb) ->["\escale{$sig_T(X_{i,1})$}"] (tb+1) ->["\escale{$sig_T(p_i$})"] (tb+2);};
\end{scope}

\begin{scope}[yshift = -1.6cm]
\node at (-1,0) {\scalebox{0.8}{$\pi_{i,2}=$}};
\foreach \i in {1,...,3} {\coordinate (\i) at (\i*2-2,0);}
\foreach \i in {1,...,3} {\node (p\i) at (\i) {\nscale{$sup_T(t_{i,2,\i})$}};}
\graph {  (p1) ->["\escale{$sig_T(x_{i})$}"] (p2) ->["\escale{$sig_T(p_i)$}"] (p3);};
\end{scope}

\end{tikzpicture}
\newline
are directed labeled paths of $\tau$.
For every $i\in \{0,\dots, m-1\}$, the events $k_{6i}, \dots, k_{6i+5}$ belong to the interface.
By Lemma~\ref{lem:key_unions_1}.\ref{lem:key_unions_1_completeness}, $K_\tau$ and $K$ ensure the following:
If $(sup_K, sig_K)$ is a region of $K_\tau$, respectively $K$, that solves the key atom $a_\tau$, respectively $\alpha$, then either $sig_K(k_0)=\dots =sig_K(k_{6m-1})=(0,b)$ or $sig_K(k_0)=\dots =sig_K(k_{6m-1})=(b, 0)$.
For every transition $s\edge{k_j}s'$, the first case implies $sup(s)=0$ and $sup(s')=b$ while the second case implies $sup(s)=b$ and $sup(s')=0$, where $j\in \{0,\dots, 6m-1\}$.
Hence, a $\tau$-region $(sup_T, sig_T)$ of $T$ being compatible with $(sup_K, sig_K)$ satisfies exactly one of the next conditions:
\begin{enumerate}
\item[(1)] \label{item:T_from_b_to_0}
$sig_T(k_0)=\dots =sig_T(k_{6m-1})=(0,b)$ and for every $i\in \{0,\dots, m-1\}$ the paths $\pi_{i,0},\pi_{i,1},\pi_{i,2}$ start at $b$ and terminate at $0$.
\item[(2)]\label{item:T_from_0_to_b} 
$sig_T(k_0)=\dots =sig_T(k_{6m-1})=(b,0)$ and for every $i\in \{0,\dots, m-1\}$ the paths $\pi_{i,0},\pi_{i,1},\pi_{i,2}$ start at $0$ and terminate at $b$.
\end{enumerate}
The construction of $T$ ensures that if (1), respectively if (2), is satisfied then there is for every $i\in \{0,\dots, m-1\}$ exactly one variable event $X\in \{X_{i,0}, X_{i,1}, X_{i,2}\}$ such that $sig(X)=(1,0)$, respectively $sig(X)=(0,1)$.
Each triple $T_{i,0}, T_{i,1}, T_{i,2}$ corresponds exactly to the clause $C_i$. 
Hence, $M=\{X\in V(\varphi) \vert sig_T(X)=(1,0)\}$ or $M=\{X\in V(\varphi) \vert sig_T(X)=(0,1)\}$, is a one-in-three model of $\varphi$, respectively.
Having sketched the plan to satisfy Lemma~\ref{lem:translator_1}.\ref{lem:translator_1_completeness}, it remains to argue that the deduced conditions (1), (2) have the announced impact on the variable events.

For a start, let (2) be satisfied and $i\in \{0,\dots, m-1\}$.
By $sig_T(k_{6i})=\dots=sig_T(k_{6i+5})=(0,b)$ we have that $sup_T(t_{i,0,1})=sup_T(t_{i,1,1})=sup_T(t_{i,1,1})=b$ and $sup_T(t_{i,0,2b+2})=sup_T(t_{i,1,b+2})=sup_T(t_{i,1,3})=0$.
Notice, for every event $e\in \{X_{i,0}, X_{i,1}, X_{i,2}, x_i, p_i\}$ there is a state $s$ such that $s\edge{e}$ and $sup_T(s)=b$ or such that $\edge{e}s$ and $sup_T(s)=0$.
Consequently, if $(m,n)\in E_\tau$ and $m < n$ then $sig(e)\not=(m,n)$.
This implies the following condition:
\begin{enumerate}
\item[(3)]\label{item:T_greater_or_equal_sup} 
If $e\in \{X_{i,0}, X_{i,1}, X_{i,2}, x_i, p_i\}$ and $s\edge{e}s'$ then $sup_T(s)\geq sup_T(s')$.
\end{enumerate}
Moreover, every variable event $X_{i,0}, X_{i,1}, X_{i,2}$ occurs $b$ times consecutively in a row. 
Hence, by Lemma~\ref{lem:observations}, we have:
\begin{enumerate}

\item[(4)]\label{item:T_state_changing_sig_is_from_1_to_0} 
If $X\in \{X_{i,0}, X_{i,1}, X_{i,2}\}$, $sig_T(X)=(m,n)$ and $m\not=n$ then $(m,n)=(1,0)$.
\end{enumerate}
The paths $\pi_{i,0}, \pi_{i,1}, \pi_{i,2}$ of $\tau$ start at $b$ and terminate at $0$.
Hence, by definition of $\tau$, for every $\pi\in \{ \pi_{i,0}, \pi_{i,1}, \pi_{i,2} \}$ there has to be an event $e_\pi$, which occurs at $\pi$, such that $sig_T(e_\pi)=(m,n)$ with $m > n$.

If for $\pi\in \{ \pi_{i,0}, \pi_{i,1}\}$ it is true that $e_{\pi}\not\in\{X_{i,0}, X_{i,1}, X_{i,2}\}$ then for $X\in \{X_{i,0}, X_{i,1}, X_{i,2}\}$ we have $sig_T(X)=(m,m)$ for some $m\in \{0,\dots, b\}$.
This yields $sup(t_{i,0,b+1})=sup(t_{i,1,b+1})=b$ and $sup(t_{i,0,b+2})=0$ which with $sup(t_{i,1,b+2})=0$ implies $sig_T(x_i)=sig_T(p_i)=(b,0)$.
By $sig_T(x_i)=(b,0)$, we obtain $sup(t_{i,2,2})=0$ and, by  $sig_T(p_i)=(b,0)$, we obtain $sup(t_{i,2,2})=b$, a contradiction.
Consequently, by Condition~$4$, there has to be an event $X\in \{X_{i,0}, X_{i,1}, X_{i,2}\}$ such that $sig_T(X)=(1,0)$.
We discuss all possible cases to show that $X$ is unambiguous.

If $sig_T(X_{i,0})=(1,0)$ then, by Lemma~\ref{lem:observations}, we have that $sup_T(t_{i,0,b+1}) = 0$.
By (3), this implies that $sup_T(t_{i,0,b+2}) = \dots = sup_T(t_{i,0,2b+1})  = 0 $ and $sig_T(x_i)=sig_T(X_{i,2})=(0,0)$.
Moreover, $sig_T(x_i)=(0,0)$ and $sup(t_{i,2,1})=b$ imply $sup(t_{i,2,2})=b$ which with $sup(t_{i,2,3})=0$ implies $sig_T(p_i)=(b,0)$.
By $sig_T(p_i)=(b,0)$ we obtain $sup(t_{i,1,b+1})=b$ which, by Lemma~\ref{lem:observations} and contraposition shows that $sig_T(X_{i,1})\not=(1,0)$.
Hence, we have $sig_T(X_{i,1})\not=(1,0)$.

If $sig_T(X_{i,2})=(1,0)$ then, by Lemma~\ref{lem:observations}, we have that $sup_T(t_{i,0,b+2}) = b$.
Again by (3), this implies that $sup_T(t_{i,0,1}) = \dots = sup_T(t_{i,0,b+2})  = b $ and $sig_T(x_i)=(m,m)$, $sig_T(X_{i,0})=(m',m')$ for some $m,m'\in \{0,\dots, b\}$.
Especially, we have that $sig_T(X_{i,0})\not=(1,0)$.
Moreover, by $sig_T(x_i)=(m,m)$, we obtain $sup_T(t_{i,2,2})=b$ implying with $sup_T(t_{i,2,3})=0$ that $sig_T(p_i)=(b,0)$.
As in the previous case this yields $sig_T(X_{i,1})\not=(1,0)$.

Finally, if $sig_T(X_{i,1})=(1,0)$ then, by Lemma~\ref{lem:observations}, we get $sup_T(t_{i,1,b+1}) = 0$.
By $sup_T(t_{i,1,b+1}) = sup_T(t_{i,1,b+2})= 0$ we conclude $sig_T(p_i)=(0,0)$ which with $sup_T(t_{i,2,3})=0$ implies $sup_T(t_{i,2,2})=0$.
Using $sup_T(t_{i,2,1})=b$ and $sup_T(t_{i,2,2})=0$ we obtain $sig_T(x_i)=(b,0)$ implying that $sup_T(t_{i,0,b+1})=b$ and $sup_T(t_{i,0,b+2})=0$.
By (3), this yields $sup_T(t_{i,0,1})=\dots =sup_T(t_{i,0,b+1})=b$ and $sup_T(t_{i,0,b+2})=\dots =sup_T(t_{i,0,2b+2})=b$ which, by Lemma~\ref{lem:observations}, implies $sig_T(X_{i,0})\not=(1,0)$ and $sig_T(X_{i,2})\not=(1,0)$. 

So far, we have proven that if (1) is satisfied then for every $i\in \{0,\dots, m-1\}$ there is exactly one variable event $X\in \{X_{i,0}, X_{i,1}, X_{i,2}\}$ such that $sig_T(X)=(1,0)$.
Consequently, the set $M=\{X\in V(\varphi) \vert sig_T(X)=(1,0)\}$ is a one-in-three model of $\varphi$.
One verifies, by analogous arguments, that (2) implies for every $i\in \{0,\dots, m-1\}$ that there is exactly one variable event $X\in \{X_{i,0}, X_{i,1}, X_{i,2}\}$ with $sig_T(X)=(0,1)$, which makes $M=\{X\in V(\varphi) \vert sig_T(X)=(0,1)\}$ a one-in-three model of $\varphi$.
Hence, a $\tau$-region of $T_\tau$ that satisfies (1) or (2) implies a one-in-three model of $\varphi$.

Reversely, if $M$ is a one-in-three model of $\varphi$ then there is a $\tau$-region $(sup_T, sig_T)$ satisfying (1) which, by Lemma~\ref{lem:observations}, is completely defined by $sup_T(t_{i,0,0})=sup_T(t_{i,1,0})=sup_T(t_{i,1,0})=0$ for $i\in \{0,\dots, m-1\}$ and 

\[sig_T(e)=
\begin{cases}
(0,b), & \text{if } e\in \{k_{0},\dots, k_{6m-1}\} \\
(0,0), & \text{if } e\in V(\varphi)\setminus M\\
(0,0), &  \text{if } (e=p_i, X_{i,1}\in M) \text{ or } (e= x_i , X_{i,1} \not\in M), 0\leq i\leq m-1 \\
(1,0), & \text{if } e \in M \\
(b,0), & \text{if } (e = x_i ,  X_{i,1} \in M) \text{ or }  (e = p_i , X_{i,1} \not\in M), 0\leq i\leq m-1 \\
\end{cases}
\]
See Figure~\ref{fig:example_1}, for a sketch of this region for $\tau\in \{\tau^2_0, \tau^2_1\}$, $\varphi_0$ and $M=\{X_0, X_4\}$.
This proves Lemma~\ref{lem:translator_1}.
\end{proof}

\subsection{The Key Unions $K_{\tau^b_2}$ and $K_{\tau^b_3}$ }\label{sec:group_extensions_keys} %

The unions $U_{\tau^b_2}, U_{\tau^b_3}$ install the same key.
More exactly, if $\tau\in \{\tau^b_2,\tau^b_3\}$ then $K_\tau$ uses only the TS $H_3$ to provide key atom $(k, h_{3,1,b-1})$ and the interface $k$ and $z$:

\noindent
\begin{tikzpicture}
\node (init) at (-0.75,0) {$H_3=$};
\node (h0) at (0,0) {\nscale{$h_{3,0,0}$}};
\node (h1) at (1,0) {\nscale{}};
\node (h_2_dots) at (1.25,0) {\nscale{$\dots$}};
\node (h_b_2) at (1.5,0) {};
\node (h_b_1) at (2.6,0) {\nscale{$h_{3,0,b-1}$}};
\node (h_b) at (4.0,0) {\nscale{$h_{3,0,b}$}};
\node (h_b+1) at (0,-1) {\nscale{$h_{3,1,0}$}};
\node (h_b+2) at (1,-1) {};
\node (h_b+2_dots) at (1.25,-1) {\nscale{$\dots$}};
\node (h_2b_3) at (1.5,-1) {};
\node (h_2b_2) at (2.6,-1) {\nscale{$h_{3,1,b-1}$}};
\graph{ 
(h0) ->["\escale{$k$}"] (h1); 
(h0) ->["\escale{$u$}", swap](h_b+1)->["\escale{$k$}"](h_b+2);
(h_b_2)->["\escale{$k$}"] (h_b_1)->["\escale{$k$}"] (h_b);
(h_2b_3)->["\escale{$k$}"] (h_2b_2);
(h_2b_2)->[swap, "\escale{$z$}"] (h_b);
};
\end{tikzpicture}

\noindent
The next lemma summarizes the intention behind $K_{\tau}$:
\begin{lemma}\label{lem:key_unions_2}
Let $\tau\in \{\tau^b_2, \tau^b_3\}$ and $E_0=\{(m,m) \vert 1 \leq m\leq b\}\cup \{  0 \}$.
\begin{enumerate}
\item\label{lem:key_unions_2_completeness}\emph{(Completeness)}
If $(sup_K, sig_K)$ is a $\tau$-region that solves $(k, h_{3,1,b-1})$ in $K_\tau$ then $sig(k)\in \{(1,0), (0,1)\}$ and $sig_K(z)\in E_0$.
\item\label{lem:key_unions_2_existence}\emph{(Existence)}
There is a $\tau$-region $(sup_K, sig_K)$ of $K_\tau$ solving $(k, h_{3,1,b-1})$ such that $sig(k)=(0,1)$ and $sig_K(z)=0$.
\end{enumerate}
\end{lemma}

\begin{proof}
For the first statement, we let $(sup_K, sig_K)$ be a region solving $\alpha_\tau$.
By $\edge{k}h_{3,1,b-1}$ and $\neg sup_K(h_{3,1,b-1})\edge{sig_K(k)}$ we immediately have $sig(K)\not\in E_0$.
Moreover, for every group event $e\in \{0,\dots, b\}$ and every state $s$ of $\tau$ we have that $s\edge{e}$.
Hence, by $\neg sup_K(h_{3,1,b-1})\edge{sig_K(k)}$ we have $sig_K(k)\not\in \{0,\dots, b\}$.
The event $k$ occurs $b$ times in a row. 
Therefore, by Lemma~\ref{lem:observations}, we have that $sig_K(k)\in \{(1,0), (0,1)\}$ and if $sig_K(k)= (1,0)$ then $sup_K(h_{3,0,b})=0$ and if $sig_K(k)= (0,1)$ then $sup_K(h_{3,0,b})=b$.
If $s\in \{0,\dots, b-1\}$ then $s\edge{(0,1)}$ is true.
Furthermore, every state $s\in \{1,\dots, b\}$ satisfies $s\edge{(1,0)}$.
Consequently, by $\neg sup_K(h_{3,1,b-1})\edge{sig_K(k)}$, if $sig_K(k)=(0,1)$ then $sup_K(h_{3,1,b-1})=b$ and if $sig_K(k)=(1,0)$ then $sup_K(h_{3,1,b-1})=0$.
This implies for $(sup_K, sig_K)$ that $sig_K(z)\in E_0$ and proves Lemma~\ref{lem:key_unions_2}.\ref{lem:key_unions_2_completeness}.
For Lemma~\ref{lem:key_unions_2}.\ref{lem:key_unions_2_existence} we easily verify that $(sup_K, sig_K)$ with $sig_K(k)=(0,1)$, $sig_K(u)=1$, $sig_K(z)=0$ and $sup_K(h_{3,0,0})=0$ properly defines a solving $\tau$-region.
\end{proof}

\subsection{The Translators $T_{\tau^b_2}$ and $T_{\tau^b_3}$}\label{sec:group_extensions_translators}%

In this section we introduce $T_{\tau^b_2}$ which is used for $U_{\tau^b_2}$ and $U_{\tau^b_3}$, that is, $T_{\tau^b_3}=T_{\tau^b_2}$. 
Let $\tau\in\{\tau^b_2, \tau^b_3\}$.
Firstly, the translator $T_{\tau}$ contains for every variable $X_j$ of $\varphi$, $j\in\{0,\dots, m-1\}$, the TSs $F_j, G_j$ below, that apply $X_j$ as  event:

\noindent
\begin{tikzpicture}
\begin{scope}
\node at (-0.75,0) {$F_j=$};
\node (f0) at (0,0) {\nscale{$f_{j,0,0}$}};
\node (f1) at (1,0) {\nscale{}};
\node (f_2_dots) at (1.25,0) {\nscale{$\dots$}};
\node (f_b_1) at (1.5,0) {};
\node (f_b) at (2.6,0) {\nscale{$f_{j,0,b}$}};
\node (f_b+1) at (0,-1) {\nscale{$f_{j,1,0}$}};
\node (f_b+2) at (1,-1) {};
\node (f_b+2_dots) at (1.25,-1) {\nscale{$\dots$}};
\node (f_2b_1) at (1.5,-1) {};
\node (f_2b) at (2.6,-1) {\nscale{$f_{j,1,b-1}$}};
\graph{ 
(f0) ->["\escale{$k$}"] (f1); 
(f0) ->["\escale{$v_j$}", swap](f_b+1)->["\escale{$k$}"](f_b+2);
(f_b_1)->["\escale{$k$}"] (f_b);
(f_2b_1)->["\escale{$k$}"] (f_2b);
(f_2b)->[swap, "\escale{$X_{j}$}"] (f_b);
};
\end{scope}
\begin{scope}[xshift= 5cm]
\node at (-0.75,0) {$G_j=$};
\node (f0) at (0,0) {\nscale{$g_{j,0}$}};
\node (f1) at (1,0) {\nscale{}};
\node (f_2_dots) at (1.25,0) {\nscale{$\dots$}};
\node (f_b_1) at (1.5,0) {};
\node (f_b) at (2.6,0) {\nscale{$g_{j,b}$}};
\node (f_b+1) at (3.8,0) {\nscale{$g_{j,b+1}$}};
\graph{ 
(f0) ->["\escale{$k$}"] (f1); 
(f_b_1)->["\escale{$k$}"] (f_b)->["\escale{$X_j$}"](f_b+1);
};
\end{scope}

\end{tikzpicture}

\noindent
Secondly, translator $T_{\tau}$ implements for every clause $C_i=\{X_{i,0}, X_{i,1}, X_{i,2}\}$ of $\varphi$, $i\in \{0,\dots, m-1\}$, the following TS $T_i$ that applies the variables of $C_i$ as events :

\noindent
\begin{tikzpicture}
\node at (-0.75,0) {$T_i=$};
\node (t0) at (0,0) {\nscale{$t_{i,0}$}};
\node (t1) at (1,0) {\nscale{}};
\node (t_2_dots) at (1.25,0) {\nscale{$\dots$}};
\node (t_b_1) at (1.5,0) {};
\node (t_b) at (2.5,0) {\nscale{$t_{i,b}$}};
\node (t_b+1) at (3.7,0) {\nscale{$t_{i,b+1}$}};
\node (t_b+2) at (4.9,0) {\nscale{$t_{i,b+2}$}};
\node (t_b+3) at (6.1,0) {\nscale{$t_{i,b+3}$}};
\node (t_b+4) at (7.3,0) {\nscale{$t_{i,b+4}$}};
\node (t_b+5) at (8.5,0) {};
\node (t_b+5_dots) at (8.75,0) {\nscale{$\dots$}};
\node (t_2b+3) at (9,0) {\nscale{}};
\node (t_2b+4) at (10.2,0) {\nscale{$t_{i, 2b+4}$}};
\graph{ 
(t0) ->["\escale{$k$}"] (t1); 
(t_b_1)->["\escale{$k$}"] (t_b) ->["\escale{$X_{i,0}$}"] (t_b+1)->["\escale{$X_{i,1}$}"] (t_b+2)->["\escale{$X_{i,2}$}"] (t_b+3)->["\escale{$z$}"] (t_b+4)->["\escale{$k$}"] (t_b+5);
(t_2b+3)->["\escale{$k$}"] (t_2b+4);
};
\end{tikzpicture} 

\noindent
Altogether, we have $T_{\tau}=(F_0,G_0,\dots, F_{m-1}, G_{m-1}, T_0,\dots,T_{m-1})$.

The next lemma summarizes the functionality of $T_{\tau}$:
\begin{lemma}\label{lem:translator_2}
If $\tau\in \{\tau^b_2, \tau^b_3\}$ then the following conditions are true:
\begin{enumerate}
\item\label{lem:translator_2_completeness}\emph{(Completeness)}
If $(sup_T, sig_T)$ is a $\tau$-region of $T_\tau$ such that $sig_T(z)\in E_0$ and $sig_T(k)=(0,1)$, respectively $sig_T(k)=(1,0)$, then $\varphi$ is one-and-three satisfiable.

\item\label{lem:translator_2_existence}\emph{(Existence)}
If $\varphi$ has a one-in-three model $M$ then there is a $\tau$-region $(sup_T, sig_T)$ of $T_\tau$ such that $sig_T(z)=0$ and $sig_T(k)=(0,1)$.
\end{enumerate}
\end{lemma}

\begin{proof}
Firstly, we argue for Lemma~\ref{lem:translator_2}.\ref{lem:translator_2_completeness}.
Let $(sup_T, sig_T)$ be a region of $T_\tau$ which satisfies $sig_T(z)\in E_0, sig_T(k)\in \{(1,0), (0,1)\}$.
By definition, $\pi_i$ defined by 

\noindent
\begin{tikzpicture}
\node (init) at (-1,0) {$\pi_i=$};
\node (t_b) at (0,0) {\nscale{$sup_T(t_{i,b})$}};
\node (t_b+1) at (2.3,0) {\nscale{$sup_T(t_{i,b+1})$}};
\node (t_b+2) at (4.6,0) {\nscale{$sup_T(t_{i,b+2})$}};
\node (t_b+3) at (6.9,0) {\nscale{$sup_T(t_{i,b+3})$}};

\graph{ 
 (t_b) ->["\escale{$sig_T(X_{i,0})$}"] (t_b+1)->["\escale{$sig_T(X_{i,1})$}"] (t_b+2)->["\escale{$sig_T(X_{i,2})$}"] (t_b+3);
 };
\end{tikzpicture} 

\noindent
is a directed labeled path in $\tau$.
By $sig_T(z)\in E_0$ and $t_{i,b+3}\edge{z}t_{i,b+4}$ we obtain that $sup_T(t_{i,b+3})=sup_T(t_{i,b+4})$.
Moreover, $k$ occurs $b$ times in a row at $t_{i,0}$ and $t_{i,b+4}$.
By Lemma~\ref{lem:observations}, this implies if $sig_T(k)=(1,0)$ then $sup_T(t_{i,b})=b$ and $sup_T(t_{i,b+4})=0$ and if $sig_T(k)=(0,1)$ then $sup_T(t_{i,b})=0$ and $sup_T(t_{i,b+4})=b$.
Altogether, we obtain that the following conditions are true:
If $sig_T(z)\in E_0, sig_T(k) = (1,0)$ then path $p_i$ starts a $0$ and terminates at $b$ and if $sig_T(z)\in E_0, sig_T(k) = (0,1)$ then the path $p_i$ starts a $b$ and terminates at $0$.

By definition of $\tau$, both conditions imply that there has to be at least one event $X\in \{X_{i,0}, X_{i,1}, X_{i,2}\}$ whose signature satisfies $sig_T(X)\not\in E_0$.
Again, our intention is to ensure that for exactly one such variable event the condition $sig_T(X)\not\in E_0$ is true.
Here, the TSs $F_0,G_0,\dots, F_{m-1},G_{m-1} $ come into play.
The aim of $F_0,G_0,\dots, F_{m-1},G_{m-1} $ is to restrict the possible signatures for the variable events as follows:
If $sig_T(k) = (1,0)$ then $X\in V(\varphi)$ implies $sig_T(X)\in E_0 \cup \{ b \}$ and if $sig_T(k) = (0,1)$ then $X\in V(\varphi)$ implies $sig_T(X)\in E_0 \cup \{ 1 \}$.

We now argue, that the introduced conditions ensure that there is exactly one variable event $X\in \{X_{i,0}, X_{i,1}, X_{i,2}\}$ with $sig_T(X)\not\in E_0$.
Remember that, by definition, if $sig_T(X)\in E_0$ then $sig^-_T(X) +sig^+_T(X) = \vert sig_T(X)\vert = 0$.

For a start, let $sig_T(z)\in E_0, sig_T(k) = (1,0)$, implying that $p_i$ starts at $b$ and terminates at $0$, and assume $sig_T(X)\in E_0 \cup \{ b \}$.
By Lemma~\ref{lem:observations}, we obtain:
\begin{equation}\label{eq:modulo=b}
(\vert sig_T(X_{i,0})\vert + \vert sig_T( X_{i,1} ) \vert +\vert sig_T( X_{i,2}) \vert)  \equiv b \text{ mod } (b+1) 
\end{equation}

Clearly, if $sig_T(X_{i,0}), sig_T(X_{i,1}), sig_T(X_{i,2})\in E_0$, then we obtain a contradiction to (\ref{eq:modulo=b}) by $\vert sig_T(X_{i,0})\vert = \vert sig_T( X_{i,1} ) \vert =\vert sig_T( X_{i,2}) \vert=0$.
Hence, there has to be at least one variable event $X\in \{ X_{i,0}, X_{i,1} , X_{i,2}  \}$ with $sig_T(X)=b$.

If there are two different variable events $X, Y\in \{ X_{i,0}, X_{i,1} , X_{i,2}  \}$ such that $sig_T(X)=sig_T(Y)=b$ and $sig_T(Z)\in E_0$ for  $Z \in \{ X_{i,0}, X_{i,1} , X_{i,2}  \}\setminus \{X, Y\}$ then, by symmetry and transitivity, we obtain: 
\begin{align}
& b \equiv  (\vert sig_T(X_{i,0})\vert + \vert sig_T( X_{i,1} ) \vert +\vert sig_T( X_{i,2}) \vert)  \text{ mod } (b+1)  && \vert (1) \\ 
& (\vert sig_T(X_{i,0})\vert + \vert sig_T( X_{i,1} ) \vert +\vert sig_T( X_{i,2}) \vert)  \equiv 2b \text{ mod } (b+1)  && \vert \text{assumpt.} \\
& b  \equiv 2b  \text{ mod } (b+1)  && \vert (2),(3) \\ 
& 2b  \equiv (b-1)  \text{ mod } (b+1)  && \vert \text{def. }  \equiv  \\ 
& b  \equiv  (b-1) \text{ mod } (b+1)  &&\vert  (4),(5)\\
& \exists m\in \mathbb{Z}: m(b+1)=1 && \vert  (6)
\end{align}
By $(7)$ we obtain $b=0$, a contradiction.
Similarly, if we assume that $\vert sig_T(X_{i,0})\vert = \vert sig_T( X_{i,1} ) \vert =\vert sig_T( X_{i,2}) \vert=b$ then we obtain 
\begin{align}
& (\vert sig_T(X_{i,0})\vert + \vert sig_T( X_{i,1} ) \vert +\vert sig_T( X_{i,2}) \vert)  \equiv 3b \text{ mod } (b+1)  && \vert \text{assumpt.} \\
& b  \equiv 3b  \text{ mod } (b+1)  && \vert (2),(8) \\ 
& 3b  \equiv (b-2)  \text{ mod } (b+1)  && \vert \text{def. }  \equiv \\ 
& b  \equiv  (b-2) \text{ mod } (b+1)  &&\vert  (9),(10)\\
& \exists m\in \mathbb{Z}: m(b+1)=2 && \vert (11) 
\end{align}
By $(12)$, we have $b\in \{0,1\}$ which contradicts $b\geq 2$.
Consequently, if $sig_T(z)\in E_0$ and $sig_T(k) = (1,0)$ and $sig_T(X)\in E_0 \cup \{ b \}$ then there is exactly one variable event $X\in \{X_{i,0}, X_{i,1}, X_{i,2}\}$ with $sig_T(X)\not\in E_0$.

If we continue with $sig_T(z)\in E_0$, $sig_T(k) = (0,1)$ and $sig_T(X)\in E_0 \cup \{ 1 \}$ then we find the following equation to be true:
\begin{equation}\label{eq:modulo=0}
(\vert sig_T(X_{i,0})\vert + \vert sig_T( X_{i,1} ) \vert +\vert sig_T( X_{i,2}) \vert)  \equiv 0 \text{ mod } (b+1) 
\end{equation}
Analogously to the former case one argues that the assumption that not exactly one variable event  $X\in \{X_{i,0}, X_{i,1}, X_{i,2}\}$ is equipped with the signature $1$, that is, $sig_T(X)\not\in E_0$, leads to the contradiction $b\in \{0,1\}$.
Altogether, we have shown that if $(sup_T, sig_T)$ is a region such that $sig_T(k)\in \{(0,1), (1,0)\}$ and $sig_T(z)\in E_0$ and if the TSs $F_0,G_0,\dots, F_{m-1}, G_{m-1}$ do as announced then there is exactly one variable event $X\in \{X_{i,0}, X_{i,1}, X_{i,2}\}$ for every $i\in \{0,\dots, m-1\}$ such that $sig_T(X)\not\in E_0$.
By other words, in that case we have that the set $M=\{X\in V(\varphi) \vert sig_T(X) \not\in E_0\}$ defines a one-in-three model of $\varphi$.

Hence, to complete the arguments for Lemma~\ref{lem:translator_2}.\ref{lem:translator_2_completeness}, it remains to argue for the announced functionality of $F_0,G_0,\dots, F_{m-1},G_{m-1} $.
Let $j\in \{0,\dots, m-1\}$.
We argue for $X_j$ that if $sup_T(k)=(1,0)$ then $sup_T(X_j)\in E_0\cup \{b\}$ and if $sup_T(k)=(0,1)$ then $sup_T(X_j)\in E_0\cup \{1\}$, respectively.

To begin with, let $sig_T(k)=(1,0)$.
The event $k$ occurs $b$ times in a row at $f_{j,0,0}$ and $g_{j,0}$ and $b-1$ times in a row at $f_{j,1,0}$.
By Lemma~\ref{lem:observations} this implies $sup_T(f_{j,0,b})=sup_T(g_{j,b})=0$ and $sup_T(f_{j,1,b-1})\in \{0,1\}$.
Clearly, if $sup_T(f_{j,0,b})=sup_T(f_{j,1,b-1})=0$ then $sig_T(X_j)\in E_0$.
We argue, $sup_T(f_{j,1,b-1})=1$ implies $sig_T(X_j) = b$.

Assume, for a contradiction, that $sig_T(X_j)\not=b $.
If $sig_T(X_j)=(m,m)$ for some $m\in \{1,\dots, b\}$ then $-sig^-_T(X_j)+sig^+_T(X_j)=\vert sig_T(X_j) \vert = 0$.
By Lemma~\ref{lem:observations} this contradicts $sup_T(f_{j,0,b})\not=sup_T(f_{j,1,b-1})$.
If $sig_T(X_j)=(m,n)$ with $m\not=n$ then the $\vert sig_T(X_j)\vert =0$.
By Lemma~\ref{lem:observations}, we have $sup_T(f_{j,0,b})=sup_T(f_{j,1,b-1})-sig^-_T(X_j)+sig^+_T(X_j)$ implying $sig_T(X_j)=(1,0)$.
But, by $sup_T(g_{j,b})=0$ and $\neg 0 \edge{(1,0)}$ in $\tau$, this contradicts $sup_T(g_{j,b})\edge{sig_T(X_j)}$.
Finally, if $sig_T(X_j) = e \in  \{0,\dots, b-1 \}$ then we have  $1 + e \not\equiv 0 \text{ mod } (b+1)$.
Again, this is a contradiction to $sup_T(f_{j,1,b-1})\edge{sig_T(X_j)}sup_T(f_{j,0,b})$.
Hence, we have $sig_T(X_j)=b$.
Overall, it is proven that if $sup_T(k)=(1,0)$ then $sup_T(X_j)\in E_0\cup \{b\}$. 

To continue, let $sig_T(k)=(0,1)$.
Similar to the former case, by Lemma~\ref{lem:observations}, we obtain that $sup_T(f_{j,0,b})=sup_T(g_{j,b})=b$ and $sup_T(f_{j,1,b-1})\in \{b-1,b\}$.
If $sup_T(f_{j,1,b-1})=b$ then $sig_T(X_j)\in E_0$.
We show that $sup_T(f_{j,1,b-1})=b-1$ implies $sig_T(X_j)=1$:
Assume $sig_T(X_j)=(m,n)\in E_\tau$.
If $m=n$ or if $m>n$ then, by $sup_T(f_{j,0,b})=sup_T(f_{j,1,b-1})-sig^-_T(X_j)+sig^+_T(X_j)$, we have $sup_T(f_{j,0,b}) < b$, a contradiction.
If $m < n$ then, by $sup_T(g_{j,b+1})=sup_T(g_{j,b})-sig^-_T(X_j)+sig^+_T(X_j)$, we get the contradiction $sup_T(g_{j,b+1}) >  b$.
Hence, $sig_T(X_j)\in \{0,\dots, b\}$.
Again, $sig_T(X_j)=e\in \{0,2\dots, b\}$ implies $(b-1 + e)\not \equiv b \text{ mod } (b+1)$ which contradicts $sup_T(f_{j,0,b})=sup_T(f_{j,1,b-1}) + \vert sig_T(X_j) \vert $.
Consequently, we obtain $sig_T(X_j)=1$ which shows that $sup_T(k)=(0,1)$ implies $sup_T(X_j)\in E_0\cup \{1\}$.
Altogether, this proves Lemma~\ref{lem:translator_2}.\ref{lem:translator_2_completeness}.

To complete the proof Lemma~\ref{lem:translator_2}, we show its second condition to be true.
To do so, we start from a one-in-three model $M\subseteq V(\varphi)$ of $\varphi$ and define the following $\tau$-region $(sup_T, sig_T)$ of $T_\tau$ that satisfies Lemma~\ref{lem:translator_2}.\ref{lem:translator_2_existence}:
For $e\in E_{T_\tau}$ we define $sig_T(e)=$
\[
\begin{cases}
(0,1), & \text{if } e = k\\
0, & \text{if } e\in \{z\}\cup (V(\varphi)\setminus M) \text{ or } e=v_j \text{ and } X_j\in M, 0\leq j\leq m-1 \\
1, & \text{if }  e \in M\cup\{u\} \text{ or } e=v_j \text{ and } X_j\not\in M, 0\leq j\leq m-1\\
\end{cases}
\]
By Lemma~\ref{lem:observations}, having $sig_T$, it is sufficient to define the values of the initial states of the constituent of $T_\tau$.
To do so, we define $sup_T(f_{j,0,0})=sup_T(g_{j,0})=t_{j,0}=0$ for $j\in \{0,\dots, m-1\}$. 
One easily verifies that $(sup_T, sig_T)$ is a well defined region of $T_\tau$.
See Figure~\ref{fig:example} which presents a concrete example of $(sup_T, sig_T)$ for $b=2$, $\varphi_0$ and $M=\{X_0, X_4\}$.
Finally, that proves Lemma~\ref{lem:translator_2}.
\end{proof}

\subsection{The Liaison of Key and Translator}\label{sec:liaison}%

The following lemma completes our reduction and finally proves Theorem~\ref{the:hardness_results}:

\begin{lemma}[Suffiency]\label{lem:liaison}
\begin{enumerate}
\item\label{lem:liaison_essp}
Let $\tau\in \{\tau^b_0,\tau^b_1, \tau^b_2,\tau^b_3\}$. 
$U_\tau$ is $\tau$-feasible, respectively has the $\tau$-ESSP, if and only if there is a $\tau$-region of $U_\tau$ solving its key atom $\alpha_\tau$ if and only if $\varphi$ has a one-in-three model.

\item\label{lem:liaison_ssp}
Let $\tau'\in \{\tau^b_0,\tau^b_1\}$. 
$W$ has the $\tau'$-SSP  if and only if there is a $\tau'$-region of $W$ solving its key atom $\alpha$ if and only if $\varphi$ has a one-in-three model.
\end{enumerate}
\end{lemma}
\begin{proof}
By Lemma~\ref{lem:key_unions_1}, Lemma~\ref{lem:translator_1}, respectively Lemma~\ref{lem:key_unions_2}, Lemma~\ref{lem:translator_2}, the respective key atoms are solvable if and only if $\varphi$ is one-in-three satisfiable.
Clearly, if all corresponding atoms are solvable the key atom is, too.
Hence, it remains to prove that the $\tau$-solvability ($\tau'$-solvability) of the key atom $\alpha_\tau$ ($\alpha$) implies the $\tau$-ESSP and $\tau$-SSP for $U_\tau$ ($\tau'$-SSP for $W$).
Due to space limitation, the corresponding proofs are moved to the appendix.
\end{proof}

\section{Conclusions}%

In this paper, we show that deciding if a TS $A$ has the $\tau$-feasibility or the $\tau$-ESSP, $\tau\in \{\tau^b_0,\dots,\tau^b_3\}$, is NP-complete.
This makes their synthesis NP-hard.
Moreover, we argue that deciding whether $A$ has the $\tau$-SSP, $\tau'\in \{\tau^b_0,\tau^b_1\}$, is also NP-complete.
It remains for future work to investigate if there are superclasses of (pure) $b$-bounded P/T-nets or their extensions where synthesis becomes tractable.
Moreover, one may search for parameters of the net-types or the input TSs for which the decision problems are \emph{fixed parameter tractable}.

\subsubsection*{Acknowledgements} 
I would like to thank Christian Rosenke and Uli Schlachter for their precious remarks.
Also, I'm thankful to the anonymous reviewers.

\bibliography{myBibliography}%

\newpage
\begin{appendix}

\section{Example for $A(U_{\tau^b_2})$ and $A(U_{\tau^b_3})$}\label{sec:example}
\newcommand{\freezer}[4]{

\ifstrequal{#4}{0}{
\begin{scope}[nodes={set=import nodes}, xshift= #2cm, yshift=#3 cm]
\coordinate (c00) at (0,0);
\coordinate(c01) at (1,0) ;
\coordinate (c02) at (2,0) ;
\coordinate (c10) at (0,-1) ;
\coordinate (c11) at (2,-1) ;

\foreach \i in {c00} {\draw(\i) circle (0.3);}
\foreach \i in {c01, c10} {\draw[dashed] (\i) circle (0.3);}
\foreach \i in {c02, c11} {\draw[dotted, thick] (\i) circle (0.3);}

\node (f00) at (0,0) {\nscale{$f_{#1,0,0}$}};
\node (f01) at (1,0) {\nscale{$f_{#1,0,1}$}};
\node (f02) at (2,0) {\nscale{$f_{#1,0,2}$}};
\node (f10) at (0,-1) {\nscale{$f_{#1,1,0}$}};
\node (f11) at (2,-1) {\nscale{$f_{#1,1,1}$}};
\graph{ 
(f00) ->[thick,"\escale{$k$}"] (f01)->[thick,"\escale{$k$}"] (f02); 
(f10) ->[thick,"\escale{$k$}"] (f11);
(f00) ->[thick,swap, "\escale{$v_#1$}"] (f10);
(f11) ->[thick,swap, "\escale{$X_#1$}"] (f02);
};
\end{scope}
}{
\begin{scope}[nodes={set=import nodes}, xshift= #2cm, yshift=#3 cm]

\coordinate (c00) at (0,0);
\coordinate(c01) at (1,0) ;
\coordinate (c02) at (2,0) ;
\coordinate (c10) at (0,-1) ;
\coordinate (c11) at (2,-1) ;

\foreach \i in {c00,c10} {\draw(\i) circle (0.3);}
\foreach \i in {c01, c11} {\draw[dashed] (\i) circle (0.3);}
\foreach \i in {c02} {\draw[dotted, thick] (\i) circle (0.3);}

\node (f00) at (0,0) {\nscale{$f_{#1,0,0}$}};
\node (f01) at (1,0) {\nscale{$f_{#1,0,1}$}};
\node (f02) at (2,0) {\nscale{$f_{#1,0,2}$}};
\node (f10) at (0,-1) {\nscale{$f_{#1,1,0}$}};
\node (f11) at (2,-1) {\nscale{$f_{#1,1,1}$}};
\graph{ 
(f00) ->[thick,"\escale{$k$}"] (f01)->[thick,"\escale{$k$}"] (f02); 
(f10) ->[thick,"\escale{$k$}"] (f11);
(f00) ->[thick,swap, "\escale{$v_#1$}"] (f10);
(f11) ->[thick,swap, "\escale{$X_#1$}"] (f02);
};
\end{scope}
}
}
\newcommand{\generator}[4]{

\ifstrequal{#4}{0}{
\begin{scope}[nodes={set=import nodes}, xshift = #2cm, yshift= #3cm, ]

\coordinate (c00) at (0,0);
\coordinate(c01) at (1,0) ;
\coordinate(c02) at (2,0) ;
\coordinate(c03) at (3,0) ;
\foreach \i in {c00} {\draw(\i) circle (0.3);}
\foreach \i in {c01} {\draw[dashed] (\i) circle (0.3);}
\foreach \i in {c02,c03} {\draw[dotted, thick] (\i) circle (0.3);}
\node (g00) at (0,0) {\nscale{$g_{#1,0}$}};
\node (g01) at (1,0) {\nscale{$g_{#1,1}$}};
\node (g02) at (2,0) {\nscale{$g_{#1,2}$}};
\node (g03) at (3,0) {\nscale{$g_{#1,2}$}};
\node (g00) at (0,0) {\nscale{$g_{#1,0}$}};
\node (g01) at (1,0) {\nscale{$g_{#1,1}$}};
\node (g02) at (2,0) {\nscale{$g_{#1,2}$}};
\node (g03) at (3,0) {\nscale{$g_{#1,2}$}};
\graph{ 
(g00) ->[thick,"\escale{$k$}"] (g01)->[thick,"\escale{$k$}"] (g02)->[thick,"\escale{$X_#1$}"] (g03); 
};
\end{scope}
}{
\begin{scope}[nodes={set=import nodes}, xshift = #2cm, yshift= #3cm, ]
\coordinate (c00) at (0,0);
\coordinate(c01) at (1,0) ;
\coordinate(c02) at (2,0) ;
\coordinate(c03) at (3,0) ;
\foreach \i in {c00,c03} {\draw(\i) circle (0.3);}
\foreach \i in {c01} {\draw[dashed] (\i) circle (0.3);}
\foreach \i in {c02} {\draw[dotted, thick] (\i) circle (0.3);}
\node (g00) at (0,0) {\nscale{$g_{#1,0}$}};
\node (g01) at (1,0) {\nscale{$g_{#1,1}$}};
\node (g02) at (2,0) {\nscale{$g_{#1,2}$}};
\node (g03) at (3,0) {\nscale{$g_{#1,2}$}};
\node (g00) at (0,0) {\nscale{$g_{#1,0}$}};
\node (g01) at (1,0) {\nscale{$g_{#1,1}$}};
\node (g02) at (2,0) {\nscale{$g_{#1,2}$}};
\node (g03) at (3,0) {\nscale{$g_{#1,2}$}};
\graph{ 
(g00) ->[thick,"\escale{$k$}"] (g01)->[thick,"\escale{$k$}"] (g02)->[thick,"\escale{$X_#1$}"] (g03); 
};
\end{scope}
}

}
\newcommand{\translator}[7]{

\ifstrequal{#7}{1}{
\begin{scope}[nodes={set=import nodes}, xshift=#5 cm, yshift=#6 cm]

\coordinate (c0) at (0,0);
\coordinate (c1) at (1,0) ;
\coordinate (c2) at (2,0) ;
\coordinate (c3) at (3,0) ;
\coordinate (c4) at (4,0) ;
\coordinate (c5) at (5,0) ;
\coordinate(c6) at (6,0) ;
\coordinate (c7) at (7,0) ;
\coordinate (c8) at (8,0) ;

\foreach \i in {c0,c6, c3,c4,c5} {\draw(\i) circle (0.3);}
\foreach \i in {c1,c7} {\draw[dashed] (\i) circle (0.3);}
\foreach \i in {c2,c8} {\draw[dotted, thick] (\i) circle (0.3);}

\node (t0) at (0,0) {\nscale{$t_{#1,0}$}};
\node (t1) at (1,0) {\nscale{$t_{#1,1}$}};
\node (t2) at (2,0) {\nscale{$t_{#1,2}$}};
\node (t3) at (3,0) {\nscale{$t_{#1,3}$}};
\node (t4) at (4,0) {\nscale{$t_{#1,4}$}};
\node (t5) at (5,0) {\nscale{$t_{#1,5}$}};
\node (t6) at (6,0) {\nscale{$t_{#1,6}$}};
\node (t7) at (7,0) {\nscale{$t_{#1,7}$}};
\node (t8) at (8,0) {\nscale{$t_{#1,8}$}};
\graph{ 
(t0) ->[thick,"\escale{$k$}"] (t1)->[thick,"\escale{$k$}"] (t2)->[thick,"\escale{$X_#2$}"] (t3)->[thick,"\escale{$X_#3$}"] (t4)->[thick,"\escale{$X_#4$}"] (t5)->[thick,"\escale{$z$}"] (t6)->[thick,"\escale{$k$}"] (t7)->[thick,"\escale{$k$}"] (t8); 
};
\end{scope}
}{
\ifstrequal{#7}{2}{
\begin{scope}[nodes={set=import nodes}, xshift=#5 cm, yshift=#6 cm]

\coordinate (c0) at (0,0);
\coordinate (c1) at (1,0) ;
\coordinate (c2) at (2,0) ;
\coordinate (c3) at (3,0) ;
\coordinate (c4) at (4,0) ;
\coordinate (c5) at (5,0) ;
\coordinate(c6) at (6,0) ;
\coordinate (c7) at (7,0) ;
\coordinate (c8) at (8,0) ;

\foreach \i in {c0,c6, c4,c5} {\draw(\i) circle (0.3);}
\foreach \i in {c1,c7} {\draw[dashed] (\i) circle (0.3);}
\foreach \i in {c2,c8, c3} {\draw[dotted, thick] (\i) circle (0.3);}

\node (t0) at (0,0) {\nscale{$t_{#1,0}$}};
\node (t1) at (1,0) {\nscale{$t_{#1,1}$}};
\node (t2) at (2,0) {\nscale{$t_{#1,2}$}};
\node (t3) at (3,0) {\nscale{$t_{#1,3}$}};
\node (t4) at (4,0) {\nscale{$t_{#1,4}$}};
\node (t5) at (5,0) {\nscale{$t_{#1,5}$}};
\node (t6) at (6,0) {\nscale{$t_{#1,6}$}};
\node (t7) at (7,0) {\nscale{$t_{#1,7}$}};
\node (t8) at (8,0) {\nscale{$t_{#1,8}$}};
\graph{ 
(t0) ->[thick,"\escale{$k$}"] (t1)->[thick,"\escale{$k$}"] (t2)->[thick,"\escale{$X_#2$}"] (t3)->[thick,"\escale{$X_#3$}"] (t4)->[thick,"\escale{$X_#4$}"] (t5)->[thick,"\escale{$z$}"] (t6)->[thick,"\escale{$k$}"] (t7)->[thick,"\escale{$k$}"] (t8); 
};
\end{scope}

}{
\begin{scope}[nodes={set=import nodes}, xshift=#5 cm, yshift=#6 cm]

\coordinate (c0) at (0,0);
\coordinate (c1) at (1,0) ;
\coordinate (c2) at (2,0) ;
\coordinate (c3) at (3,0) ;
\coordinate (c4) at (4,0) ;
\coordinate (c5) at (5,0) ;
\coordinate(c6) at (6,0) ;
\coordinate (c7) at (7,0) ;
\coordinate (c8) at (8,0) ;

\foreach \i in {c0,c6,c5} {\draw(\i) circle (0.3);}
\foreach \i in {c1,c7} {\draw[dashed] (\i) circle (0.3);}
\foreach \i in {c2,c8, c3, c4} {\draw[dotted, thick] (\i) circle (0.3);}

\node (t0) at (0,0) {\nscale{$t_{#1,0}$}};
\node (t1) at (1,0) {\nscale{$t_{#1,1}$}};
\node (t2) at (2,0) {\nscale{$t_{#1,2}$}};
\node (t3) at (3,0) {\nscale{$t_{#1,3}$}};
\node (t4) at (4,0) {\nscale{$t_{#1,4}$}};
\node (t5) at (5,0) {\nscale{$t_{#1,5}$}};
\node (t6) at (6,0) {\nscale{$t_{#1,6}$}};
\node (t7) at (7,0) {\nscale{$t_{#1,7}$}};
\node (t8) at (8,0) {\nscale{$t_{#1,8}$}};
\graph{ 
(t0) ->[thick,"\escale{$k$}"] (t1)->[thick,"\escale{$k$}"] (t2)->[thick,"\escale{$X_#2$}"] (t3)->[thick,"\escale{$X_#3$}"] (t4)->[thick,"\escale{$X_#4$}"] (t5)->[thick,"\escale{$z$}"] (t6)->[thick,"\escale{$k$}"] (t7)->[thick,"\escale{$k$}"] (t8); 
};
\end{scope}
}
}
}
\begin{figure}[H]
\centering
\begin{tikzpicture}
\begin{scope}
\coordinate (c00) at (0,0);
\coordinate(c01) at (1,0) ;
\coordinate (c02) at (2,0) ;
\coordinate (c10) at (0,-1) ;
\coordinate (c11) at (2,-1) ;

\foreach \i in {c00} {\draw(\i) circle (0.3);}
\foreach \i in {c01,c10} {\draw[dashed] (\i) circle (0.3);}
\foreach \i in {c02,  c11} {\draw[dotted, thick] (\i) circle (0.3);}
\node (h00) at (0,0) {\nscale{$h_{3,0,0}$}};
\node (h01) at (1,0) {\nscale{$h_{3,0,1}$}};
\node (h02) at (2,0) {\nscale{$h_{3,0,2}$}};
\node (h10) at (0,-1) {\nscale{$h_{3,1,0}$}};
\node (h11) at (2,-1) {\nscale{$h_{3,1,1}$}};
%
\graph{ 
(h00) ->[thick,"\escale{$k$}"] (h01)->[thick,"\escale{$k$}"] (h02); 
(h10) ->[thick,swap, "\escale{$k$}"] (h11);
(h00) ->[thick,"\escale{$u$}", swap](h10);
(h11) ->[thick,"\escale{$z$}", swap](h02);
};
\end{scope}
\draw (0,1) circle (0.3);
\node (q_0) at (0,1) {\nscale{$q_0$}};
\draw[->, thick] (q_0)--(h00)node [pos=0.4, left ] {\escale{$y_0$}};
\freezer{0}{3}{0}{1}
\draw (3,1) circle (0.3);
\node (q_1) at (3,1) {\nscale{$q_1$}};
\draw[->, thick] (q_1)--(f00)node [pos=0.4, left ] {\escale{$y_1$}};
\draw[->, thick] (q_0)--(q_1)node [pos=0.5,above] {\escale{$w_1$}};
\freezer{1}{6}{0}{0}
\draw (6,1) circle (0.3);
\node (q_2) at (6,1) {\nscale{$q_2$}};
\draw[->, thick] (q_2)--(f00)node [pos=0.4, left ] {\escale{$y_2$}};
\draw[->, thick] (q_1)--(q_2)node [pos=0.5,above] {\escale{$w_2$}};
\freezer{2}{9}{0}{0}
\draw (9,1) circle (0.3);
\node (q_3) at (9,1) {\nscale{$q_3$}};
\draw[->, thick] (q_3)--(f00)node [pos=0.4, left ] {\escale{$y_3$}};
\draw[->, thick] (q_2)--(q_3)node [pos=0.5, above] {\escale{$w_3$}};
\freezer{3}{9}{-3}{0}
\coordinate (corner1) at (11.5,1) {};
\coordinate(corner2) at (11.5,-2) {};
\draw (9,-2) circle (0.3);
\node (q_4) at (9,-2) {\nscale{$q_4$}};
\draw[->, thick]  (q_3)--(corner1)--(corner2)->(q_4) node [pos=0.5, above] {\escale{$w_4$}};
\draw[->, thick] (q_4)--(f00)node [pos=0.4, left] {\escale{$y_4$}};
\freezer{4}{6}{-3}{1}
\draw (6,-2) circle (0.3);
\node (q_5) at (6,-2) {\nscale{$q_5$}};
\draw[->, thick]  (q_5)->(f00)node [pos=0.4, left] {\escale{$y_5$}};
\draw[->, thick]  (q_4)->(q_5) node [pos=0.5, above] {\escale{$w_5$}};
\freezer{5}{3}{-3}{0}
\draw (3,-2) circle (0.3);
\node (q_6) at (3,-2) {\nscale{$q_6$}};
\draw[->, thick]  (q_6)->(f00)node [pos=0.4, left] {\escale{$y_6$}};
\draw[->, thick]  (q_5)->(q_6) node [pos=0.5, above] {\escale{$w_6$}};
\translator{0}{0}{1}{2}{1}{-5}{1}
\coordinate (corner3) at (0,-2);
\draw (0,-5) circle (0.3);
\node (q_7) at (0,-5) {\nscale{$q_7$}};
\draw[->, thick] (q_6)--(corner3)->(q_7)node [pos=0.5, left] {\escale{$w_7$}};
\draw[->, thick] (q_7)--(t0)node [pos=0.4, above] {\escale{$y_7$}};
\translator{1}{2}{0}{3}{1}{-6}{2}
\draw (0,-6) circle (0.3);
\node (q_8) at (0,-6) {\nscale{$q_8$}};
\draw[->, thick]  (q_7)->(q_8)node [pos=0.5, left] {\escale{$w_8$}};
\draw[->, thick]  (q_8)->(t0)node [pos=0.4, above] {\escale{$y_8$}};
\translator{2}{1}{3}{0}{1}{-7}{3}
\draw (0,-7) circle (0.3);
\node (q_9) at (0,-7) {\nscale{$q_9$}};
\draw[->, thick]  (q_8)->(q_9)node [pos=0.5, left] {\escale{$w_9$}};
\draw[->, thick]  (q_9)->(t0)node [pos=0.4, above] {\escale{$y_9$}};
\translator{3}{2}{4}{5}{1}{-8}{2}
\draw (0,-8) circle (0.3);
\node (q_10) at (0,-8) {\nscale{$q_{10}$}};
\draw[->, thick]  (q_9)->(q_10)node [pos=0.5, left] {\escale{$w_{10}$}};
\draw[->, thick]  (q_10)->(t0)node [pos=0.4, above] {\escale{$y_{10}$}};
\translator{4}{1}{5}{4}{1}{-9}{3}
\draw (0,-9) circle (0.3);
\node (q_11) at (0,-9) {\nscale{$q_{11}$}};
\draw[->, thick]  (q_10)->(q_11)node [pos=0.5, left] {\escale{$w_{11}$}};
\draw[->, thick]  (q_11)->(t0)node [pos=0.4, above] {\escale{$y_{11}$}};
\translator{5}{4}{3}{5}{1}{-10}{1}
\draw (0,-10) circle (0.3);
\node (q_12) at (0,-10) {\nscale{$q_{12}$}};
\draw[->, thick]  (q_11)->(q_12)node [pos=0.5, left] {\escale{$w_{12}$}};
\draw[->, thick]  (q_12)->(t0)node [pos=0.4, above] {\escale{$y_{12}$}};
\generator{0}{0}{-12}{1}
\draw (0,-11) circle (0.3);
\node (q_13) at (0,-11) {\nscale{$q_{13}$}};
\draw[->, thick]  (q_13)->(g00)node [pos=0.4, left] {\escale{$y_{13}$}};
\generator{1}{3.75}{-12}{0}
\draw (3.75,-11) circle (0.3);
\node (q_14) at (3.75,-11) {\nscale{$q_{14}$}};
\draw[->, thick]  (q_14)->(g00)node [pos=0.4, left] {\escale{$y_{14}$}};
\generator{2}{7.5}{-12}{0}
\draw (7.5,-11) circle (0.3);
\node (q_15) at (7.5,-11) {\nscale{$q_{15}$}};
\draw[->, thick]  (q_15)->(g00)node [pos=0.4, left] {\escale{$y_{15}$}};
\draw[->, thick]  (q_12)->(q_13)node [pos=0.5, left] {\escale{$w_{13}$}}->(q_14)node [pos=0.5, above] {\escale{$w_{14}$}}->(q_15)node [pos=0.5, above] {\escale{$w_{15}$}};
\coordinate (corner4) at(11,-11);
\coordinate (corner5) at(11,-13);
\generator{3}{7.5}{-14}{0}
\draw (7.5,-13) circle (0.3);
\node (q_16) at (7.5,-13) {\nscale{$q_{16}$}};
\draw[->, thick]  (q_16)->(g00)node [pos=0.4, left] {\escale{$y_{16}$}};
\generator{4}{3.75}{-14}{1}
\draw (3.75,-13) circle (0.3);
\node (q_17) at (3.75,-13) {\nscale{$q_{17}$}};
\draw[->, thick]  (q_17)->(g00)node [pos=0.4, left] {\escale{$y_{17}$}};
\generator{5}{0}{-14}{0}
\draw (0,-13) circle (0.3);
\node (q_18) at (0,-13) {\nscale{$q_{18}$}};
\draw[->, thick]  (q_18)->(g00)node [pos=0.4, left] {\escale{$y_{18}$}};
\draw[->, thick]  (q_15)--(corner4)--(corner5)->(q_16) node [pos=0.5, above] {\escale{$w_{16}$}}->(q_17) node [pos=0.5, above] {\escale{$w_{17}$}}->(q_18) node [pos=0.5, above] {\escale{$w_{18}$}};
\end{tikzpicture}
\caption{For $\tau\in \{\tau^2_2,\tau^2_3\}$ the TS $A(U_\tau)$ where $\varphi=\{C_0,\dots, C_{5}\}$ with $C_0=\{X_0,X_1,X_2\},\ C_1= \{X_2,X_0,X_3\},\ C_2= \{X_1,X_3,X_0\},\ C_3= \{X_2,X_4,X_5\},\ C_4=\{X_1,X_5,X_4\},\ C_5= \{X_4,X_3,X_5\}$ and a sketched $\tau$-region $(sup, sig)$ solving $(k, h_{3,1,1})$.
If $sup(s)=0$ ($sup(s)=1$, $sup(s)=2$) then $s$ is surrounded by a solid (dashed, dotted) circle.
The signature is chosen in compliance to Lemma~\ref{lem:union_validity}, Lemma~\ref{lem:key_unions_2} and Lemma~\ref{lem:translator_2}.
The model is $M=\{X_0,X_4\}$.}
\label{fig:example}
\end{figure}

\section{Proofs for Section~\ref{sec:hardness_results}}%

\subsection{Proofs of Lemma~\ref{lem:union_validity} and Lemma~\ref{lem:observations}}

\begin{proof}[Proof of Lemma~\ref{lem:union_validity}]
\emph{If}:
If $(sup, sig)$ is a $\tau$-region of $A(U)$ that, for $e\in E_U, s,s'\in S_U$, solves $(e,s)$, respectively $(s,s')$, then projecting $(sup,sig)$ to the component TSs of $U$ yields a $\tau$-region of $U$ that solves the respective separation atom in $U$.
Hence, the $\tau$-(E)SSP of $A(U)$ implies the $\tau$-(E)SSP of $U$.

\emph{Only-if}:
Let $0_\tau=(0,0)$ if $(0,0)\in E_\tau$ and, otherwise, $0_\tau=0$.
A $\tau$-region $(sup, sig)$ of $U$ that solves $(s,s')$, respectively $(e, s)$, can be extended to a corresponding $\tau$-region $(sup', sig')$ of $A(U)$ by setting:
\begin{align*}
sup'(s'') &= \begin{cases}
sup(s''), & \text{if } s'' \in S_U,\\
sup(s), & \text{if } s'' \in Q 
\end{cases}\\
sig'(e') &= \begin{cases}
sig(e'), & \text{if } e' \in E_U,\\
0_\tau, & \text{if } e' \in W ,\\
(sup(s) - sup(s_{0,A_i}),0) & \text{if } e' = y_i \text{ and } sup(s_{0,A_i})  < sup(s),  0 \leq i \leq n \\
( 0, sup(s_{0,A_i} ) - sup(s)) & \text{if } e' = y_i \text{ and } sup(s_{0,A_i})  \geq  sup(s),  0 \leq i \leq n 
\end{cases}
\end{align*}

A $\tau$-region $(sup, sig)$ defined in that way inherits the property to solve $(e,s)$, respectively $(s,s')$, from $(sup, sig)$ and solves $(e,q_i)$ for $i\in \{0,\dots, n\}$ as, by definition, $sup(q_i)=sup(s)$ for all $i\in \{0,\dots, n\}$.
Consequently, as for every event $e\in E_U$ there is at least one state $s\in  S_U$ such that $(e,s)$ is a valid ESSP atom of $U$, the atom $(e,q_i)$ is solvable for every $e\in E_U$ and $i\in \{0,\dots, n\}$.
As a result, to prove the $\tau$-(E)SSP for $A(U)$ it remains to argue that the SSP atoms states $(q_0,\cdot),\dots, (q_n,\cdot)$ and the ESSP atoms $(w_1,\cdot),\dots, (w_n,\cdot), (y_0,\cdot),\dots, (y_n,\cdot)$ are solvable in $A(U)$.
If $i\in \{0,\dots, n\}$ and $s\in S_{A(U)}, e\in E_{A(U)}$ then the following region $(sup, sig)$ simultaneously solves every valid atom $(y_i, \cdot)$, $(q_i,\cdot) $ and, if it exists, $(w_{i+1}, \cdot)$ in $A(U)$:
\[
sup(s) =\begin{cases}
	0, & \text{if } s=q_i\\
	b, & \text{otherwise }
	\end{cases}\\
\text{   }
sig(e)=\begin{cases}
	(0,b) & \text{if } e = y_i \text{ or } ( i < n \text{ and } e=w_{i+1})\\
	(b,0) & \text{if } 1 < i  \text{ and } e=w_{i-1}\\
	0_\tau & \text{ otherwise} \\
	\end{cases}
\]
\end{proof}

\begin{proof}[Proof of Lemma~\ref{lem:observations}]
(\ref{lem:sig_summation_along_paths}): 
The first claim follows directly from the definitions of $\tau$ and $\tau$-regions.

(\ref{lem:absolute_value}): 
The \textit{If}-direction is trivial.
For the \textit{Only-if}-direction we show that the assumption $(m,n)\not\in \{(1,0),(0,1)\}$ yields a contradiction.

By (\ref{lem:sig_summation_along_paths}), we have that $sup(s_b)=sup(s_0) + b\cdot(n-m)$.
If $\vert n-m\vert > 1$, then we get a contradiction to $sup(s_0)\geq 0$ or to $sup(s_b)\leq b$.
Hence, if $n\not=m$ then $\vert n-m\vert =1$ implying $m > n$ or $m < n$.
For a start, we show that $ m > n$ implies $m=1, n=0$, that is, 
by $n \leq m-1$ and $sup(s_0)\leq b$ we obtain the estimation
\[
sup(s_{b-1}) = sup(s_0) +(b-1)(n-m)\ \leq \ b + (b-1)(m-1-m) = 1
\]
By $n < m \leq sup(s_{b-1})\leq 1$ we have $(m,n)=(1,0)$.
Similarly, we obtain that $(m,n)=(0,1)$ if $m < n$.
Hence, if $sig(e)=(m,n)$ and $n\not=m$ then $sig(e) \in \{(1,0),(0,1)\}$.
The second statement follows directly from (\ref{lem:sig_summation_along_paths}).
\end{proof}

\subsection{Completion of the Proof of Lemma~\ref{lem:liaison}}%

To complete the proof of Lemma~\ref{lem:liaison}, we stepwise prove the following statements in the given order:

\begin{enumerate}\label{en:todo_for_lemma_hardness_results}

\item\label{en:essp_atom_1}
If $\tau=\tau^b_1$ then the $\tau$-solvability of $(k , h_{1, 2b+4})$ in $U_\tau$ implies its $\tau$-ESSP.

\item\label{en:essp_atom_0}
If $\tau=\tau^b_0$ then the $\tau$-solvability of $(k , h_{0, 4b+1})$ in $U_\tau$ implies its $\tau$-ESSP.

\item\label{en:essp_implies_ssp}
If $\tau\in \{\tau^b_0,\tau^b_1\}$ then the $\tau$-ESSP of $U_\tau$ implies its $\tau$-SSP.

\item\label{en:essp_atom_2}
If $\tau\in \{\tau^b_2, \tau^b_3 \}$ then the $\tau$-solvability of $(k , h_{3,1,b-1})$ in $U_\tau$ implies its $\tau$-ESSP.

\item\label{en:essp_implies_ssp_Z}
If $\tau\in \{\tau^b_2,\tau^b_3\}$  then the $\tau$-ESSP of $U_\tau$ implies its $\tau$-SSP.

\item\label{en:ssp_atom}
If $\tau\in \{\tau^b_0,\tau^b_1\}$ then the $\tau$-solvability of $(h_{2,0} ,  h_{2,b})$ in $W$ implies its $\tau$-SSP.
 
\end{enumerate}

\subsubsection{Proof of Statement~\ref{en:essp_atom_1} and Statement~\ref{en:essp_atom_0}}%

We prove for $ \tau \in \{ \tau^b_0, \tau^b_1\}$ that the $\tau$-solvability of the key atom $a_\tau$ in $U_\tau$ implies the $\tau$-solvability of all ESSP atoms by the presentation of corresponding regions.
To do so, we provide for every ESSP atom $(e,s)$ of $U_\tau$ a corresponding $\tau$-region $(sup, sig)$ solving it.
For the sake of simplicity, these regions are often presented as rows of a table with the shape and meaning:
See Table~\ref{tab:first_table} for the first example.

\begin{enumerate}
\item\label{e}
\emph{e}: 
Here, $e$ means the event of the ESSP atoms $(e, \cdot)$ which are solved by the region of this row.
The corresponding states are listed in the \emph{states}-cell.
It is always the case that a $\tau$-region $(sup, sig)$ that solves such an atom $(e, \cdot)$ satisfies $sig(e) = (0,n)$ for some $n\in \mathbb{N}^+$.
\item\label{states}
\emph{states}:
All listed states $s$ such that $(e, s)$ is $\tau$-solved by the region of the corresponding row.
\item\label{initials}
\emph{initials}: 
By Lemma~\ref{lem:observations}, a $\tau$-region of $U_\tau$ is fully defined by its signature and the support of the initial states of the constituent TSs.
Hence, this cell explicitly presents the supports of the initial states of the TSs of $U_\tau$, which are actually affected by an event having a signature different from $(0,0)$.
The initial states of all other TSs, that is, all those constituents which have no event in their event set with a signature different from $(0,0)$, are assumed to be mapped to $b$.
Certainly, this condemns the states of all unaffected TSs to have the same support $b$.
As mentioned above, every $(e,\cdot )$ solving $\tau$-region $R=(sup, sig)$ satisfies $sig(e)=(0,n), n\in \mathbb{N}^+$.
Thus, for every state $s$ of an unaffected TS the atom $(e, s)$ is automatically solved by $R$, as $sup(s)=b$ and $\neg b\edge{(0,n)}$ for $n\geq 1$.
For the sake of readability, we never mention these states explicitly in the \emph{states}-cell.
\item
\emph{sig}:
The signatures of the events of $U_\tau$ with a value different from $(0,0)$.
The signature of the other events is $(0,0)$.
\item
\emph{constituents}:
For the sake of transparency, the constituents which are affected by events with a signature different from $(0,0)$.
Note that, by the discussion of (\ref{e}) and (\ref{initials}), for every state $s$ of a constituent which is not mentioned here it is true that the ESSP atom $(e,s)$ is also solved.
\end{enumerate}

Moreover, especially in the presented tables, we apply several shortcuts to make the presentations more lucid:
\begin{enumerate}
\item
If $s\in S_{U_\tau}$ is an initial state of an affected TS with support $sup(s)=s_\tau \in S_\tau$ then we write $s=s_\tau$.
\item
We differentiate between $i$-indexed and $j$-indexed events and insinuate the following double meaning:
If '$i$' occurs explicitly in the index of a presented event, respectively state, for example $k_{6i+1}$, respectively $t_{i,0,0}$, then it is assumed that $i\in \{0,\dots, m-1\}$ is arbitrary but fixed.
In contrast, if not stated explicitly otherwise, if '$j$' occurs explicitly in the index of a presented event or state then $j$ represents all possible values for this type of state or event.
For example, we write $d_{j,1,0}$ to abridge the enumeration $d_{0,1,0},\dots, d_{6m-1,1,0}$. 
\end{enumerate}

\begin{proof}[Statement~\ref{en:essp_atom_1} ]
Let $\tau=\tau^b_1$. 
The following table presents for atoms $(e, \cdot)$ of $U_\tau$ solving $\tau$-regions, where $e\in \{z_0,z_1,o_0,o_1, k,k_{6i}, k_{6i+2}, k_{6i+4}\}$.
\begin{longtable}{p{1cm} p{3.1cm}p{4.4cm} p{2.1cm} p{1.5cm}}
\caption{For $\tau=\tau^b_1$: Solving $\tau$-regions for atoms $(e, \cdot)$ of $U_\tau$ where $e\in \{z_0,z_1,o_0,o_1, k,k_{6i}, k_{6i+2}, k_{6i+4}\}$.}
\label{tab:first_table}
\endfirsthead
\endhead
\emph{e}& \emph{initials}& \emph{states} & \emph{sig} & \emph{constituents}\\ \hline
$z_0$ & $h_{1,0}=b$ & \raggedright{$S_{H_1}\setminus \{h_{1,2b+2},h_{1,3b+5}\} $} & \raggedright{$z_0=(0,b)$\\ $k=(1,0)$} & $H_1$\\
$z_0$ & $h_{1,0}=0$ & \raggedright{$h_{1,2b+2},h_{1,3b+5}$} & \raggedright{$z_0=(0,b)$,\\ $z_1=(b,0)$} & $H_1$\\ \hline
$z_1$ & \raggedright{$h_{1,0}=b$, $d_{j,1,0}=0$} & \raggedright{$S_{H_1}\setminus \{h_{1,b},h_{1,b+1},h_{1,3b+5}\} $, $\{d_{j,1,1}, d_{j,1,2}, d_{j,1,3}\}$} & \raggedright{$z_1=(0,b)$, $o_0=(0,b)$, $k=(1,0)$} & $H_1, D_{j,1}$\\

$z_1$ & \raggedright{$h_{1,0}=b$, $d_{j,1,0}=b$} & \raggedright{$\{h_{1,b},h_{1,b+1}, h_{1,3b+5}, d_{j,1,0}\}$} & \raggedright{$o_0=(b,0)$, $z_1=(0,b)$} & $H_1, D_{j,1}$\\ \hline
$o_0$ & \raggedright{$h_{1,0}=b$, $d_{j,1,0}=0$} &  \raggedright{$S_{H_1}\setminus \{h_{1,2b+4},\dots, h_{1,3b+5}\}$, $\{d_{j,1,1}, d_{j,1,2}, d_{j,1,3}\}$}  &  \raggedright{$o_0=(0,b)$, $z_0=(b,0)$} & $H_1, D_{j,1}$\\

$o_0$ &  \raggedright{$h_{1,0}=0$, $d_{j,1,0}=0$} &  \raggedright{$\{h_{1,2b+4},\dots, h_{1,3b+5}\}$} &  \raggedright{$o_0=(0,b)$} & $H_1, D_{j,1}$\\ \hline
$o_1$ & \raggedright{$h_{1,0}=t_{j,0,0}=b$, $t_{j,1,0}=t_{j,2,0}=b$, $0\leq j\leq 3m-1: d_{2j,1,0}=b$, $d_{2j+1,1,0}=0$ } &  \raggedright{$S_{H_1}\setminus \{h_{1,b+1},\dots, h_{1,2b+2}\}$, $0\leq j\leq 3m-1: S_{D_{2j,1}}$}  &  \raggedright{$o_1=(0,b)$, $z_0=(b,0)$, $z_1=(0,b)$, $k_{2j}=(b,0)$} & $H_1, D_{j,1}$, $T_\tau$\\

$o_1$ & \raggedright{$h_{1,0}=t_{j,0,0}=b$, $t_{j,1,0}=t_{j,2,0}=b$, $0\leq j\leq 3m-1: d_{2j,1,0}=0$, $d_{2j+1,1,0}=b$ } &  \raggedright{$ \{h_{1,b+1},\dots, h_{1,2b+2}\}$, $0\leq j\leq 3m-1: S_{D_{2j+1,1}}$}  &  \raggedright{$o_1=(0,b)$, $z_1=(b,0)$, $k_{2j+1}=(b,0)$} & $H_1, D_{j,1}$, $T_\tau$\\

$o_1$ &  \raggedright{$h_{1,0}=0$, $d_{j,1,0}=0$} &  \raggedright{remaining states} &  \raggedright{$o_1=(0,b)$} & $H_1, D_{j,1}, $\\ \hline
$k$ & key region &  $S_{H_1}$& see Lemma~\ref{lem:key_unions_1}, Lemma~\ref{lem:translator_1}  & \\ 
$k$ & $h_{1,0}=0$ &  \raggedright{$S_{U_\tau} \setminus S_{H_1}$} &  \raggedright{$k=(0,1)$, $z_0=(b,0)$} & $H_1$\\ \hline
$k_{6i}$ & \raggedright{$h_{1,0}=b$, $d_{j,1,0}=b$, $t_{i,0,0}=0$} & $d_{6i,1,0}$ & \raggedright{$k_{6i}=(0,b)$, $o_0=(b,0)$} & $H_1, T_{i,0}, D_{j,1}$ \\ 
$k_{6i}$ & \raggedright{$t_{i,0,0}=d_{6i,1,0}=0$} & remaining states & $k_{6i}=(0,b)$ & $T_{i,0}, D_{6i,1}$\\ \hline
$k_{6i+2}$ & \raggedright{$h_{1,0}=b$, $d_{j,1,0}=b$, $t_{i,1,0}=0$} & $d_{6i+2,1,0}$ & \raggedright{$k_{6i+2}=(0,b)$, $o_0=(b,0)$} & $H_1, T_{i,1}, D_{j,1}$ \\ 
$k_{6i+2}$ & \raggedright{$t_{i,1,0}=d_{6i+2,1,0}=0$} & remaining states & $k_{6i+2}=(0,b)$ & $T_{i,1}, D_{6i+2,1}$\\ \hline
$k_{6i+4}$ & \raggedright{$h_{1,0}=b$, $d_{j,1,0}=b$, $t_{i,2,0}=0$} & $d_{6i+4,1,0}$ & \raggedright{$k_{6i+4}=(0,b)$, $o_0=(b,0)$} & $H_1, T_{i,2}, D_{j,1}$ \\ 
$k_{6i+4}$ & \raggedright{$t_{i,2,0}=d_{6i+4,1,0}=0$} & remaining states & $k_{6i+4}=(0,b)$ & $T_{i,2}, D_{6i+4,1}$\\ 
\end{longtable}

It remains to prove the solvability of the valid atoms $(e,\cdot)$ of $U_\tau$ where $e\in E_{T_\tau}$.
To do so, we need the following notations:
If $i\in \{0,\dots, m-1\}$ and $\alpha \in \{0,1,2\}$ are arbitrary but fixed then by $i',i'',\beta, \gamma$ we mean the indices $i',i''\in \{0,\dots, m-1\}\setminus \{i\}$ and $\beta, \gamma \in  \{0,1,2\}$ such that $X_{i,\alpha}=X_{i',\beta}=X_{i'',\gamma}$, that is, $i'$ and $\beta$, respectively $i''$ and $\gamma$, determine the second, respectively third, occurrence of $X_{i,\alpha}$ in $U_\tau$.
The following table shows for an arbitrary but fixed $i\in \{0,\dots, m-1\}$ and all possible values for $\alpha\in \{0,1,2\}$ the solvability of $(X_{i,\alpha},s)$ for the states $s$ of $U_\tau$ which are not in $T_{i',\beta}$ or $T_{i'',\gamma}$ and the solvability of $(x_i,s)$ and $(p_i,s)$ for all states $s$ of $U_\tau$.
Please note, that if $\beta,\gamma\in \{0,2\}$ then $X_{i',\beta}\in E_{T_{i',0}}$ and  $X_{i'',\gamma}\in E_{T_{i'',0}}$, otherwise, if $\beta =\gamma=1$ then $X_{i',\beta}\in E_{T_{i',1}}$ and  $X_{i'',\gamma}\in E_{T_{i'',1}}$.
We can abbreviate this case analyses by identifying $T_{i',\beta\text{mod}2}$, respectively $T_{i'',\gamma\text{mod}2}$, as the translator where $X_{i',\beta}$, respectively $X_{i'',\gamma}$, occur in. 
To abridge, we define $S^{\beta}_\gamma=S_{T_{i', \beta\text{mod}2 }}\cup S_{T_{i'', \gamma\text{mod}2 }}$.
By the arbitrariness of $i$, this approach proves the solvability of every valid ESSP atom $(e, \cdot)$ in $U_\tau$, where $e$ is an event of $T_\tau$.
\begin{longtable}{p{1cm} p{3.5cm}p{3cm} p{2.5cm} p{2cm}}
\caption{For $\tau={\tau^b_1}$: Solving $\tau$-regions for atoms $(e, \cdot )$ of $U_\tau$ where $e\in E_{T_\tau}$.}
\label{tab:second_table}
\endfirsthead
\endhead
\emph{e}& \emph{initials}& \emph{states} & \emph{sig} & \emph{constituents}\\ \hline
$x_i$ & \raggedright{$t_{i,0,0}=0$, $t_{i,2,0}=b$, $d_{6i+4,0,0}=b$, } & $S_{T_{i,2}}$ &  \raggedright{$x_i=(0,b)$, $k_{6i+4}=(b,0)$} & \raggedright{$T_{i,0}, T_{i,2}$, $D_{6i+4,1}$}\arraybackslash \\

$x_i$ & \raggedright{$t_{i,0,0}=b$, $t_{i,2,0}=0$, $\alpha=0: t_{i',\beta \text{mod}2,0}=b$, $t_{i'',\gamma\text{mod}2,0}=b$} & $S_{T_{i,0}}$ & \raggedright{$x_i=(0,b)$, $X_{i,0} =(1,0)$} & \raggedright{$T_{i,0}, T_{i,2}$, $T_{i',\beta\text{mod}2}$, $T_{i'',\gamma\text{mod}2}$}\arraybackslash \\

$x_i$ & \raggedright{$t_{i,0,0}=t_{i,2,0}=0$} & remaining states  & $x_i=(0,b)$ & $T_{i,0}, T_{i,2}$\\ \hline
$p_i$ & \raggedright{$t_{i,0,0}=t_{i,1,0}=b$, $t_{i,2,0}=b$,\newline $\alpha=1$: $t_{i',\beta\text{mod}2,0}=b$, $t_{i'',\gamma\text{mod}2,0}=b$} & $S_{T_{i,1}},S_{T_{i,2}} $ & \raggedright{$p_i=(0,b)$, $x_i=(b,0)$, $X_{i,1}=(1,0)$ } & \raggedright{$T_{i,0}, T_{i,1}, T_{i,2}$, $T_{i',\beta\text{mod}2}$, $T_{i'',\gamma\text{mod}2}$}\arraybackslash \\

$p_i$ & \raggedright{$t_{i,1,0}=t_{i,2,0}=0$} & remaining states & $ p_i=(0,b)$ & $T_{i,1}, T_{i,2}$\\ \hline
$k_{6i+1}$ & \raggedright{$h_{1,0}=d_{j,1,0}=b$, $t_{i,0,0}=0$} & \raggedright{$d_{6i+1,1,0}$, $S_{T_\tau}\setminus S_{T_{i,0}}$} & \raggedright{$k_{6i+1}=(0,b)$, $o_0=(b,0)$} & $H_1, T_{i,0}, D_{j,1}$ \\ 

$k_{6i+1}$ & \raggedright{$t_{i,0,0}=b$, $d_{6i+1,1,0}=0$, $\alpha=2$: $t_{i',\beta\text{mod}2,0}=b$, $t_{i'',\gamma\text{mod}2,0}=b$} & remaining states & \raggedright{$k_{6i+1}=(0,b)$, $X_{i,2}=(1,0)$} & \raggedright{$T_{i,0}, D_{6i+1,1}$, $T_{i',\beta\text{mod}2}$, $T_{i'',\gamma\text{mod}2}$}\arraybackslash \\ \hline
$k_{6i+3}$ & \raggedright{$h_{1,0}=d_{j,1,0}=b$, $t_{i,1,0}=t_{i,2,0}=b$} & \raggedright{$d_{6i+3,1,0}$, $S_{T_\tau}\setminus S_{T_{i,2}}$} & \raggedright{$k_{6i+3}=(0,b),\newline  o_0=p_i=(b,0)$} & \raggedright{$H_1, T_{i,1}, T_{i,2}$, $D_{j,1}$}\arraybackslash \\ 

$k_{6i+3}$ & \raggedright{$d_{6i+3,1,0}=t_{i,1,0}=0$} & remaining states & $k_{6i+3}=(0,b)$ & $D_{6i+3,1}, T_{i,1}$ \\ \hline
$k_{6i+5}$ & \raggedright{$h_{1,0}=d_{j,1,0}=b$, $t_{i,1,0}=t_{i,2,0}=b$} & \raggedright{$d_{6i+5,1,0}$, $S_{T_\tau}\setminus S_{T_{i,1}}$} & \raggedright{$k_{6i+5}=(0,b),\newline  o_0=p_i=(b,0)$} & \raggedright{$H_1, T_{i,1}, T_{i,2}$, $D_{j,1}$}\arraybackslash \\ 

$k_{6i+5}$ & \raggedright{$d_{6i+5,1,0}=t_{i,2,0}=0$} & remaining states & $k_{6i+5}=(0,b)$ & $D_{6i+5,1}, T_{i,2}$ \\ \hline
$X_{i,0}$ & \raggedright{$t_{i,0,0}=d_{6i,1,0}=b$, $\alpha=0$: $t_{i',\beta\text{mod}2, 0}=0$, $t_{i'',\gamma\text{mod}2,0}=0$} & $S_{T_{i,0}}$ & \raggedright{$X_{i,0}=(0,1)$, $k_{6i}=(b,0)$} & \raggedright{$D_{6i,1}, T_{i',\beta\text{mod}2} $, $T_{i,0},  T_{i'',\gamma\text{mod}2}$}\arraybackslash \\ 

$X_{i,0}$ & \raggedright{$ t_{i,0,0}=0 $, $\alpha=0$: $ t_{i',\beta\text{mod}2, 0} = 0 $, $ t_{i'',\gamma\beta\text{mod}2,0} = 0 $} & \raggedright{$S_{U_\tau} \setminus (S_{T_{i,0}} \cup S^{\beta}_\gamma$)}  & $X_{i,0}=(0,1)$ & \raggedright{$T_{i,0}$, $T_{i',\beta\text{mod}2}$, $T_{i'',\gamma\text{mod}2}$}\arraybackslash \\ \hline
$X_{i,1}$ & \raggedright{$t_{i,1,0}=d_{6i+2,1,0}=b$, $\alpha=1$: $t_{i',\beta\text{mod}2, 0}=0$, $t_{i'',\gamma\text{mod}2,0}=0$} & $S_{T_{i,1}}$ & \raggedright{$X_{i,1}=(0,1)$, $k_{6i+2}=(b,0)$} & \raggedright{$D_{6i+2,1}, T_{i,1}$, $T_{i',\beta\text{mod}2}$, $T_{i'',\gamma\text{mod}2}$} \arraybackslash \\ 

$X_{i,1}$ & \raggedright{$t_{i,1,0}=0$, $\alpha=1$: $t_{i',\beta\text{mod}2, 0}=0$, $t_{i'',\gamma\text{mod}2,0}=0$} & \raggedright{$S_{U_\tau} \setminus (S_{T_{i,1}}\cup S^{\beta}_\gamma)$}  & $X_{i,1}=(0,1)$ & \raggedright{$T_{i,1}$, $T_{i',\beta\text{mod}2}$,  $T_{i'',\gamma\text{mod}2}$} \arraybackslash \\ \hline
$X_{i,2}$ & \raggedright{$t_{i,0,0}=t_{i,2,0}=b$, $t_{i',\beta\text{mod}2 , 0}=0$, $t_{i'',\gamma\text{mod}2 ,0}=0$} & $S_{T_{i,0}}$ & \raggedright{$X_{i,2}=(0,1)$, $x_i=(b,0)$} & \raggedright{$T_{i,0}, T_{i,2}$, $T_{i',\beta\text{mod}2 }$, $T_{i'',\gamma\text{mod}2 }$} \arraybackslash \\

$X_{i,2}$ & \raggedright{ $t_{i,0,0}=0$, $t_{i',\beta\text{mod}2, 0}=0$, $t_{i'',\gamma\text{mod}2,0}=0$ } & \raggedright{$S_{U_\tau} \setminus (S_{T_{i,0}} \cup S^{\beta}_\gamma)$}  & \raggedright{$X_{i,2}=(0,1)$} & \raggedright{$T_{i,2}$, $T_{i',\beta\text{mod}2 }$, $T_{i'',\gamma\text{mod}2 }$}\arraybackslash \\
\end{longtable}
\end{proof}

\begin{proof}[Statement~\ref{en:essp_atom_0}]
Let $\tau=\tau^b_0$.
We show that the solvability of $(k, h_{0,4b+1})$ in  $U_\tau$ implies the $\tau$-ESSP for $U_\tau$.
The type $\tau$ has the following obvious property:
If $U=U(A_1,\dots, A_n)$ is a union, $e\in E_{U}$ and if $S=\{s \in S_{A_i} \mid 1\leq i\leq n : e\not\in E_{A_i}\}$ is the set of states of all TSs implemented by $U$ which do not have $e$ in their event set then we can solve $(e,s)$ for all states $s\in S$ by a $\tau$-region $(sup, sig)$ which is defined by $sup(s)=b$ for all $s\in S_U\setminus S$, $sup(s)=0$ for all $s\in S$, $sig(e)=(b,b)$ and $sig(e')=(0,0)$ for all $e'\in E_{U}\setminus \{e\}$.
Hence, in the following, for all $e\in E_{U_\tau}$, we restrict ourselves to the presentation of regions of $U_\tau$ that altogether solves ESSP atoms $(e, s)$ for states $s$ of TSs that actually implement $e$.
By the former observation, this proves every atom $(e,\cdot) $ of $U_\tau$ to be solvable.

The following table presents corresponding regions for a lot of ESSP atoms of $U_\tau$.
However, for some atoms we need regions which are better discussed individually and these atoms are served first.

$(z)$:
The solvability of $(z, s)$ for $s\in S_{H_0}\setminus \{h_{0,2b},h_{0,4b+1}, h_{0,6b+1}\}$ is already proven by the region that solves $(k, h_{0,4b+1})$ presented in the proofs of Lemma~\ref{lem:key_unions_1}.\ref{lem:key_unions_1_existence}, Lemma~\ref{lem:translator_1}.\ref{lem:translator_1_existence}.
The first row of Table~\ref{tab:third_table} proves $(z, s)$  to be solvable for $s\in \{h_{0,2b},h_{0,4b+1}, h_{0,6b+1}\}$.
Hence, every $(z, \cdot)$ is solvable.

$(k, o_1)$:
The solvability of $(k, s)$ for $s\in S_{H_0}\setminus \{h_{0, 4b+2}, \dots, h_{0, 5b}\}$ is already done by the region that solves $(k, h_{0,4b+1})$ presented in the proof of Lemma~\ref{lem:key_unions_1}.\ref{lem:key_unions_1_existence}, Lemma~\ref{lem:translator_1}.\ref{lem:translator_1_existence}.

The solvability of $(o_1, s)$ for 
\[
s\in \{ h_{0,0},\dots, h_{0,b}, h_{0,2b+1},\dots, h_{0,3b+1}, h_{0,5b+1},\dots, h_{0,6b+1}\}
\] 
and the solvability of $(k ,s)$ for $s\in \{ h_{0, 4b+2}, \dots, h_{0, 5b} \}$ can be done as follows:
We use the region $(sup,sig)$ where $sup(h_{0,0})=b$, $sup(d_{j,0,0})=0$ and $sig(o_1)=(0,1)$, $sig(o_0)=(0,b)$, $sig(k)=(b,b)$, $sig(z)=(1,0)$, $sig(k_j)=(b,0)$ for $j\in \{0,\dots, 6m-1\}$ of $K_\tau$ and extend it for $i\in \{0,\dots, m-1\}$ appropriately corresponding to the region $(sup_T, sig_T)$ given for  the proof of Lemma~\ref{lem:translator_1}.\ref{lem:translator_1_existence}:
\begin{enumerate}
\item
$sup(t_{i,0,0})= sup(t_{i,1,0})= sup(t_{i,2,0}) =b$, 
\item
$sig(x_i)=(0,b)$ if $sig''(x_i)=(b,0)$, else $sig(x_i)=sig''(x_i)$,
\item
$sig(p_i)=(0,b)$ if $sig''(p_i)=(b,0)$, else $sig(p_i)=sig''(p_i)$,
\item
for $X\in V(\varphi)$: $sig(X)=(0,1)$ if $sig''(X)=(1,0)$, else $sig(X)=sig''(X)$.
\end{enumerate}

Moreover, to solve $(o_1, s)$ for 
\[
s\in \{ h_{0,b+1},\dots, h_{0,2b-1}, h_{0,3b+2},\dots, h_{0,4b}, d_{0,0,1},\dots, d_{6m-1,0,1}\}
\] 
we extend the region $(sup,sig)$ of $K_\tau$ with $sup(h_{0,0})=0$, $sup(d_{j,0,0})=b$ and $sig(o_1)=(b,b)$, $sig(o_0)=(b,0)$, $sig(z)=(0,1)$, $sig(k_j)=(0,b)$ for $j\in \{0,\dots, 6m-1\}$, by $sup(s)=sup_T(s)$ and $sig(e)=sig_T(e)$ for $s\in S_{T_\tau}$ and $e\in E_{T_\tau}$ where $(sup_T, sig_T)$ is defined in the proof of Lemma~\ref{lem:translator_1}.\ref{lem:translator_1_existence}.

Finally, the region presented in the 4th row of Table~\ref{tab:third_table} solves $(o_1, s)$ for $s\in \{h_{0,2b}, d_{0,0,0}, \dots, d_{6m-1,0,0}\}$.
Altogether it is justified, to consider every atom $(k, \cdot)$ and $(o_1,\cdot)$ to be solvable in $U_\tau$. 

$(o_0)$: The corresponding regions are given in Table~\ref{tab:third_table}.

For the solvability of atoms induced by the remaining events we exploit the already defined regions of Table~\ref{tab:first_table} and Table~\ref{tab:second_table}.
If $(sup, sig)$ is a region of the last three rows of  Table~\ref{tab:first_table} or a region of Table~\ref{tab:second_table} then we use it to create a region $(sup',sig')$ where we replace initials $h_{1,0}, d_{j,1,0}$ by $h_{0,0},d_{j,0,0}$, that is, $sup'(h_{0,0})=sup(h_{1,0})$, $sup'(d_{j,0,0})=sup(d_{j,1,0})$, and define $sup'(s)=sup(s)$ for the other affected initials, respectively, and let $sig'=sig$.
One can easily verify, that, altogether, the generated regions solve the remaining ESSP atoms of $U_\tau$.

\begin{longtable}{p{1cm} p{3cm}p{4.5cm} p{2cm} p{1.5cm}}
\caption{For $\tau=\tau^b_0$: Solving $\tau$-regions for atoms $(e, \cdot)$ of $U_\tau$ where $e\in \{z,o_0,o_1\}$.}
\label{tab:third_table}
\endfirsthead
\endhead
\emph{e}& \emph{initials}& \emph{states} & \emph{sig} & \emph{constituents}\\ \hline
$z$ & $h_{0,0}=0$ & \raggedright{$h_{0,2b},h_{0,4b+1}, h_{0,6b+1}$} & \raggedright{$z=(0,1)$,\\ $o_0=(b,0)$} & $H_0, D_{j,0}$\\ \hline
$o_0$ &  \raggedright{$h_{0,0}=b$, $d_{j,0,0}=0$} &  \raggedright{$h_{0,0},\dots, h_{0,2b-1}$, $d_{0,0,0},\dots, d_{6m-1,0,0}$} &  \raggedright{$o_0=(0,b)$, $z=(1,0)$} & $H_0, D_{j,0}$\\ 

$o_0$ &  \raggedright{$h_{0,0}=0$, $d_{j,0,0}=0$} &  \raggedright{remaining states} &  \raggedright{$o_0=(0,b)$} & $H_0, D_{j,0}$\\ \hline
$o_1$ & \raggedright{$h_{0,0}=d_{j,0,0}=b$} &  \raggedright{$h_{0,2b}, d_{0,0,0},\dots, d_{6m-1,0,0}$}  &  \raggedright{$o_1=(0,1)$, $o_0=(b,0)$} & $H_0, D_{j,0}$\\
\end{longtable}

\end{proof}

\subsubsection{Proof of Statement~\ref{en:essp_implies_ssp}}%

To justify Statement~\ref{en:essp_implies_ssp}, we observe, that the constituents of $U_\tau$ are all \emph{linear} TSs, that is, every constituent $A$ of $U_\tau$ is a finite directed labeled path: $ A= s_0 \edge{e_1}  \dots  \edge{e_t} s_t$, where all states $s_0,\dots, s_t$ are pairwise different.
The next Lemma shows that the $\tau$-ESSP of a linear TS $A$ always implies its $\tau$-SSP.
Consequently, if $\tau\in \{\tau^b_0, \tau^b_1\}$ then the $\tau$-ESSP of $U_\tau$ implies its $\tau$-SSP by the following lemma:

\begin{lemma}\label{lem:essp_implies_ssp}
Let $b\in \mathbb{N}^+$, $\tau \in \{ \tau^b_0, \tau^b_1 \}$ and $ A= s_0 \edge{e_1}  \dots  \edge{e_t} s_t$ be a linear TS having the $\tau$-ESSP.
\begin{enumerate}
\item\label{lem:essp_implies_ssp_infinite_sequence}
If $q_0\edge{e_1}\dots \edge{e_m}q_m \edge{e_1}\dots \edge{e_m}q_{2m}$ is a subpath of $A$ then there has to be a $\tau$-region $(sup,sig)$ of $A$ such that $sup(q_0)\not=sup(q_{m})$.
\item\label{lem:essp_implies_ssp_proof}
If $A$ is finite then it has the $\tau$-SSP.
\end{enumerate}
\end{lemma}
\begin{proof}
(\ref{lem:essp_implies_ssp_infinite_sequence}): 
Assume, for a contradiction, that $q_0\edge{e_1}\dots \edge{e_m}q_m \edge{e_1}\dots \edge{e_m}q_{2m}$ satisfies the equality $sup(q_0)=sup(q_{m})$ for every $\tau$-region $(sup,sig)$ of $A$.
We argue, that this sequence is continued by another transition $q_{2m}\edge{e_1}q_{2m+1}$ and that for every $\tau$-region $(sup, sig)$ of $A$ the equality $sup(q_1)=sup(q_{m+1})$ is satisfied.
This makes $q_1\edge{e_2}\dots \edge{e_1}q_{m +1}\edge{e_2}\dots \edge{e_1}q_{2m+1}$ a new starting point from which, by the same argumentation, we get another sequence that can be continued.
Hence, there is a state $s\in S_A$ and an event $e\in E_A$ such that $s_t\edge{e}s$, which is a contradiction.

For the proof, let $(sup,sig)$ be an arbitrary region. 
By Lemma~\ref{lem:observations}.\ref{lem:sig_summation_along_paths} we obtain 
\begin{align}
\label{eq:support_values_2}
sup(q_m) &= sup(q_0) + \sum_{i=1}^{m} sig^-(e_i) + \sum_{i=1}^{m} sig^+(e_i)  \\
\label{eq:support_values_3}
 sup(q_{2m}) &= sup(q_m) +\sum_{i=1}^{m} sig^-(e_i) + \sum_{i=1}^{m} sig^+(e_i) 
\end{align}
By Equation~\ref{eq:support_values_2} and $sup(q_0)=sup(q_m)$ we get $\sum_{i=1}^{m} sig^-(e_i) + \sum_{i=1}^{m} sig^+(e_i) =0$, implying, by Equation~\ref{eq:support_values_3}, $sup(q_{2m})=sup(q_m)=sup(q_0)$.
Hence, as $(sup,sig)$ was arbitrary, a valid ESSP atom $(e_1, q_{2m})$ contradicts the $\tau$-ESSP of $A$.
Thus, there is a state $q_{2m+1}$ such that $q_{2m}\edge{e_1}q_{2m+1}$.
Moreover, as $\delta_\tau$ is a function, by $q_0\edge{e_1}, q_{m}\edge{e_1}$ and $sup(q_0)=sup(q_m)$ we obtain $sup(q_1)=sup(q_{m+1})$, too.

(\ref{lem:essp_implies_ssp_proof}):
Assume, that there is a sequence $s_{i}\edge{e_{i+1}}\dots \edge{e_j}s_j\edge{e_{j+1}} \dots \edge{e_t}s_t $ in $A$ such that $(s_i, s_j)$ is not solvable.
As every region $(sup,sig)$ of $A$ satisfies $sup(s_i)=sup(s_j)$, by the $\tau$-ESSP, we have that $s_j\not=s_t$ and $e_{i+1}=e_{j+1}$.
Let $k\in \mathbb{N}$ such that $j=i+k$ and let $1 \le \ell \leq k $ be the biggest index such that $e_{i+1}=e_{j+1}, e_{i+2}=e_{j+2},\dots, e_{i+\ell}=e_{j+\ell}$.
If $\ell < k$, then, by the ESSP of $A$, there is a region $(sup, sig)$ separating $e_{i+\ell+1}$ from $s_{j+\ell}$.
By Lemma~\ref{lem:observations}.\ref{lem:sig_summation_along_paths} we have that
\begin{align} 
sup(s_{i+\ell}) & =sup(s_i) + \sum_{m=1}^{\ell} sig^-(e_{i+m}) + \sum_{m=1}^{\ell} sig^+(e_{i+m})\\ 
sup(s_{j+\ell}) & =sup(s_j) + \sum_{m=1}^{\ell} sig^-(e_{i+m}) + \sum_{m=1}^{\ell} sig^+(e_{i+m}) 
\end{align}

which, by $sup(s_i)=sup(s_j)$, implies $sup(s_{i+\ell})=sup(s_{j+\ell})$ contradicting $\neg (sup(e_{j+\ell})\edge{sig(e_{i+\ell+1})})$.
Hence, we have $\ell=k$. 
This implies that we have a sequence $s_i\edge{e_{i+1}}\dots\edge{e_{i+k}}s_j\edge{e_{i+1}}\dots \edge{e_{i+k}}s_{j+k}$ where $sup(s_i)=sup(s_j)$ for all regions of $A$.
By (\ref{lem:essp_implies_ssp_infinite_sequence}) this contradicts the linearity of $A$.
Hence, $(s_i,s_j)$ is $\tau$-solvable and $A$ has the $\tau$-SSP.

\end{proof}

\subsubsection{Proof of Statement~\ref{en:essp_atom_2}}%

Let $\tau\in \{\tau^b_2,\tau^b_3\}$.
We define the set of all initial states of the TSs implemented by $U_\tau$ by $I=\{h_{3,0,0}, t_{j,0}, f_{j,0,0},g_{j,0}\mid 0\leq j\leq m-1\}$.
A lot of separation atoms are solve by Table~\ref{tab:fourth_table} presented at the bottom of this subsection.
However, some atoms need to be discussed individually and or need some additional instructions how their corresponding rows in Table~\ref{tab:fourth_table} are to interpret. 

($k$):
The key region inhibits $k$ in $H_3$ and the region of the first row of the Table~\ref{tab:fourth_table} separates $k$ from the remaining states.

($z$):
Let $i, \ell \in \{0,\dots, m-1\}$,such that $X_\ell=X_{i,2}$ and $i',i''\in \{0,\dots, m-1\}\setminus \{i\}$ be the indices of the translators (clauses) of the second and third occurrence of $X_{i,2}$: $X_{i,2}\in E_{T_{i}}\cap E_{T_{i'}}\cap E_{T_{i''}}$.
Using these definitions, the region presented in the second row of Table~\ref{tab:fourth_table} shows the separation of $z$ in $T_i$ and from $h_{3,0,0}$.
By the arbitrariness of $i$ this proves $z$ to be separable from all states of $T_\tau$.

For the separation of $z$ from the states of $S_{H_3}\setminus \{ h_{3,0,0} \}$ see row three of Table~\ref{tab:fourth_table} and, finally, see the 4th row of Table~\ref{tab:fourth_table} for the separation of $z$ from the remaining states, that is $S_{F_j}\cup S_{G_j},j\in \{0,\dots, m-1\}$. 

($v_\ell$):
Let $\ell\in \{0,\dots, m-1\}$.
The separation of $v_\ell$ in $U_\tau$ affects the variable event $X_\ell$ and we assume $i,i',i''\in \{0,\dots, m-1\}$ to be the respective indices such that $X_\ell\in E_{T_{i}} \cap E_{T_{i'}}\cap E_{T_{i''}}$.
Using these indices, the seventh and eighth row of Table~\ref{tab:fourth_table} prove $v_\ell$ to be separable from all states of $U_\tau$.

($V(\varphi)$):
For the separation of the variable events we proceed as follows:
If $i, i',i''\in \{0,\dots, m-1\}$ and $\alpha \in \{0,1,2\}$ such that $X_{i,\alpha}\in  E_{T_{i}} \cap E_{T_{i'}}\cap E_{T_{i''}}$ then we explicitly present regions for the separation of $X_{i,\alpha}$ at the states in question of $S_{U_\tau}\setminus (S_{T_{i'}} \cup S_{T_{i''})}$.
By the arbitrariness of $i$ and $\alpha$ this proves $X_{i,\alpha}$ to be separable  in $T_{i'}, T_{i''}$, too, and, consequently, in $U_\tau$.

$(X_{i,0})$:
Let $i,i',i'', \ell \in \{0,\dots, m-1\}$ such that $X_{i,0}=X_\ell$ and $X_\ell\in E_{T_{i}} \cap E_{T_{i'}}\cap E_{T_{i''}}$.
The 9th row is dedicated to the separation of $X_\ell$ at the states $f_{\ell,1,0},\dots, f_{\ell,1,b-2}$ and $g_{\ell,0},\dots, g_{\ell, b-1}$ and $t_{i,0},\dots, t_{i,b-1}$.
After that, the 10th row shows $X_\ell$ to be separable  at the remaining states of $S_{U_\tau}\setminus (S_{T_{i'}} \cup S_{T_{i''}})$.

$(X_{i,1})$:
Let $\ell_0,\ell_1,i_0,\dots, i_3\in \{0,\dots, m-1\}$ such that $X_{i,0}=X_{\ell_0}\in E_{T_{i}} \cap E_{T_{i_0}}\cap E_{T_{i_1}}$ and $X_{i,1}=X_{\ell_1}\in E_{T_{i}} \cap E_{T_{i_2}}\cap E_{T_{i_3}}$.
The 11th row show the separation of $X_{\ell_1}$ at $f_{\ell_1,1,0},\dots, f_{\ell, 1, b-2}$ and $g_{\ell_1,0},\dots, g_{\ell_1, b-1}$ .
After that, the 12th row shows $X_{\ell_1}$ to be separable  at the remaining states of $S_{U_\tau}\setminus (S_{T_{i_2}} \cup S_{T_{i_3})}$.
To separate$X_{\ell_1}$ at the states $t_{i,0}, \dots, t_{i,b}$, a lot of cases analyses is necessary as the variable event $X_{i,0}$ comes into play.
Hence, to define an appropriate region, we have to analyze in which constellation the events $X_{i,0}$ and $X_{i,1}$  occur a second and a third time.
Roughly said, the following cases are possible:
\begin{enumerate}
\item
$X_{i,0}$ and $X_{i,1}$ occur a second (third) time together in another translator and $X_{i,1}$ occur \emph{left} from $X_{i,0}$, for example, $i_0=i_2$ and  $t_{i_0, b}\edge{ X_{i,1} }$ and $t_{i_0, b+1} \edge{ X_{i,0} }$, respectively $t_{i_0,b+2}\edge{ X_{i,0} }$, or $t_{i_0,b+1}\edge{X_{i,1}}$ and $t_{i_0,b+2}\edge{X_{i,0}}$.
\item
$X_{i,0}$ and $X_{i,1}$ occur a second (third) time together in another translator and $X_{i,1}$ occur always \emph{right} from $X_{i,0}$ as, for example, it is the case for $T_i$.
\item
$X_{i,0}$ and $X_{i,1}$ occur not again in a common translator.
\end{enumerate}
To define an appropriate region, in the following, we will discuss all possible cases individually.
Firstly, for all cases, the signature $sig$ is defined by $sig(X_{\ell_1}) =(0, b)$, $sig(X_{i,\ell_0})=sig(v_{\ell_1})=1$, $sig(v_{\ell_0})=b$ and $sig(e)=0$ for $e\in E_{U_\tau}\setminus  \{v_{\ell_0}, v_{\ell_1}, X_{\ell_0}, X_{\ell_1} \}$.
Independent from the different cases, we have that $sup(t_{i,0})=sup(f_{\ell_1,0})=b$ and $sup(f_{\ell_0,0})=1$. 
The challenge for the further initials is to achieve, that each further source $s$ of a $X_{\ell_1}$-labeled transition, that is $s\edge{X_{\ell_1}}$, is mapped to $0$.

If $X_{\ell_0} \not\in E_{T_{i_2}}$, respectively $X_{\ell_0} \not\in E_{T_{i_3}}$, then we simply define $sup(t_{i_2,0})=0$, respectively $sup(t_{i_3,0})=0$.

If $X_{\ell_0} \in E_{T_{i_2}}$, respectively $X_{\ell_0} \in E_{T_{i_3}}$, and if $X_{\ell_0}$ occurs left from $X_{\ell_1}$ then the situation is similar to $T_i$ and we define $sup(t_{i_2,0})=b$, respectively $sup(t_{i_3,0})=b$.
Otherwise, if $X_{\ell_0}$ occurs right from $X_{\ell_1}$ then we define $sup(t_{i_2,0})=0$, respectively $sup(t_{i_3,0})=0$.

Finally, the values of $t_{i_0,0}, t_{i_1,0}$ are defined in dependence of one of the former cases:
Actually, if $t_{i_0,0} \in \{t_{i_2,0}, t_{i_3,0}\}$, respectively $t_{i_1,0} \in \{t_{i_2,0}, t_{i_3,0}\}$, then the former case properly defines the support of $t_{i_0,0}$, respectively $t_{i_1,0}$.
Otherwise, we set $sup(t_{i_0,0})=0$, respectively $sup(t_{i_1,0})=0$.

$(X_{i,2})$: 
The separation of $X_{i,2}$ can be perfectly in the same way shown to the separation of $X_{i,1}$.
Hence, for simplicity, we refrain from the explicit presentation of the separating regions.

\noindent\begin{longtable}{p{0.7cm} p{3.8cm}p{3cm} p{3.5cm} p{1.7cm}}
\caption{Inhibiting regions of $U^{\tau}_\varphi$ for $z,u, v_\ell, X_{i,0}, X_{i,1}$}
\label{tab:fourth_table}
\endfirsthead
\endhead
\emph{e} & \emph{initials}& \emph{states} & \emph{sig} & \emph{constituents}\\ \hline
$k$ & $s\in I: s=0$ & remaining states & \raggedright{$k=(0,1)$, $z=v_j=1$} &$U_\tau$\\ \hline
$z$ & \raggedright{$t_{i,0}=t_{i',0}=t_{i'',0}=h_{3,0,0}=b$, $f_{j,0,0}=g_{j,0}=1$, $j\not\in \{i,i',i''\}: t_{j,0}=0$} & $S_{T_i}, h_{3,0,0}$ & \raggedright{$z=(0,b)$,\newline $u=X_\ell=1$, $v_\ell=b$} &$T^0_\tau, H_3, F_\ell, G_\ell$\\
$z$ & \raggedright{$h_{3,0,0}=f_{j,0,0}=g_{j,0}=0$, $t_{j,0}=1$}  & $S_{H_3}\setminus \{h_{3,0,0}\}$  & $z=(0,b), k=1$, $u=2$ & $U_\tau$ \\
$z$ &  \raggedright{ $h_{3,0,0}=b$, $t_{j,0}=0$} & remaining states & \raggedright{$z=(0,b)$, $u=1$} & $H_3, T^0_\tau$\\ \hline
$u$ & \raggedright{$s\in I: s=0$} & $h_{3,0,1},\dots, h_{3,0, b}$ & $u=(0,b)$, $z=2$, $k=1$ & $ U_\tau$\\
$u$ & \raggedright{$h_{3,0,0}=0$, $t_{j,0}=1$} & remaining states & $u=(0,b), z=1$ & $H_3, T^0_\tau$\\ \hline
$v_\ell$ & \raggedright{$s\in I: s=0$} & $f_{\ell,0,1},\dots, f_{\ell,0,b}$ & \raggedright{$v_\ell=(0,b)$, $X_\ell=2$, $k=1$} & $ U_\tau$\\
$v_\ell$ & \raggedright{$f_{\ell,0,0}=0$, $s\in I\setminus \{f_{\ell,0,0} \}: s=1$} & remaining states & $v_\ell=(0,b), X_\ell=1$ & $F_\ell, T_{i}, T_{i'}, T_{i''}$\\ \hline
$\underbrace{X_{i,0}}_{= X_\ell}$ 
& 
\raggedright{$g_{\ell,0}=1$,  $t_{i,0}=t_{i',0}=t_{i'',0}=1$, $s\in I\setminus \{t_{i,0},t_{i',0}, t_{i'',0},g_{\ell,0}\}$: $s=0$} 
& 
$f_{\ell,1,0},\dots, f_{\ell,1, b-2}$, $t_{i,0},\dots, t_{i,b-1}$, $g_{\ell,0},\dots, g_{\ell, b-1}$ 
& 
\raggedright{$X_\ell=(0,b)$, $v_\ell=2$, $k=1$} & $ U_\tau$\\
$\underbrace{X_{i,0}}_{= X_\ell}$ & \raggedright{$f_{\ell,0,0}=b, g_{\ell,0}=0$, $t_{i,0}=t_{i',0}=t_{i'',0}=0$} & remaining states & $X_\ell=(0,b)$, $v_\ell=1$& $ F_\ell, G_\ell$, $T_{i}, T_{i'}, T_{i''}$ \\ \hline
$\underbrace{X_{i,1}}_{= X_{\ell_1}}$ & \raggedright{$g_{\ell_1,0}=1$,$t_{i,0}=t_{i_2,0}=t_{i_3,0}=1$, $s\in I\setminus\{g_{\ell_1,0},t_{i,0},t_{i_2,0},t_{i_3,0} \}: s=0$} & $f_{\ell_1,1,0},\dots, f_{\ell_1,1,b-2}$, $g_{\ell_1,0},\dots, g_{\ell_1, b-1}$ & $X_{\ell_1}=(0,b)$, $v_{\ell_1}=2$, $k=1$ & $ U_\tau$\\

$\underbrace{ X_{i,1} }_{= X_{\ell_1}} $ & \raggedright{ $f_{ \ell_1,0, 0} = b, g_{\ell_1,0}=0 $, $t_{i,0}=t_{i_2,0}=t_{i_3,0}=0$}  & \raggedright{remaining of $S_{U_\tau}\setminus (S_{T_{i_2} }\cup S_{T_{i_3} })$, but not $t_{i,0},\dots, t_{i,b+1} $}& \raggedright{$X_{\ell_1}=(0,b)$, $v_{\ell_1}=1$} & $ F_{\ell_1}, G_{\ell_1}$, $T_{i}, T_{i_2}, T_{i_3}$ \\ 
\end{longtable}

\subsubsection{Proof of Statement~\ref{en:essp_implies_ssp_Z}}%

Actually, the $\tau$-SSP is already proven by the $\tau$-regions presented for the $\tau$-ESSP in Statement~\ref{en:essp_atom_2}:
\begin{enumerate}
\item
If  $s_0\edge{k}\dots \edge{k}s_{n}$ is a sequence in $U_\tau$, then the pairwise separation of $s_{i-1}, s_i$ for $i\in \{1,\dots, n\}$ is done by the region that solves  $(k,h_{3,1, b-1})$.
\item
The separation of $h_{3,0,0}, \dots, h_{3,0,b}$ from $h_{3,1,0},\dots, h_{3,1, b-1}$ is done by the region of the 2nd row of Table~\ref{tab:fourth_table} as well as the separation of $f_{j,0,0},\dots, f_{j,0, b}$ from $f_{j, 1,0},\dots, f_{j, 1,b-1}$ for $j\in \{0,\dots, m-1\}$.
This finishes the state separation in $H_3, F_0,\dots, F_{m-1}$.
\item
If $i,\ell\in \{0,\dots, m-1\}$ such that $X_{i,0}=X_\ell$ then 
\begin{enumerate}
\item
the 10th row of Table~\ref{tab:fourth_table} presents a region, that separates $t_{i,0},\dots, t_{i,b}$ from $t_{i,b+1},\dots, t_{i, 2b+4}$ and $g_{\ell,0},\dots, g_{\ell,b}$ from $g_{\ell,b+1}$,
\item
the region of the 12th row separates $t_{i,b+1}$ from $t_{i,b+2},\dots, t_{i, 2b+4}$ and the corresponding region for $X_{i,2}$ separates $t_{i,b+2}$ from $t_{i,b+3},\dots, t_{i, 2b+4}$ and
\item
the region of the 4th row separates $t_{i,b+3}$ from $t_{i,b+4},\dots, t_{i, 2b+4}$.
\end{enumerate}
By the arbitrariness of $i$, this completes the state separation in $U_\tau$.
\end{enumerate}

\subsubsection{Proof of Statement~\ref{en:ssp_atom}}%

Let $\tau\in \{\tau^b_0,\tau^b_1\}$. 
To show, that the solvability of $(h_{2,0}, h_{2,1})$ in $W$ implies its $\tau$-SSP we explicitly present regions of $W$ that, altogether, solve the SSP atoms induced by $H_2$ and $D_{0,1},\dots, D_{6m-1,1}$.
Furthermore, we observe that the regions of Table~\ref{tab:first_table} and Table~\ref{tab:second_table}, which were originally dedicated to the $\tau^b_1$-separation of $k_{6i},\dots, k_{6i+5},x_i,p_i,X_{i,0}, X_{i,1}, X_{i,2}$, can be fitted into $\tau$-regions of $W$. 
In fact, this is doable simply by replacing the support of the initial state $h_{1,0}$ by the appropriate value for $h_{2,0}$.
Consequently, the union $T_\tau$, has the $\tau$-ESSP with regions of $W$.
Moreover, $T_\tau$ consists only of linear TSs.
Hence, by Lemma~\ref{lem:essp_implies_ssp}, $T_\tau$ has the $\tau$-SSP.
Finally, with the added regions concerning  $H_2, D_{0,0},\dots, D_{6m-1,0}$, the $\tau$-SSP for $W$ is proven:

For a start,  if $(s,s')$ satisfies the condition $s,s'\in S_0=\{h_{2,0}, \dots, h_{2,b}\}$ or $s,s'\in S_1=\{h_{2,b+1}, \dots, h_{2,2b+1}\}$ or $s,s'\in S_2=\{h_{2,2b+2}, \dots, h_{2,3b+2}\}$ or $s,s'\in S_3=\{d_{j,1,0}, d_{j,1,1}\}$ or $s,s'\in S_4=\{d_{j,1,2}, d_{j,1,3}\}$ then $(s,s')$ is already solved by the key region presented for the proof of Lemma~\ref{lem:key_unions_1} and Lemma~\ref{lem:translator_1}.
Moreover, the $\tau$-region $(sup,sig)$ where:
\begin{enumerate}
\item
 $sup(h_{2,0})=sup(d_{j,1,0})=0$ and $sig(o_0)=(0,b)$ separates all states of $S_0$ from all states of $S_1\cup S_2$,

\item
$sup(h_{2,0})=sup(d_{j,1,0})=0$ and $sig(o_1)=(0,b)$ separates all states of $S_1$ from all states of $ S_2$,

\item
$sup(d_{j,1,0})=0$ and $sig(k_j)=(0,b)$ solves $(s,s')$ where $s\in S_3$ and $s'\in S_4$. 
Altogether, this proves $H_2, D_{0,1},\dots, D_{6m-1,1}$ to have the $\tau$-SSP with regions of $W$.
With the former discussion this implies the $\tau$-SSP of $W$.
\end{enumerate}

\end{appendix}

\end{document}